%% file: main.tex
\documentclass[a4paper]{article}       
\usepackage{fullpage}
\usepackage{mathtools}
\usepackage{amsthm}
\usepackage{graphicx}
\usepackage{url}
\usepackage{color}
\usepackage{algorithm2e}
\usepackage{tikz}
\usepackage{pstricks}

\newcommand\floor[1]{\left\lfloor#1\right\rfloor}
\newcommand\msize[1]{{\left|#1\right|}}

\newcommand\cur[1]{\mbox{\rm\sc Cur}(#1)}
\newcommand\dist{{\rm dist}}
\newcommand\calF{{\cal F}}
\newtheorem{theorem}{Theorem}
\newtheorem{lemma}[theorem]{Lemma}

\renewenvironment{proof}{\par\noindent{\sf Proof.}}{\par}

\newtheorem{obs}[theorem]{Observation}
\newtheorem{cor}[theorem]{Corollary}

\begin{document}


\title{Efficient Folding Algorithms for Regular Polyhedra
\thanks{A part of this research was presented at CCCG 2020.
A part of this research is supported by JSPS KAKENHI Grant Number JP17H06287 and 18H04091.}}



\author{
Tonan Kamata$^1$ \and
Akira Kadoguchi$^2$ \and
Takashi Horiyama$^3$ \and
Ryuhei Uehara$^1$
}
\date{1 School of Information Science, Japan Advanced Institute
              of Science and Technology (JAIST), Ishikawa, Japan \\
              \texttt{\{kamata,uehara\}@jaist.ac.jp} \\
2 Intelligent Vision \&{} Image Systems (IVIS), Tokyo, Japan \\
              \texttt{akira.kadoguchi@ivis.co.jp} \\
3 Faculty of Information Science and Technology, Hokkaido
          University, Hokkaido, Japan \\
	  \texttt{horiyama@ist.hokudai.ac.jp}}
\maketitle

\begin{abstract} 
We investigate the folding problem that asks if a polygon $P$ can be folded to a polyhedron $Q$ for given $P$ and $Q$.
Recently, an efficient algorithm for this problem has been developed when $Q$ is a box.
We extend this idea to regular polyhedra, also known as Platonic solids.
The basic idea of our algorithms is common, which is called stamping.
However, the computational complexities of them are different depending on their geometric properties.
We developed four algorithms for the problem as follows.
(1) An algorithm for a regular tetrahedron, which can be extended to a
tetramonohedron.
(2) An algorithm for a regular hexahedron (or a cube), which is much efficient than the previously known one.
(3) An algorithm for a general deltahedron, which contains the cases that $Q$ is a regular octahedron or a regular icosahedron.
(4) An algorithm for a regular dodecahedron.
Combining these algorithms, we can conclude that the folding problem can be solved pseudo-polynomial time 
when $Q$ is a regular polyhedron and other related solid.

\end{abstract}

\section{Introduction}
\label{sec:intro}
\input{intro}

\section{Preliminaries}
\label{sec:pre}
\input{pre}

\section{Regular Tetrahedron and Tetramonohedron}
\label{sec:tetra}
\input{tetra}

\section{Cube}
\label{sec:cube}
\input{cube}

\section{Deltahedra and Regular Dodecahedron}
\label{sec:delta}
\input{delta}

\section{Concluding Remarks}

In this paper, we give a series of design scheme of pseudo-polynomial algorithms for
solving the folding problem for given simple polygon $P$ and convex polyhedron $Q$.
When $Q$ is a regular polyhedra (also known as a Platonic solid) or some variants,
our algorithm runs efficiently. We have some open problems for extension.

The extension to convex polyhedra that consist of finite regular polygons is not so difficult.
Most results in this paper work except the estimation of the number of possible vectors.
If we allow to use non-regular polygons, it is not easy to estimate the number of possible vectors
joining two common vertices of $P$ and $Q$. Thus we may need a different approach.

The extension to concave polyhedra is more challenging.
In our algorithms, we use the convexity of a polyhedron in several places.
For example, the contact graph of the faces of a polyhedron is not necessarily acyclic for a concave polyhedron
(a simple example is given in \cite[Fig.~22.6]{DemaineORourke2007}).
In such a case, the set of cut lines is not connected, and the contact graph is not a tree.
Moreover, a vertex of $Q$ may have curvature $360^\circ$.
Finding a nontrivial set of concave polyhedra that allows us to solve the folding problem efficiently
is another interesting open problem.

\bibliographystyle{spmpsci}      
\bibliography{main}   

\end{document}

%% file: intro.tex
In 1525, the German painter Albrecht D\"{u}rer published his masterwork 
on geometry~\cite{durer1525underweysung}, whose title translates as
``On Teaching Measurement with a Compass and Straightedge for lines, planes, and whole bodies.''  
In the book, he presented each polyhedron by drawing a {\em net},
which is an unfolding of the surface of the polyhedron to a planar layout
without overlapping by cutting along its edges.
To this day, it remains an important open problem 
whether every convex polyhedron has a net by cutting along its edges.
On the other hand, when we allow to cut anywhere, 
any convex polyhedron can be unfolded to a planar layout without overlapping.
There are two known algorithms; one is called \emph{source unfolding}, and 
the other is called \emph{star unfolding} (see \cite{DemaineORourke2007}).

\begin{figure}[h]
\centering
\includegraphics[width=0.3\linewidth]{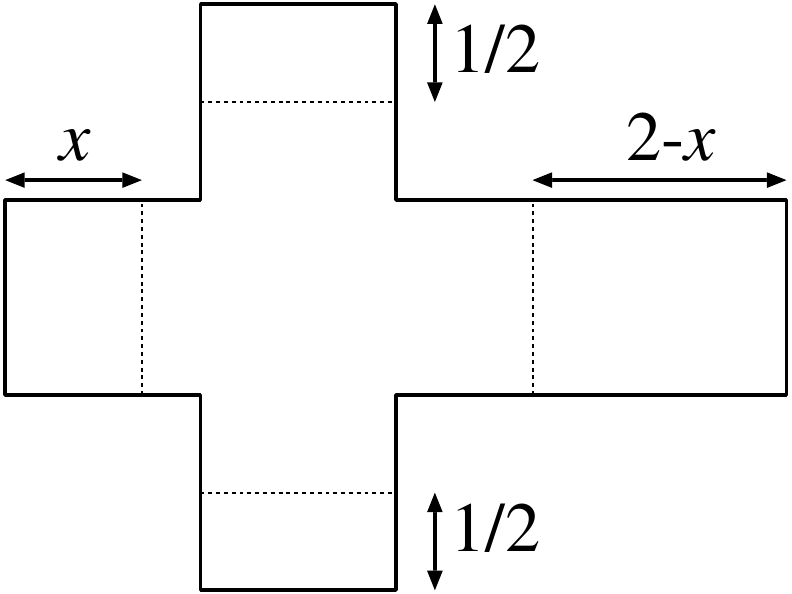}
\caption{A Latin cross made by six unit squares.
 For any real number $x$ with $0<x<1$, folding along dotted lines,
we can obtain a doubly-covered fat cross.}
\label{fig:cross}
\end{figure}

In order to understand unfolding, it is interesting to look at the inverse:
what kind of polyhedra can be folded from a given polygonal sheet of paper?
For example, the Latin cross, which is one of the eleven nets of a cube, 
can be folded to 23 different convex polyhedra by 85 distinct ways of folding \cite{DemaineORourke2007}
and an infinite number of doubly covered concave polygons (\figurename~\ref{fig:cross}).
Comprehensive surveys of folding and unfolding can be found in \cite{DemaineORourke2007}.
In this simple example, we can find that the convexity of 
a polyhedron plays an important role in this context.

We investigate the folding problem when both a polygon $P$ and a polyhedron $Q$ are explicitly given.
That is, for a given polygon $P$ and a polyhedron $Q$, the folding problem asks if $P$ can fold to $Q$ or not. 
This is a natural problem, however, there are a few results so far.
When $Q$ is a regular tetrahedron, we have a mathematical characterization of its net \cite{Akiyama2007};
according to this result, $P$  can fold to $Q$ if and only if $P$ is a kind of tiling.
%
%
%
Recently, the folding problem was investigated for the case that $Q$ is a box.
Some special cases were investigated in \cite{HoriyamaMizunashi2016,XHSU2017},
and the problem for a box $Q$ is solved in \cite{MHU2020} in general.
The running time of the algorithm in \cite{MHU2020} is $O(D^{11}n^2(D^5+\log n))$, where $D$ is the diameter of $P$.
In the algorithm, $Q$ is given as just a ``box'' without size, and the algorithm tried all feasible sizes.
If $Q$ is explicitly given as a ``cube'',
the running time of the algorithm in \cite{MHU2020} is reduced to $O(D^7 n^2(D^5+\log n))$ time.

In this paper, we will show that the folding problem can be solved efficiently for a regular polyhedron,
which is also known as a Platonic solid. While there are five Platonic solids, 
our main result consists of four algorithms. 
That is, we investigate four problems depending on the target solid 
and show pseudo-polynomial time algorithms for them.

The first algorithm solves the folding problem for a regular tetrahedron.
We give a bit stronger algorithm that solves the folding problem for a tetramonohedron, 
which is a tetrahedron that consists of four congruent acute triangles.
As we have already mentioned, a mathematical characterization of a net of a regular tetrahedron is given in
\cite{Akiyama2007}, and its extension to a tetramonohedron is given by
Akiyama and Matsunaga \cite{AkiyamaMatsunaga2020}\footnote{In the literature, a tetramonohedron is called an \emph{isotetrahedron}.}.
However, they only gave mathematical characterizations for the solids, and as far as the authors know,
any algorithmic result for checking these characterizations has never 
been given explicitly\footnote{The authors thank an anonymous referee of \cite{KKHU2020}, who mentioned this point.}.
We give a pseudo-polynomial time algorithm for the folding problem for a tetramonohedron.

The next algorithm solves the folding problem for a regular hexahedron, or a cube.
We give a bit stronger algorithm that solves the folding problem for a box of size $a\times b\times c$, where $a$, $b$, and $c$ are integers.
As mentioned above, a known algorithm for this problem in \cite{MHU2020} runs
in $O(D^7 n^2(D^5+\log n))$ time when $a,b,c$ are explicitly given.
We improve this running time to $O(D^2 n^3)$ time.
(In our context, the running time can be represented as $O(L (L+n) n^2)$ time, where $L$ is the perimeter of $P$.)

The third algorithm solves the folding problem for a regular dodecahedron.

The last algorithm solves the folding problem for a set of special convex deltahedra.
Usually, a deltahedron means a polyhedron whose faces are congruent equilateral triangles.
(More precisely, it is a regular tetrahedron, 
a regular octahedron, a regular icosahedron, 
a triangular bipyramid, 
a pentagonal bipyramid, a snub disphenoid, a triangulated triangular prism, or a gyroelongated square bipyramid.)
This algorithm cannot deal with a special case that all vertices have curvature $180^\circ$.
Therefore, it cannot deal with a regular tetrahedron among the set of deltahedron.
On the other hand, the algorithm can deal with a vertex of curvature $360^\circ$, 
where six equilateral triangles make a flat hexagonal face. 
That is, we allow each face to consist of coplanar regular triangles, or each face can be any convex polyiamond, 
which consists of convex polygon obtained by gluing a collection of equal-sized equilateral triangles
arranged with coincident sides (see \figurename~\ref{fig:HiDelta} for example).
Thus there are an infinite of non-strictly convex deltahedra that our algorithm can deal with.
In summary, if $Q$ has a non-strictly convex deltahedra with at least two vertices of curvature not equal to $180^\circ$,
our algorithm solves the folding problem in pseudo-polynomial time.
This set includes a regular octahedron and a regular icosahedron.
Therefore, using these algorithms, we can efficiently solve the folding problem for Platonic solids (and more).

\begin{figure}[h]
\centering
\includegraphics[width=0.7\linewidth]{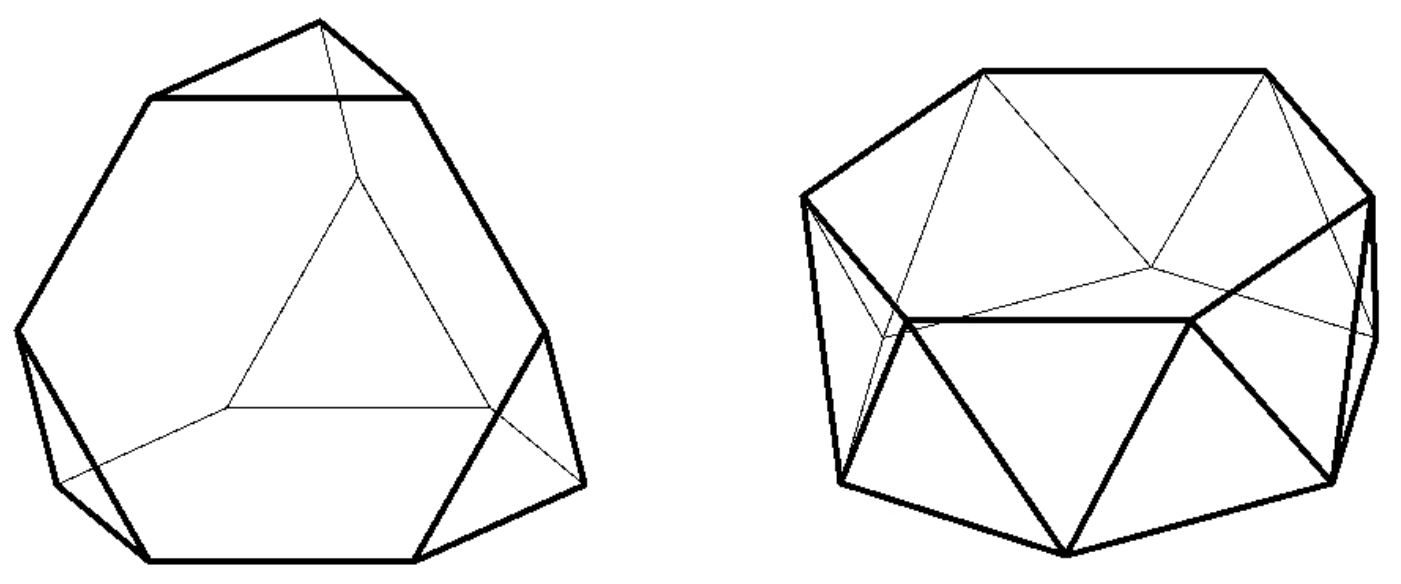}
\caption{Examples of non-strictly convex deltahedra}
\label{fig:HiDelta}
\end{figure}







\begin{table}\centering
\begin{tabular}{|rc|}\hline
 Platonic solids & Running times\\ \hline
 Regular tetrahedron & $O(L(L+n)n^2)$ \\ 
 Regular hexahedron (cube) & $O(L(L+n)n^2)$ \\
 Regular octahedron & $O(L(L+n)n^2)$ \\
 Regular dodecahedron & $O((L+n)^4n^2 )$ \\
 Regular icosahedron & $O(L(L+n)n^2)$ \\\hline
\end{tabular}
\caption{Running times of our algorithms for Platonic solids, where
$n$ is the number of vertices of $P$ and $L$ is the perimeter of $P$.}
\label{tab:alg}
\end{table}

The running times of our algorithms can be summarized in Table \ref{tab:alg}.

%% file: pre.tex
We first state our problem.
For a given polygon $P$ and a polyhedron $Q$, the \emph{folding problem} asks if $P$ can fold to $Q$ or not.
Since we will design an algorithm for each specific $Q$, 
 the input is a polygon $P=(p_0,p_1,\ldots,p_{n-1},p_n=p_0)$.
Let $x(p)$ and $y(p)$ be the $x$-coordinate and $y$-coordinate of a point $p$, respectively. 
For a line segment $\ell$, $\msize{\ell}$ denotes its length (precisely, the length of a geodesic line).
We assume the real RAM model for computation; 
each coordinate is an exact real number, and the running time is measured by the number of 
mathematical operations.

Since the length of the edges of $Q$ can be computed from the area of $P$,
we assume that the length of an edge of $Q$ is 1 when $Q$ is a regular polyhedron 
and $P$ has an area that is consistent with $Q$ without loss of generality.
%
Each algorithm has the information of $Q$ which is represented in 
the standard form in computational geometry (see \cite{BCKO}).
Precisely, for each $i,j,k$, the polyhedron $Q$ consists of vertices $q_i=(x(q_i),y(q_i),z(q_i))$, 
edges $\{q_i,q_j\}$, and faces $f_1,\ldots,f_k$, 
where each $f_i$ is represented by a cycle of vertices 
in counterclockwise-order in relation to the normal vector of the face.

We will construct a one-to-one correspondence between each point on $Q$ with a point on $P$ on the $xy$-plane where $P$ is placed.
To deal with that, we define a \emph{local coordinate} of a point $q_j$ on a face $f_i$ of $Q$ by $q_j=(f_i;x(q_j),y(q_j))$.
To simplify, when $Q$ is a regular polyhedron, we suppose that 
each face $f_i$ of $Q$ has a vertex of $f_i$ with local coordinate $(f_i;0,0)$ and another vertex with local coordinate $(f_i;1,0)$.
That is, each face $f_i$ has the edge $((f_i;0,0),(f_i;1,0))$ of unit length.
We note that each point inside of $f_i$ has a unique local coordinate,
a point on an edge of $f_i$ except its endpoints has two local coordinates $(f_i;x,y)$ and $(f_{i'};x',y')$,
where $f_{i'}$ is the face sharing the edge, and
each vertex of the polygon $f_i$ has $d$ local coordinates,
where $d$ is the number of the faces sharing the vertex on $Q$.
(Namely, the value of $d$ is 3, 4, or 5 in this paper.)

\begin{figure}
\centering
\includegraphics[width=0.6\textwidth]{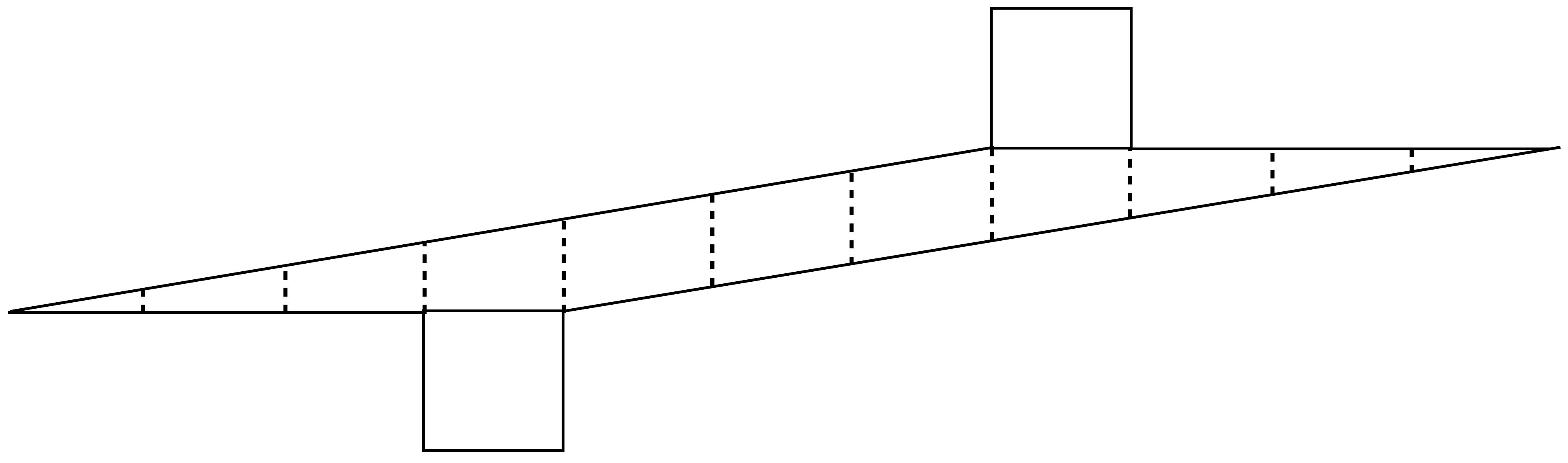}
\caption{An unfolding of a unit cube having long edges}
\label{fig:longl}
\end{figure}

In order to estimate time complexity, for a given polygon $P=(p_0,p_1,\ldots,p_{n-1},p_n=p_0)$,
we define its \emph{diameter} $D_P$ and \emph{perimeter} $L_P$ as follows:
\begin{eqnarray*}
D_P &=& \max_{p,p' \mbox{\scriptsize on }\partial P } \msize{p p'}\\
L_P &=& \sum_{0\le i<n}\msize{p_i p_{i+1}},
\end{eqnarray*}
where $\msize{p q}$ is the distance between two points $p$ and $q$, 
and $\partial P$ is the boundary of the polygon $P$.
When the polygon $P$ is clear, the subscript $P$ is omitted.
We also denote by $\ell_{\max}$ the length of the longest edge of $P$ 
defined by $\max_{0\le i<n}\msize{p_i p_{i+1}}$.
We observe that $\ell_{\max}\le D_P$ and $2D_P<L_P$ for any simple polygon $P$ of positive area.
In some $P$, $\ell_{\max}$ can be quite long compared to the face of $Q$ (see \figurename~\ref{fig:longl}). 
We give the upper bound of the number of faces that an edge of length $\ell_{\max}$ can go through on $Q$.
\begin{lemma}
\label{lem:traverse}
Let $Q$ a polyhedron that consists of convex polyiamonds, integral rectangles, or regular pentagons as its faces.
Then an edge of length $\ell_{\max}$ can go through $M$ faces on $Q$, where $M=O(D_P)$.
\end{lemma}
\begin{proof}
We first focus on $Q$ that consists of faces of convex polyiamonds.
Then the minimum angle of a face is $60^\circ$.
Therefore, an edge of length $\sqrt{3}$ can penetrate at most 4 faces on $P$ since
 the minimum distance between $p$ and $q$ in \figurename~\ref{fig:traverse}(a) is $\sqrt{3}$.
In other words, the number of faces that an edge of length $\ell_{\max}$ 
has intersections is at most $\ell_{\max}\times \frac{4}{\sqrt{3}}$.

Next we turn to the integral rectangular faces.
Let $a$ be the length of the shortest edge of the rectangles.
Then $a$ is an integer by assumption.
Thus, by a similar argument with \figurename~\ref{fig:traverse}(b),
we can show that the number of faces that an edge of length $\ell_{\max}$ 
has intersections is at most $\ell_{\max}\times \frac{2}{a}$.
When a face is a regular pentagon, we can use similar arguments.

By $\ell_{\max}\le D_P$ and $2D_P<L_P$, we have $M=O(D_P)$ and $M=O(L_P)$ in both cases.
\qed\end{proof}

We note that when the minimum angle in the faces and the minimum distance between two non-adjacent edges
of $Q$ are bounded by some constants, we have the same claim with some geometric
parameters given by the bounds.

\begin{figure}
\centering
\input{./figure/traverse.tex}
\caption{(a) 
An edge of length $\sqrt{3}$ can penetrate at most 4 regular triangles, and 
(b) an edge of length $a$ can penetrate at most 2 rectangles.}
\label{fig:traverse}
\end{figure}
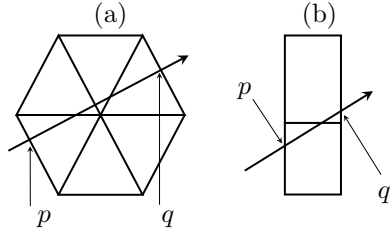

Let $Q$ be a polyhedron and $q$ be a point on $Q$.
The \emph{curvature} $\cur{q}$ at $q$ is the angle defined by the value $360^\circ-A$
and \emph{co-curvature} at $q$ is the angle $A$,
where $A$ is the total angle of the angles on faces of $Q$ adjacent to $q$.
That is, the curvature of $q$ on a convex polyhedron $Q$ is 
less than $360^\circ$ if and only if $q$ is a vertex of $Q$.

We will use the following Gauss-Bonnet Theorem:
\begin{theorem}[{Gauss-Bonnet Theorem}]
\label{th:gauss}
The total sum of the curvature of all vertices of a convex polyhedron is $720^\circ$.
\end{theorem}
See \cite[Sec.~21.3]{DemaineORourke2007} for details.

\begin{figure}
\centering
\includegraphics[width=0.7\linewidth]{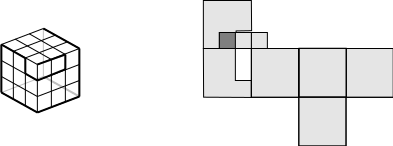}
\caption{A simple development of a cube overlaps.}
\label{fig:cube-overlap}
\end{figure}

Let $Q$ be a convex polyhedron. 
A \emph{development} of a convex polyhedron $Q$ results when we cut $Q$ along a set of polygonal lines
and unfold on a plane. If the development produces a connected simple non-overlapping polygon $P$, 
we call $P$ a \emph{net} of $Q$.
We note that the property of non-overlapping is counterintuitive; we may have it 
even for a simple development of a cube (\figurename~\ref{fig:cube-overlap}).
We assume that any cut ends at a point with curvature less than $360^\circ$.
Otherwise, since $Q$ is convex, it makes a redundant cut on $P$, 
which can be eliminated (the proof can be found in \cite[Theorem 3]{MU2008}).
Let $T$ be the set of cut lines on $Q$ to obtain a net $P$.
Then the following is well known as a basic property of unfolding 
(see \cite[Sec.~22.1.3]{DemaineORourke2007} for details):
\begin{theorem}
\label{th:spanning}
$T$ forms a spanning tree of the vertices of $Q$.
\end{theorem}

\subsection{Properties of Unfolding}

A \emph{tetramonohedron} is a tetrahedron that consists of four congruent acute triangles.
We here note that a vertex of curvature $180^\circ$ on a polyhedron can be on an edge on its unfolding.
Therefore, we need some ideas to find vertices of curvature
$180^\circ$ of $Q$ on a polygon $P$ which may be a net of $Q$.
In this context, a tetramonohedron has a special property:

\begin{lemma}
\label{lem:tetra}
Let $Q$ be a convex polyhedron. 
Then $Q$ is a tetramonohedron if and only if the curvature of every vertex is $180^\circ$.
\end{lemma}
\begin{proof}
If $Q$ is a tetramonohedron, by its symmetric property, 
each vertex $q$ consists of three distinct angles of a congruent triangle.
(Otherwise, we cannot glue corresponding edges to form $Q$ from four acute congruent triangles.)
Thus the curvature at $q$ is $180^\circ$. 

In order to show the opposite, 
we assume that every vertex of a convex polyhedron $Q$ has curvature $180^\circ$.
Then, by the Gauss-Bonnet Theorem, $Q$ has four vertices. 
Let $q_0,q_1,q_2,q_3$ be these four vertices.
We cut along three shortest straight lines $q_0q_1$, $q_0q_2$, and $q_0q_3$ on $Q$, respectively.
(We note that it is called star unfolding in literature \cite[Sec.~24.3]{DemaineORourke2007}.)
Since the curvature at any point on $Q$ except $q_0,q_1,q_2,q_3$ is $360^\circ$,
we can take three non-crossing straight lines from $q_0$ to $q_1$, $q_2$, and $q_3$ on $Q$.
By developing $Q$ from $q_0$ along these three cut lines, 
we obtain a polygon $P=(q_0,q_1,q_0',q_2,q_0'',q_3,q_0)$.
Then, by assumption, curvatures at $q_1$, $q_2$, $q_3$ are all $180^\circ$.
That is, $P$ is a triangle with three vertices $q_0$, $q_0'$, and $q_0''$.
Moreover, each edge of the triangle consists of two cut lines which form an edge on $Q$.
Therefore, $q_1$, $q_2$, and $q_3$ are all the middle points of three edges 
$q_0q_0'$, $q_0'q_0''$, and $q_0''q_0$ of the triangle $P=(q_0,q_0',q_0'',q_0)$, respectively.
Thus all four triangles $q_0q_1q_3$, $q_1q_0'q_2$, $q_3q_2q_0''$, and $q_2q_3q_1$ are congruent.
If they are not acute, $P$ cannot fold to any solid by these crease lines (see, e.g., \cite{Uehara2020}).
Therefore, we can conclude that $Q$ is a tetramonohedron.
\qed\end{proof}

Combining Lemma \ref{lem:tetra} and Theorem \ref{th:gauss}, we have
the following corollary.
\begin{cor}
\label{cor:nontetra}
Let $Q$ be a convex polyhedron. When $Q$ is not a tetramonohedron,
$Q$ has at least two vertices of curvature not equal to $180^\circ$.
\end{cor}
\begin{proof}
Let $q_0,\ldots,q_k$ be the vertices of $Q$. Since $Q$ is a convex polyhedron, $k\ge 3$.
When $k=3$, the only possible solid is a doubly covered triangle. Then $Q$ satisfies the claim.
If $k>4$, by Theorems \ref{th:gauss}, at least two vertices have curvature not equal to $180^\circ$. 
Thus we focus on the case $k=4$. 
By Lemma \ref{lem:tetra}, since $Q$ is not a tetramonohedron, four vertices cannot have curvatures equal to $180^\circ$.
By the Gauss-Bonnet Theorem, it is impossible that three vertices have curvature $180^\circ$ except one.
Thus $Q$ has at least two vertices $q$ and $q'$ of curvature not equal to $180^\circ$.
\qed\end{proof}

The following simple lemma is basic but important:
\begin{lemma}[{\cite{DemaineORourke2007}}]
\label{lem:boundary}
Let $Q$ be a convex polyhedron, and $P$ be a net of $Q$. 
Then all vertices of $Q$ appear on $\partial P$.
\end{lemma}
\begin{proof}
Since $Q$ is convex, each vertex of $Q$ has positive curvature.
Hence it cannot correspond to an interior point of $P$. Thus we have the lemma.
\qed\end{proof}

The following lemma is useful for dealing with polyhedra except tetramonohedra.
\begin{lemma}
\label{lem:non180}
Let $Q$ be a convex polyhedron that is not a tetramonohedron, and $P$ be a net of $Q$.
Then $P$ has at least two vertices of angle not equal to $180^\circ$ that correspond to distinct vertices of $Q$.
Moreover, if $Q$ has no vertex of curvature $180^\circ$, $P$ has at least two vertices such that 
each of the vertices on $P$ is glued to the corresponding vertex of $Q$ without extra angle from $P$.
\end{lemma}

\begin{proof}
By Corollary \ref{cor:nontetra}, $Q$ has at least two vertices $q$ and $q'$ of curvature not equal to $180^\circ$.
Then, by Lemma \ref{lem:boundary}, $q$ and $q'$ correspond to distinct vertices of $P$.
Now consider the set $S_q$ of vertices of $P$ that are glued together to form $q$.
Then, since the curvature of $q$ is not equal to $180^\circ$ and less than $360^\circ$,
at least one of the elements in $S_q$ has an angle not equal to $180^\circ$.
Thus $q$ produces at least one vertex on $\partial P$ of angle not equal to $180^\circ$.
We have another vertex on $\partial P$ produced by $q'$ by the same argument.

Now we assume that $Q$ has no vertex of curvature $180^\circ$.
By Theorem \ref{th:spanning}, the set of cut lines forms a spanning tree of the vertices of $Q$.
Then each leaf of the tree corresponds to a vertex of $Q$, and this vertex forms a vertex of $P$ since 
the curvature is not $180^\circ$. Since any tree has at least two leaves, we have the lemma.
\qed
\end{proof}

\subsection{Outline of Algorithms}

Lemma \ref{lem:non180} plays an important role in this paper except Section \ref{sec:tetra}.
For each (regular) polyhedron except tetramonohedron, we can use at least two vertices $p_i$ and $p_{i'}$ of $P$ such that
these vertices can fold to two vertices $q_j$ and $q_{j'}$ without any extra vertices of $P$, respectively.
That is, intuitively, we can start gluing at the vertices $p_i$ and $p_{i'}$ of $P$ to 
fold $q_j$ and $q_{j'}$ of $Q$ and each gluing can be done by zipping from the vertex locally.
In the context of stamping, the outline of the algorithms can be described as follows.

\begin{algorithm}[h]
 \caption{Common outline of our folding algorithm (except tetramonohedron)}
 \label{alg:common}
 \SetKwInOut{Input}{Input}
 \SetKwInOut{Output}{Output}
 \Input{A polygon $P=(p_0,p_1,\ldots,p_{n-1},p_0)$ and a convex polyhedron $Q$}
 \Output{All ways of folding $P$ to $Q$ (if one exists)}
  Let $\{q_0,\ldots,q_{m-1}\}$ be the set of vertices of $Q$\;
 \ForEach{pair of two vertices $\{p_{i},p_{i'}\}$ of $P$}{
   \ForEach{vertex $q_j$ of $Q$}{
       Check if $Q$ is reachable from $p_i$ to $p_{i'}$ on $P$ by stamping $Q$ 
         so that $p_i$ and $p_{i'}$ correspond to $q_j$ and $q_{j'}$, respectively,
         for some $q_{j'}$ with $j\neq j'$ on $Q$\;
       Check if $P$ is a net of $Q$ by folding and gluing $P$ based on the partition of $P$ by stamping of $Q$\;
   }
 }
\end{algorithm}
The stamping operation and the check of the gluing are key aspects in the algorithm.
We describe the details of each of them.

\subsection{Stamping}
\label{sec:stamping}

When we fold $P$ to $Q$, we have two viewpoints depending on which we focus on.
One viewpoint is that we fix $Q$ in 3D space and ``wrap'' $Q$ by folding $P$ on it. 
The other is that we put $P$ on a plane and ``roll'' $Q$ on it.
The second idea is called \emph{stamping} and used to give 
a characterization of a net of a regular tetrahedron in \cite{Akiyama2007}.
In \cite{Akiyama2007}, Akiyama rolls a regular tetrahedron on a plane as a \emph{stamper} 
and obtains a tiling by the stamping. 
The key property of the stamping in \cite{Akiyama2007} is that
a regular tetrahedron has the same direction and position
when it returns to the original position, no matter what the route is.
Therefore, the cut lines of any net on the surface of a regular tetrahedron 
tile plane, or the net tiles plane based on $180^\circ$ rotation symmetry.

In \cite{MHU2020}, the authors extend the idea of stamping to the box folding problem.
When the stamper is a box, the positions and directions are changed by the rolling, 
however, all faces on the plane are always orthogonal,
 or each edge of a face is always parallel to one of the $xy$-axes
 defined by the first face on the plane.

In this paper, we extend the idea to the general polyhedra.
We first give the details of the stamping.
It consists of two steps.

In the first step, we put $Q$ on $P$. 
Precisely, we put the vertex $q_0=(0,0,0)$ of $Q$ on the vertex $p_0=(0,0)$ of $P$ on the plane.
We adjust their relative angle suitably, which will be discussed later for each case.
(In fact, this ``first position'' is the first crucial point.)
We assume that the face $f_0$ of $Q$ is put on the plane.
(The local coordinate of $q_0$ on $f_0$ is $(f_0;0,0)$.)
We consider the intersection of $f_0$ and $P$. In general, the intersection is not connected;
it is a set of simple polygons. Among them, we define a simple polygon $F_0$ by a simple polygon that contains $p_0=q_0$.
When two or more simple polygons share the point $p_0$ in the intersection, we choose any one of them.
This $F_0$ is said to be the \emph{initial region} of $P$.
From $F_0$, we obtain a sequence of regions $F_i$ with $i>0$ by stamping.

Now we consider the intersection of two boundaries $\partial F_0 \cap \partial P$.
We classify each edge $e={q_i,q_{i'}}$ of $f_0$ to two groups.
We first note that by Lemma \ref{lem:boundary}, if $q_i$ and $q_{i'}$ should be 
outside of $P$ or on $\partial P$. Otherwise, we cannot fold $Q$ by $P$ from this first position.
If $e \cap (\partial F_0 \cap \partial P)$ contains no point in $P\setminus \partial P$, 
we say that $e$ is \emph{closed}. Intuitively, the face $f_0$ contains no more points in $P$ beyond the
closed edge $e$, and hence we do not need to roll $Q$ on the edge $e$ to extend the net.
Otherwise, we say that $e$ is \emph{open}. 
For each open edge $e$ on $F_0$, the net should be extended beyond $e$, and
hence we roll $Q$ along $e$ from $F_0$ to obtain the next face.
We note that $e \cap (\partial F_0 \cap \partial P)$ can consist of two or more line segments.
Let $e_1$ and $e_2$ be two distinct line segments in $e \cap (\partial F_0 \cap \partial P)$.
Then we obtain two different faces $F_1$ and $F_2$ sharing $e_1$ and $e_2$ with $F_0$, respectively.
After rolling $Q$ along $e_0$, we have the following claim: $F_1\cap F_2=\emptyset$.
Otherwise, $P$ should contain a hole surrounded by $F_0,F_1,F_2$,
which contradicts the assumption that $P$ is a simple polygon. 
Therefore, for each line segment in $e \cap (\partial F_0 \cap \partial P)$ for each edge $e$ of $f_0$,
we obtain a new set of faces $F_1,F_2,\ldots$ surrounding $F_0$.

For the notational simplicity, we number the faces in $f_i\cap P$ in this manner.
That is, from $F_0$, we give the index of each intersection $f_i\cap P$ as $F_1,F_2,\ldots,$
in the breadth-first search (BFS) manner.
We here define the distance $\dist(F_0,F_i)$ for each $i>0$ by 
the number of rollings made to reach to $F_i$ in the shortest way.
That is, $\dist(F_0,F_0)=0$ and $\dist(F_0,F_i)=\min_{i'}\dist(F_0,F_{i'})+1$,
where $F_{i'}$ is a region that sharing an edge $e$ with $F_i$.
Intuitively, $F_i$ is obtained by rolling $Q$ from $F_{i'}$ along the edge $e$.
By the property of the BFS, we have $i'<i$ for them.

Aside the numbering, the implementation of the stamping is simpler when we apply the 
depth-first search (DFS) manner rather than the BFS manner.
That is, starting from $F_0$, as the second step of the stamping, 
we repeat rolling $Q$ as long as we can along open edges,
and back to the last unvisited open edge when we get stuck.
Therefore, hereafter, we will describe the algorithm in DFS style,
while the indices of regions $F$ are numbered in BFS manner,
and these indices are supposed to be precomputed in the same time bound of the algorithm.

In the stamping, the following observation is important:
\begin{obs}
\label{obs:out}
Let $P$ be a net of $Q$. When we perform the stamping of $Q$ on $P$ starting from the initial region $F_0$, 
and it is the right place and the right direction to fold $Q$, no vertex of $Q$ is placed inside of $P$.
\end{obs}
\begin{proof}
If a vertex $q$ is put inside of $P$, 
$q$ cannot be folded from $P$ since it has curvature less than $360^\circ$. 
It contradicts to the assumption that $P$ is a net of a convex polyhedron $Q$.
\qed
\end{proof}

Let $\calF=\{F_0,F_1,\ldots\}$ be the set of regions starting from the initial region $F_0$.
Now we define a graph $T(P,Q,F_0)=(\calF,E)$ be a contact graph of $\calF$. 
That is, $\{F_i,F_j\}\in E$ if and only if they share an open edge $e$ such that
$F_j$ is stamped by rolling $Q$ from $F_i$ along the edge $e$ or vice versa.
Then the following lemma states that the graph $T(P,Q,F_0)$ is a well-defined tree
if $P$ is a net of $Q$ and $F_0$ is the right initial region.
\begin{lemma}
\label{lem:tree}
Let $P$ be a net of $Q$. 
We assume that we perform the stamping of $Q$ on $P$ starting from an initial region $F_0$,
and it is the right place and the right direction to fold $Q$.
Then the contact graph $T(P,Q,F_0)$ is a tree.
\end{lemma}
\begin{proof}
By Observation \ref{obs:out}, each region $F_i$ with $i>0$ is defined by some rolling.
Thus the contact graph is trivially connected.
To show that $T(P,Q,F_0)$ is acyclic by contradictions,
we consider two cases that may produce a cycle in the contact graph.

The first case is that a region $F_i$ from $F_{i'}$ by rolling and the other region $F_j$ from $F_{j'}$
have a collision; that is, $F_i\cap F_j\neq\emptyset$ in $P$.
In this case, we have $i'<i$ and $j'<j$. We assume that this pair $(i',j')$ is the minimum in
the lexicographical ordering. (Note that $i'=j'$ is possible.)
Without loss of generality, we assume that $i<j$.
The contact graph induced by $\{F_0,F_1,\ldots,F_{j-1}\}$ is a tree that contains $F_{i'}$, $F_{j'}$, and $F_{i}$, 
and the one induced by $\{F_0,F_1,\ldots,F_{j}\}$ has a cycle by adding $F_j$, which overlaps with $F_i$.
In this case, if $F_i$, $F_j$, and other regions surround a space,
it contradicts that $P$ is a simple polygon with no hole.
On the other hand, when $F_i$ and $F_j$, and other regions share a point
(for example, $F_{i'}=F_{j'}$, $F_i$ and $F_j$ share a vertex of $Q$),
it means that a vertex of $Q$ is put inside of $P$, which contradicts Lemma \ref{lem:boundary}.

The second case is that a region $F_k$ is obtained from $F_{k'}$ and $F_{k''}$ with $k>k'>k''$.
That is, $F_k$ and $F_{k'}$ share an edge $e'$, $F_k$ and $F_{k''}$ share an edge $e''$, and 
$e'\neq e''$ with $k'\neq k''$. In this case, we can use a similar argument in the first case
by considering as $F_{k'}=F_{i'}$, $F_{k''}=F_{j'}$, and $F_{k}=F_{i}=F_{j}$.
That is, we have a hole in $P$ or a vertex of $Q$ is inside of $P$.

By arguments above, we can conclude that $T(P,Q,F_0)$ forms a tree 
if $P$ is a net of $Q$ and $F_0$ is the right initial region.
\qed
\end{proof}

Hereafter, we assume that $T(P,Q,F_0)$ is a tree (otherwise, the algorithm rejects the input) rooted at $F_0$.
Therefore, we can uniquely define the parent $F_i$ for each region $F_j$ in $\calF$ (except $F_0$),
and it is easy to observe that $F_j$ is obtained by rolling $Q$ from $F_i$ along 
an open edge $e$ shared by $F_i$ and $F_j$.
Thus we obtain the following theorem:
\begin{theorem}
\label{th:stamping}
Let $P$ be a net of $Q$. 
We assume that we perform the stamping of $Q$ on $P$ starting from an initial region $F_0$,
and it is the right place and the right direction to fold $Q$.
Let $\calF$ be the set of regions obtained by the stamping of $Q$.
Then the contact graph $T(P,Q,F_0)=(\calF,E)$ is a tree, 
and it can be computed in $O(\msize{\calF}n)$ time in general.
When $Q$ is a regular polyhedron, the running time is reduced to $O(\msize{\calF}+n)$ time.
\end{theorem}
\begin{proof}
From $F_0$, we traverse and construct the tree $T(P,Q,F_0)$ and obtain the sequence of 
regions $\calF=(F_0,F_1,F_2,\ldots)$ incrementally (in the DFS or BFS manner).
At a region $F_i$, we roll $Q$ along an open edge and 
put a face $f$ of $Q$ on $P$, then we obtain the next region $F_{i'}$.
In order to compute $F_{i'}$ that is the intersection of a face $f$ of $Q$ and $P$, 
it takes $O(n)$ time to traverse the vertices in $P$.
The construction part of the tree $T(P,Q,F_0)$ requires time proportional to $\msize{\calF}$.
Therefore, the running time of $T(P,Q,F_0)$ is $O(\msize{\calF}n)$ time in general.

When $Q$ is a regular polyhedron and each face $f$ is a fixed regular polygon,
by performing DFS incrementally on $T(P,Q,F_0)$,
we can construct it so that each vertex of $P$ is touched constant time.
Then the running time can be reduced to $O(\msize{\calF}+n)$ time.
\qed
\end{proof}

The size $\msize{\calF}$ will be discussed in each case, which depends on the polyhedron $Q$.
During the stamping, each region $F_i$ is stamped by a face of $Q$.
At that time, each vertex $v$ of the polygonal region $F_i$ can also obtain
the corresponding local coordinate $(f_i;x(v),y(v))$ on $Q$.
Therefore, in the next phase, we can use the local coordinate of each vertex of a region $F_i$ in a constant time.

By the stamping of $Q$, the boundary $\partial P$ of the polygon $P=(p_0,p_1,\ldots,p_{n-1},p_n=p_0)$ is 
also partitioned into line segments which are intersections of faces of $Q$.
Precisely, each point $p$ on $\partial P$ is (1) an original vertex $p_i$ for some $i$ or (2) an internal point of an edge $p_{i}p_{i+1}$
from the viewpoint of $\partial P$, and
(a) a vertex $q_j$ for some face of $Q$,
(b) an internal point of an edge of $Q$, or
(c) an internal point of a face of $Q$.
We consider $\partial P$ as a new polygon $P'=(p'_0,p'_1,\ldots,p'_{n'-1},p'_{n'}=p'_0)$,
where $p'_{i}$ is a vertex in cases (1), (b), or (c).
Intuitively, when $P'$ is a net of $Q$, each vertex of $P'$ is either a crease point of $P$ which should be folded to form $Q$
or a corner of $P$ of angle not equal to $180^\circ$.
As shown in Lemma \ref{lem:non180}, if $Q$ is not a tetramonohedron,
we have at least two vertices on $P'$ that directly (i.e., without extra angle) correspond to the vertices of $Q$.
During the stamping, we can obtain $P'$ and the set of such vertices without extra computation time within a constant factor.
Let $S$ be the set of the vertices of $P'$, and call them \emph{gluing points} of $P$.

\subsection{Check of gluing}
\label{sec:glue-check}


In this phase, we check if we can fold $Q$ from $P$ by folding along the crease lines given in the first stamping phase.
It may seem to be the first phase is enough. However, we have not yet checked 
if some regions in $\calF$ cause overlap on a face of $Q$. 
In other words, we have to check each face of $Q$ is made by a certain set of 
regions of $P$ by gluing without overlap or hole.

This can be done by checking the new polygon $P'=(p'_0,p'_1,\ldots,p'_{n'-1},p'_{n'}=p'_0)$, which is constructed by the stamping.
By Theorem \ref{th:spanning}, when $P$ is an unfolding of $Q$, the set $T$ of cut lines on $Q$ forms a spanning tree.
Therefore, each line segment $\ell$ in $T$ appears twice as $\ell'$ and $\ell''$ on $P'$
(except its endpoints) and the pairs of $\ell'$ and $\ell''$ form a nest structure on $P'$.
This nest structure is discussed in \cite[Sec.~25.2.1]{DemaineORourke2007} for a dynamic programming formulation of
checking of all ways of edge-to-edge gluings.
Their algorithm checks all ways of edge-to-edge gluing based on dynamic programming.
On the other hand, in our glue checking, we can find each pair of line segments in
the polygon $P'=(p'_0,p'_1,\ldots,p'_{n'-1},p'_{n'}=p'_0)$ by using their local coordinates.
Precisely, for the given polygon $P'$, each point $p'_i$ has its local coordinate $(f;x,y)$
(or it has a list of local coordinates if it corresponds to a vertex of $Q$ or a point on an edge of $Q$).

\begin{theorem}
\label{th:glue}
Let $P'=(p'_0,p'_1,\ldots,p'_{n'-1},p'_{n'}=p'_0)$ be the polygon given by the stamping.
Then the gluing check of $P'$ that asks if $P'$ can fold to $Q$ can be done in $O(n')$ time.
\end{theorem}
\begin{proof}
We maintain $P'$ by a doubly linked list; each item corresponds to $p'_i$
which stores $p'_{i-1}$, $p'_{i+1}$, $\msize{p'_i p'_{i-1}}$, $\msize{p'_i p'_{i+1}}$, and $\angle p'_{i-1} p'_{i} p'_{i+1}$.
We also inherit the set $S$ of gluing points in $P'$ from the stamping step
 such that each gluing point $p$ has the angle $A$ at the point on $\partial P'$
      that corresponds to a vertex $q$ of $Q$ with $\cur{q}=360^\circ-A$.
Intuitively, each gluing point should be zipped up from the point to fold $Q$ from $P'$.
In other words, each gluing point corresponds to a leaf of the spanning tree $T$ of the vertices of $Q$ to cut and unfold to $P$.

Therefore, our gluing starts from any gluing point.
We pick up arbitrary gluing point $p'_i$ in $S$.
Then glue two line segments $p'_{i-1}p'_i$ and $p'_ip'_{i+1}$ from $p'_i$.
Now we have two cases.
The first case is $\msize{p'_i p'_{i-1}} \neq \msize{p'_i p'_{i+1}}$.
Without loss of generality, we assume that $\msize{p'_i p'_{i-1}} < \msize{p'_i p'_{i+1}}$.
In this case, we glue up to $p'_{i-1}$ from $p'_i$.
That is, we remove $p'_i$ from $S$ and $P'$, replace $p'_i$ in $p'_{i-1}$ by $p'_{i+1}$
with length $\msize{p'_{i} p'_{i+1}}-\msize{p'_{i-1} p'_{i}}$,
and replace $p'_i$ in $p'_{i+1}$ by $p'_{i-1}$ with length $\msize{p'_{i} p'_{i+1}}-\msize{p'_{i-1} p'_{i}}$.
The angle $\angle p'_{i-2} p'_{i-1} p'_{i}$ in $p'_{i-1}$ is replaced by
$\angle p'_{i-2} p'_{i-1} p'_{i+1}=\angle p'_{i-2} p'_{i-1} p'_{i}+180^\circ$.
The second case is $\msize{p'_i p'_{i-1}} = \msize{p'_i p'_{i+1}}$.
In this case, we glue $p'_{i-1} p'_i$ and $p'_{i} p'_{i+1}$ completely,
remove $p'_i$ from $S$ and $P'$, and we merge $p'_{i-1}$ and $p'_{i+1}$ to a new vertex $p'$ on $P'$ so that
$\msize{p'_{i-2} p'}=\msize{p'_{i-2} p'_{i-1}}$, $\msize{p',p'_{i+2}}=\msize{p'_{i+1} p'_{i+2}}$, 
and the angle at $p'$ is given by $\angle p'_{i-2} p' p'_{i+2}=\angle p'_{i} p'_{i-1} p'_{i-2}+ \angle p'_{i} p'_{i+1} p'_{i+2}$.
If the angle at $p'$ is equal to the co-curvature at the corresponding point on $Q$
(which can be done by checking the corresponding local coordinate in a constant time),
$p'$ is put into $S$ as a new gluing point. Otherwise, $p'$ is just a new vertex on $P'$.

From the viewpoint of the spanning tree $T$, each gluing step from a gluing point in $S$ corresponds to
removal of one leaf from $T$. Therefore, it is not difficult to see that the algorithm correctly works.
Each gluing process decreases at least one edge from $P'$ in a constant time.
Therefore, the running time of the gluing check is $O(n')$ time.
\qed
\end{proof}

%% file: figure/traverse.tex
\ifx\du\undefined
  \newlength{\du}
\fi
\setlength{\du}{15\unitlength}
\begin{tikzpicture}
\pgftransformxscale{1.000000}
\pgftransformyscale{-1.000000}
\definecolor{dialinecolor}{rgb}{0.000000, 0.000000, 0.000000}
\pgfsetstrokecolor{dialinecolor}
\definecolor{dialinecolor}{rgb}{1.000000, 1.000000, 1.000000}
\pgfsetfillcolor{dialinecolor}
\pgfsetlinewidth{0.050000\du}
\pgfsetdash{}{0pt}
\pgfsetdash{}{0pt}
\pgfsetbuttcap
\pgfsetmiterjoin
\pgfsetlinewidth{0.050000\du}
\pgfsetbuttcap
\pgfsetmiterjoin
\pgfsetdash{}{0pt}
\definecolor{dialinecolor}{rgb}{1.000000, 1.000000, 1.000000}
\pgfsetfillcolor{dialinecolor}
\fill (4.550000\du,11.000000\du)--(5.600000\du,13.000000\du)--(3.500000\du,13.000000\du)--cycle;
\definecolor{dialinecolor}{rgb}{0.000000, 0.000000, 0.000000}
\pgfsetstrokecolor{dialinecolor}
\draw (4.550000\du,11.000000\du)--(5.600000\du,13.000000\du)--(3.500000\du,13.000000\du)--cycle;
\pgfsetlinewidth{0.050000\du}
\pgfsetdash{}{0pt}
\pgfsetdash{}{0pt}
\pgfsetbuttcap
\pgfsetmiterjoin
\pgfsetlinewidth{0.050000\du}
\pgfsetbuttcap
\pgfsetmiterjoin
\pgfsetdash{}{0pt}
\definecolor{dialinecolor}{rgb}{1.000000, 1.000000, 1.000000}
\pgfsetfillcolor{dialinecolor}
\fill (6.650000\du,11.000000\du)--(7.700000\du,13.000000\du)--(5.600000\du,13.000000\du)--cycle;
\definecolor{dialinecolor}{rgb}{0.000000, 0.000000, 0.000000}
\pgfsetstrokecolor{dialinecolor}
\draw (6.650000\du,11.000000\du)--(7.700000\du,13.000000\du)--(5.600000\du,13.000000\du)--cycle;
\pgfsetlinewidth{0.050000\du}
\pgfsetdash{}{0pt}
\pgfsetdash{}{0pt}
\pgfsetbuttcap
\pgfsetmiterjoin
\pgfsetlinewidth{0.050000\du}
\pgfsetbuttcap
\pgfsetmiterjoin
\pgfsetdash{}{0pt}
\definecolor{dialinecolor}{rgb}{1.000000, 1.000000, 1.000000}
\pgfsetfillcolor{dialinecolor}
\fill (5.600000\du,13.006250\du)--(6.650000\du,15.006250\du)--(4.550000\du,15.006250\du)--cycle;
\definecolor{dialinecolor}{rgb}{0.000000, 0.000000, 0.000000}
\pgfsetstrokecolor{dialinecolor}
\draw (5.600000\du,13.006250\du)--(6.650000\du,15.006250\du)--(4.550000\du,15.006250\du)--cycle;
\pgfsetlinewidth{0.050000\du}
\pgfsetdash{}{0pt}
\pgfsetdash{}{0pt}
\pgfsetbuttcap
{
\definecolor{dialinecolor}{rgb}{0.000000, 0.000000, 0.000000}
\pgfsetfillcolor{dialinecolor}
\definecolor{dialinecolor}{rgb}{0.000000, 0.000000, 0.000000}
\pgfsetstrokecolor{dialinecolor}
\draw (4.550000\du,11.000000\du)--(6.650000\du,11.000000\du);
}
\pgfsetlinewidth{0.050000\du}
\pgfsetdash{}{0pt}
\pgfsetdash{}{0pt}
\pgfsetbuttcap
{
\definecolor{dialinecolor}{rgb}{0.000000, 0.000000, 0.000000}
\pgfsetfillcolor{dialinecolor}
\definecolor{dialinecolor}{rgb}{0.000000, 0.000000, 0.000000}
\pgfsetstrokecolor{dialinecolor}
\draw (3.500000\du,13.000000\du)--(4.550000\du,15.006250\du);
}
\pgfsetlinewidth{0.050000\du}
\pgfsetdash{}{0pt}
\pgfsetdash{}{0pt}
\pgfsetbuttcap
{
\definecolor{dialinecolor}{rgb}{0.000000, 0.000000, 0.000000}
\pgfsetfillcolor{dialinecolor}
\definecolor{dialinecolor}{rgb}{0.000000, 0.000000, 0.000000}
\pgfsetstrokecolor{dialinecolor}
\draw (7.700000\du,13.000000\du)--(6.650000\du,15.006250\du);
}
\pgfsetlinewidth{0.050000\du}
\pgfsetdash{}{0pt}
\pgfsetdash{}{0pt}
\pgfsetbuttcap
{
\definecolor{dialinecolor}{rgb}{0.000000, 0.000000, 0.000000}
\pgfsetfillcolor{dialinecolor}
\pgfsetarrowsend{stealth}
\definecolor{dialinecolor}{rgb}{0.000000, 0.000000, 0.000000}
\pgfsetstrokecolor{dialinecolor}
\draw (3.300000\du,13.900000\du)--(7.800000\du,11.500000\du);
}
\pgfsetlinewidth{0.000000\du}
\pgfsetdash{}{0pt}
\pgfsetdash{}{0pt}
\pgfsetbuttcap
{
\definecolor{dialinecolor}{rgb}{0.000000, 0.000000, 0.000000}
\pgfsetfillcolor{dialinecolor}
\pgfsetarrowsend{stealth}
\definecolor{dialinecolor}{rgb}{0.000000, 0.000000, 0.000000}
\pgfsetstrokecolor{dialinecolor}
\draw (3.847049\du,15.293750\du)--(3.843750\du,13.787500\du);
}
\pgfsetlinewidth{0.000000\du}
\pgfsetdash{}{0pt}
\pgfsetdash{}{0pt}
\pgfsetbuttcap
{
\definecolor{dialinecolor}{rgb}{0.000000, 0.000000, 0.000000}
\pgfsetfillcolor{dialinecolor}
\pgfsetarrowsend{stealth}
\definecolor{dialinecolor}{rgb}{0.000000, 0.000000, 0.000000}
\pgfsetstrokecolor{dialinecolor}
\draw (7.072049\du,15.256250\du)--(7.071250\du,11.936875\du);
}
\definecolor{dialinecolor}{rgb}{0.000000, 0.000000, 0.000000}
\pgfsetstrokecolor{dialinecolor}
\node[anchor=west] at (3.800000\du,15.600000\du){$p$};
\definecolor{dialinecolor}{rgb}{0.000000, 0.000000, 0.000000}
\pgfsetstrokecolor{dialinecolor}
\node[anchor=west] at (6.900000\du,15.600000\du){$q$};
\definecolor{dialinecolor}{rgb}{0.000000, 0.000000, 0.000000}
\pgfsetstrokecolor{dialinecolor}
\node[anchor=west] at (3.900000\du,15.600000\du){};
\pgfsetlinewidth{0.050000\du}
\pgfsetdash{}{0pt}
\pgfsetdash{}{0pt}
\pgfsetmiterjoin
\definecolor{dialinecolor}{rgb}{1.000000, 1.000000, 1.000000}
\pgfsetfillcolor{dialinecolor}
\fill (10.200000\du,11.000000\du)--(10.200000\du,13.200000\du)--(11.600000\du,13.200000\du)--(11.600000\du,11.000000\du)--cycle;
\definecolor{dialinecolor}{rgb}{0.000000, 0.000000, 0.000000}
\pgfsetstrokecolor{dialinecolor}
\draw (10.200000\du,11.000000\du)--(10.200000\du,13.200000\du)--(11.600000\du,13.200000\du)--(11.600000\du,11.000000\du)--cycle;
\pgfsetlinewidth{0.050000\du}
\pgfsetdash{}{0pt}
\pgfsetdash{}{0pt}
\pgfsetmiterjoin
\definecolor{dialinecolor}{rgb}{1.000000, 1.000000, 1.000000}
\pgfsetfillcolor{dialinecolor}
\fill (10.200000\du,13.200000\du)--(10.200000\du,15.000000\du)--(11.600000\du,15.000000\du)--(11.600000\du,13.200000\du)--cycle;
\definecolor{dialinecolor}{rgb}{0.000000, 0.000000, 0.000000}
\pgfsetstrokecolor{dialinecolor}
\draw (10.200000\du,13.200000\du)--(10.200000\du,15.000000\du)--(11.600000\du,15.000000\du)--(11.600000\du,13.200000\du)--cycle;
\pgfsetlinewidth{0.050000\du}
\pgfsetdash{}{0pt}
\pgfsetdash{}{0pt}
\pgfsetbuttcap
{
\definecolor{dialinecolor}{rgb}{0.000000, 0.000000, 0.000000}
\pgfsetfillcolor{dialinecolor}
\pgfsetarrowsend{stealth}
\definecolor{dialinecolor}{rgb}{0.000000, 0.000000, 0.000000}
\pgfsetstrokecolor{dialinecolor}
\draw (9.200000\du,14.400000\du)--(12.400000\du,12.400000\du);
}
\pgfsetlinewidth{0.000000\du}
\pgfsetdash{}{0pt}
\pgfsetdash{}{0pt}
\pgfsetbuttcap
{
\definecolor{dialinecolor}{rgb}{0.000000, 0.000000, 0.000000}
\pgfsetfillcolor{dialinecolor}
\pgfsetarrowsend{stealth}
\definecolor{dialinecolor}{rgb}{0.000000, 0.000000, 0.000000}
\pgfsetstrokecolor{dialinecolor}
\draw (9.400000\du,12.600000\du)--(10.130100\du,13.713700\du);
}
\pgfsetlinewidth{0.000000\du}
\pgfsetdash{}{0pt}
\pgfsetdash{}{0pt}
\pgfsetbuttcap
{
\definecolor{dialinecolor}{rgb}{0.000000, 0.000000, 0.000000}
\pgfsetfillcolor{dialinecolor}
\pgfsetarrowsend{stealth}
\definecolor{dialinecolor}{rgb}{0.000000, 0.000000, 0.000000}
\pgfsetstrokecolor{dialinecolor}
\draw (12.600000\du,14.500000\du)--(11.673850\du,13.001200\du);
}
\definecolor{dialinecolor}{rgb}{0.000000, 0.000000, 0.000000}
\pgfsetstrokecolor{dialinecolor}
\node[anchor=west] at (8.800000\du,12.400000\du){$p$};
\definecolor{dialinecolor}{rgb}{0.000000, 0.000000, 0.000000}
\pgfsetstrokecolor{dialinecolor}
\node[anchor=west] at (12.300000\du,15.000000\du){$q$};
\definecolor{dialinecolor}{rgb}{0.000000, 0.000000, 0.000000}
\pgfsetstrokecolor{dialinecolor}
\node[anchor=west] at (5.200000\du,10.500000\du){(a)};
\definecolor{dialinecolor}{rgb}{0.000000, 0.000000, 0.000000}
\pgfsetstrokecolor{dialinecolor}
\node[anchor=west] at (10.400000\du,10.500000\du){(b)};
\definecolor{dialinecolor}{rgb}{0.000000, 0.000000, 0.000000}
\pgfsetstrokecolor{dialinecolor}
\node[anchor=west] at (10.600000\du,10.400000\du){};
\end{tikzpicture}

%% file: tetra.tex
In this section, we give a pseudo-polynomial time algorithm for solving the folding problem for a
tetramonohedron $Q$.
We can solve the folding problem for a regular tetrahedron as a special case of it.
Precisely, the input of the folding problem in this section is $P$ and three lengths $a,b,c$ 
of an acute triangle of four congruent faces of $Q$.
To simplify, we denote by $Q(a,b,c)$ the tetramonohedron defined by the edge lengths.
By the Heron's formula, the surface area of $Q$ is given by $4\sqrt{s(s-a)(s-b)(s-c)}$, where $s=(a+b+c)/2$.
Thus, without loss of generality, we assume that the area of $P$ is equal to the surface area of $Q(a,b,c)$.
The reason why this special case is difficult is that 
each vertex of $Q(a,b,c)$ has curvature $180^\circ$ (Lemma \ref{lem:tetra}).
When the cut line ends at a vertex and $Q(a,b,c)$ is unfolded, the vertex makes $180^\circ$ on $\partial P$.
Thus we cannot find the vertex of $Q(a,b,c)$ on a straight line of $\partial P$ straightforwardly.
The mathematical characterization of a net of a tetramonohedron is done by a tiling, which is called
\emph{p2 tiling} in which rotation symmetry plays an important role. 
However, finding rotation centers is not easy since the point has curvature $180^\circ$.
In order to construct the pseudo-polynomial time algorithm, 
we use the notion of Conway condition \cite{Schattschneider1980}.

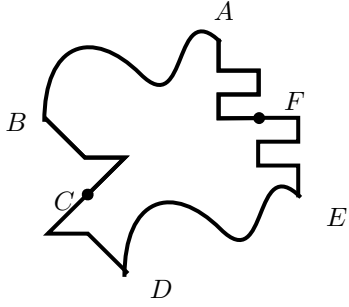
\begin{figure}\centering
\input{figure/Conway.tex}
\caption{A Conway tiling.}
\label{fig:Conway}
\end{figure}

Let $P$ be a polygon and $p,q$ be two points in $\partial P$. 
Then $[p,q]$ denotes the part of boundary $\partial P$ starting from $p$ to $q$ in counterclockwise.
Then, a polygon $P$ is called \emph{Conway tile} if it has six points 
$A$, $B$, $C$, $D$, $E$, and $F$ in $\partial P$ in counterclockwise such that
(1) $[A,B]$ can be moved to $[D,E]$ by translation $\tau$ with $\tau(A)=E$ and $\tau(B)=D$,
(2) each of $[B,C]$, $[C,D]$, $[E,F]$, and $[F,A]$ is rotation symmetry with respect to the midpoint of it, and
(3) at least three of these six points are distinct (\figurename~\ref{fig:Conway}).
In \cite{AkiyamaMatsunaga2020}, they prove that $P$ is a net of a tetramonohedron
if and only if $P$ is a Conway tile.
Using this characterization, we show the following theorem:

\begin{theorem}
\label{th:tetra} 
Let $P$ be a polygon of $n$ vertices, and $a,b,c$ be positive real numbers.
We can decide whether $P$ can fold to a tetramonohedron $Q(a,b,c)$ or not in $O(L(L+n)n^2)$ time.
\end{theorem}

Without loss of generality, we assume that the area of $P$ is equal to the surface area of $Q(a,b,c)$.
We prove the theorem by following the proof of the characterization in \cite{AkiyamaMatsunaga2020}.
Let four vertices of a given tetramonohedron $Q$ be $v_1,v_2,v_3,v_4$.
As shown in Theorem \ref{th:spanning}, the set of cut lines of $Q$ to unfold $P$ 
forms a tree $T(Q)$ spanning all four vertices $v_1,v_2,v_3,v_4$.
Thus $T(Q)$ contains at least two leaves of degree 1.
We first note that no point $p$ on the surface of $Q$ cannot make 
a leaf in $T(Q)$ except on four vertices $v_1,v_2,v_3,v_4$ since the curvature at $p$ is $360^\circ$ 
(such a point is made by a ``redundant'' cut and reduced on $P$ when it is unfolded).
That is, every leaf of $T(Q)$ corresponds to one of the four vertices of $Q$.
The vertices of degree greater than 2 are also defined in the same manner as the graph theory.
We now define the vertex of degree 2 more carefully.
Let $p$ be a point in $T(Q)$ such that two line segments $\ell$ and $\ell'$ are incident to $p$.
Then, we have two angles around $p$ between $\ell$ and $\ell'$.
The point $p$ is a \emph{vertex} $p$ on $T(Q)$ of degree 2 
if and only if one of these two angles is not equal to $180^\circ$.
In other words, each point on $T(Q)$ is not considered as a vertex of degree 2 when 
the point is surrounded by two $180^\circ$ angles.
Intuitively, two consecutive line segments in a kinked line on the surface of $Q$ share
a ``vertex'' of $T(Q)$ of degree 2 if they make non-$180^\circ$ angle at the point. 
\begin{obs}
\label{obs:deg2}
Each vertex $v_i$ of $Q$ is always a vertex of $T(Q)$.
\end{obs}
\begin{proof}
If $v_i$ has degree 1 or greater than 2, it is a vertex of $T(Q)$.
If it has degree 2, it should be a vertex of $T(Q)$ since the curvature of $v_i$ is $180^\circ$.
\qed\end{proof}

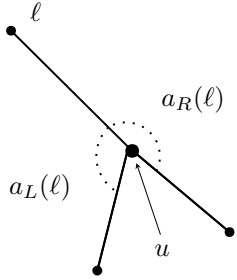
\begin{figure}\centering
\input{figure/angles.tex}
\caption{The two angles $a_R(\ell)$ and $a_L(\ell)$ for a leaf $\ell$.}
\label{fig:angles}
\end{figure}

Now we focus on the spanning tree $T(Q)$ of $Q$;
every vertex $v_i$ of $Q$ is also a vertex of $T(Q)$, and every leaf of $T(Q)$ is a vertex of $Q$.
For this tree $T(Q)$, we introduce some notations.
For each leaf $\ell$ of $T(Q)$, the \emph{associated edge} $e(\ell)$ of $\ell$ 
is the unique edge incident to $\ell$. 
When $e(\ell)=\{u,\ell\}$, $u$ is the \emph{parent} of $\ell$.
On $T(Q)$, $a_R(\ell)$ and $a_L(\ell)$ are the angles made by 
$e(\ell)$ with its neighbor edges sharing the parent of $\ell$ in clockwise and counterclockwise,
 respectively (see \figurename~\ref{fig:angles}).

A leaf $\ell$ of $T(Q)$ corresponds to a vertex of a tetramonohedron $Q$.
Therefore, the leaf is mapped to a unique point in $\partial P$. 
However, since $Q$ is a tetramonohedron, 
the curvature at $\ell$ is $180^\circ$ by Lemma \ref{lem:tetra}, 
hence it is an inner point on an edge (or a line segment) of $P$.
Thus we need some tricks to find it for a given $P$.
By the definition of vertices of $T(Q)$, 
if the leaf $\ell$ satisfies the following property $(*)$,
$\ell$ is the midpoint of an edge:

\medskip 
\noindent{\sf Property $(*)$}: 
Both of $a_R(\ell)\neq 180^\circ$ and $a_L(\ell)\neq 180^\circ$.
\medskip 

If a leaf $\ell$ of $T(Q)$ has Property $(*)$, the neighbor points on $\partial P$ are vertices of $P$.
That is, for the parent $u$ of $\ell$, 
we have an edge (or straight line) $u'u''$ on $\partial P$ such that $u'$ and $u''$ are vertices of $P$
and they are glued together to make the vertex $u$ on $Q$ and 
the midpoint of the line segment $u'u''$ in $\partial P$ corresponds to $\ell$ on $T(Q)$.
In other words, when the curvatures at $u'$ and $u''$ are both not equal to $180^\circ$, 
it is easy to find the corresponding vertex $\ell$ of $T(Q)$, or a vertex $v_i$ on $Q$ of 
curvature $180^\circ$.

Since $T(Q)$ is a tree, it has at least two leaves.
We first consider the case that two leaves have Property $(*)$ as Type 1.
We have two exceptional cases, which will be handled as Type 2 and Type 3 later.

\paragraph{Type 1:}
In Type 1, the $T(Q)$ has at least two leaves that satisfy Property $(*)$.
Let $\ell$ be a leaf satisfying Property $(*)$ and $v$ be the parent of $\ell$.
Then it is easy to see that there is a line segment 
made from two copies of the edge $v \ell$ of $T(Q)$ on $\partial P$.
In other words, $\partial P$ has two consecutive vertices $p_i$ and $p_{i+1}$ such that
the midpoint of the line segment $p_i p_{i+1}$ will be folded to $\ell$, 
which is one of $v_1,v_2,v_3,v_4$, and $p_i$ and $p_{i+1}$ are glued together to make 
the vertex $v$ on $Q$. 
Since $T(Q)$ has two leaves, we can obtain two of $v_1,v_2,v_3,v_4$.
Thus we obtain the following algorithm:

\begin{algorithm}[h]
 \caption{Folding algorithm for Type 1}
 \label{alg:type1}
 \SetKwInOut{Input}{Input}
 \SetKwInOut{Output}{Output}
 \Input{A polygon $P=(p_0,p_1,\ldots,p_{n-1},p_0)$ and three positive numbers $a,b,c$}
 \Output{All ways of folding $P$ to a tetramonohedron $Q(a,b,c)$ in Type 1 (if one exists)}
 \ForEach{pair of two edges $\{e,e'\}$ of $P$}{
   take midpoints $m$ of $e$ and $m'$ of $e'$\;
   \ForEach{triangular lattice defined by $(a,b,c)$ having two grid points $m$ and $m'$}{
      Perform the stamping of $Q$ on $P$ along the lattice\;
      Check if $P$ is a net of $Q$ by folding and gluing $P$ based on the partition of $Q$ by stamping\;
   }
 }
\end{algorithm}

\begin{figure}\centering
\input{figure/lattice.tex}
\caption{A polygon $P$ on a lattice by $(a,b,c)$.}
\label{fig:lattice}
\end{figure}
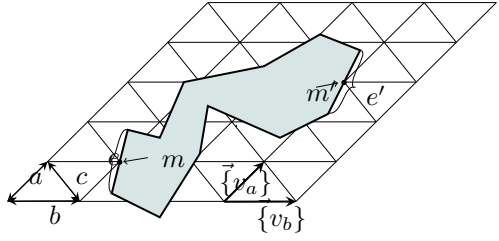

Since the algorithm checks all combinations, 
the correctness of Algorithm \ref{alg:type1} is trivial:
It decides whether $P$ can fold to a tetramonohedron $Q$ when $T(Q)$ is in Type 1.
In order to evaluate running time of Algorithm \ref{alg:type1},
we show a technical lemma:
\begin{lemma}
\label{lem:lattice}
Let $P$ be a simple polygon $(p_0,p_1,\ldots,p_{n-1},p_0)$,
and $m, m'$ be two midpoints of two edges $e,e'$ of coordinates $(x,y)$ and $(x',y')$.
Let $Q(a,b,c)$ be a tetramonohedron having the surface area equal to the area of $P$.
Then the number of triangular lattice defined by $(a,b,c)$ having two grid points $m$ and $m'$ is $O(L)$.
\end{lemma}
\begin{proof}
An illustration of the desired lattice is depicted in \figurename~\ref{fig:lattice}.
That is, a feasible lattice for $m$ and $m'$ has grid points on $m$ and $m'$.
Then, we introduce two vectors $\vec{v_a}$ and $\vec{v_b}$ 
that span two edges of the unit lattice triangle of length $a$ and $b$, respectively.
Then, it is easy to see that $m$ and $m'$ are on grid points if and only if 
there are two integers $k_a$ and $k_b$ such that $\vec{mm'}=k_a \vec{v_a}+k_b \vec{v_b}$.
Then we have $-L/a \le k_a \le L/a$. Thus the number of possible $k_a$ is $O(L+n)$.
Once $k_a$ is fixed, we can compute if $k_b$ is a reasonable integer or not.
Therefore, 
we can conclude that the number of triangular lattices 
defined by $(a,b,c)$ having two grids $m$ and $m'$ is $O(L)$.
\qed\end{proof}

Now we show the running time:
\begin{lemma}
\label{lem:type1}
Algorithm \ref{alg:type1} runs in $O(L(L+n)n^2)$ time.
\end{lemma}
\begin{proof}
The number of pairs of two edges is $O(n^2)$.
The stamping of $Q$ on $P$ and gluing check of $P$ takes $O(L+n)$ time
as shown in Theorem \ref{th:stamping} and Theorem \ref{th:glue}.
By Lemma \ref{lem:lattice}, the number of possible triangular lattices for 
a given pair of grid points is $O(L)$. 
Thus Algorithm \ref{alg:type1} runs in $O(L(L+n)n^2)$.
\qed 
\end{proof}

We here give two other exceptional cases, 
and we show that any $T(Q)$ is in one of Types 1, 2, and 3 later.

\begin{figure}
\centering
\input{figure/Type2.tex}
\caption{Each leaf $v_i$ satisfies $a_R(v_i)=180^\circ$ or $a_L(v_i)=180^\circ$ in Type 2.}
\label{fig:type2}
\end{figure}
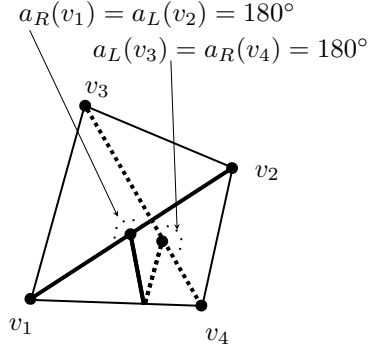

\paragraph{Type 2:}
In Type 2, the set of cut lines of $Q$ contains two independent line segments, say, $v_1 v_2$ and $v_3 v_4$.
When we first cut along these lines, we obtain a cylinder,
which is called the ``rolling belt'' in \cite{DemaineORourke2007}.
After that, the cylinder is cut, unfolded, and $P$ is obtained from it. 
If the last cut line(s) does not touch any of four vertices of $Q$, 
we can observe that for each leaf $\ell$, one of $a_R(\ell)$ and $a_L(\ell)$ is $180^\circ$ (\figurename~\ref{fig:type2}).
Therefore, there are two edges $e$ and $e'$ in $\partial P$ corresponding to $v_1 v_2$ and $v_3 v_4$. 
Thus $e$ is parallel to $e'$ on $P$ and the length of $\msize{e}=\msize{e'}=2\msize{v_1v_2}=2\msize{v_3v_4}$.
In this case, since a cylinder is obtained when $P$ is glued except $e$ and $e'$ to a tetramonohedron $Q$,
we can obtain (another) tetramonohedron $Q'$ when we fold along the crease line joining 
two midpoints of $e$ and $e'$. 
In other words, when $P$ is given, we can fold infinitely many distinct tetramonohedra.
Hence, without loss of generality, we assume that one endpoint of $e$ and another endpoint of $e'$ 
are vertices of curvature $180^\circ$ of a tetramonohedron $Q'$.
Thus we can determine if $P$ can fold to a tetramonohedron by checking if $P$ can fold to a cylinder by gluing except $e$ and $e'$.

\begin{algorithm}[h]
 \caption{Folding algorithm for Type 2}
 \label{alg:type2}
 \SetKwInOut{Input}{Input}
 \SetKwInOut{Output}{Output}
 \Input{A polygon $P=(p_0,p_1,\ldots,p_{n-1},p_0)$}
 \Output{A way of folding $P$ to a tetramonohedron in Type 2 (if one exists)}
 \ForEach{pair of two edges $\{e,e'\}$ of $P$}{
   \If{if $e=\{v_i,v_{i+1}\}$ and $e'=\{v_{j},v_{j+1}\}$ are parallel and $\msize{e}=\msize{e'}$}{
      Check if the path $(v_{i+1},\ldots,v_{j})$ can be glued to $(v_{j+1},\ldots,v_{i})$\;
      \lIf{Two paths are glued}{output ``Yes''}
   }
 }
\end{algorithm}

The matching of the paths $(v_{i+1},\ldots,v_{j})$ and $(v_{j+1},\ldots,v_{i})$
can be done by checking the following conditions:
$\angle v_{i-1}v_{i}v_{i+1}+\angle v_{i}v_{i+1}v_{i+2}=180^\circ$,
$\angle v_{j-1}v_{j}v_{j+1}+\angle v_{j}v_{j+1}v_{j+2}=180^\circ$,
$\angle v_{i+k}v_{i+k+1}v_{i+k+2}+\angle v_{j-k}v_{j-k+1}v_{j-k+2}=360^\circ$ for each $k=0,1,2,\ldots$,
$\msize{v_{i+k+1}v_{i+k+2}}=\msize{v_{j-k+1} v_{j-k}}$ for each $k=0,1,2,\ldots$,
and the number of vertices in $v_{i+1},\ldots,v_j$ and 
the number of vertices in $v_{j+1},\ldots,v_{i}$ are equal.
The correctness of Algorithm \ref{alg:type2} in this case is trivial.
\begin{lemma}
\label{lem:type2}
Algorithm \ref{alg:type2} runs in $O(n^3)$ time.
\end{lemma}
\begin{proof}
The number of pairs of two edges is $O(n^2)$.
The matching of two paths can be checked in linear time.
Thus Algorithm \ref{alg:type2} runs in $O(n^3)$ time.
\qed 
\end{proof}

\begin{figure}
\centering
\input{figure/Type3.tex}
\caption{Type 3: Only $v_3$ satisfies Condition $(*)$.}
\label{fig:type3}
\end{figure}
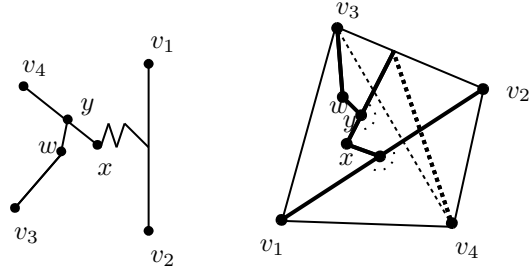

\paragraph{Type 3:}
In Type 3, the set of cut lines of $Q$ contains two independent line segments, 
say, $v_1 v_2$ and $x v_4$ with $x\neq v_3$. 
In this case, only $v_3$ satisfies Property $(*)$.
On $T(Q)$, we have the situation shown in \figurename~\ref{fig:type3};
$v_1$ and $v_2$ are joined by a straight line and there are three vertices
$x,y,w$ such that $v_4$ and $x$ are joined by a straight line,
$y$ is on the line segment $v_4 y$, and $v_3$ is joined to $y$ 
with some cut lines. (Note: $v_3$ can be $w$.)
Therefore, on $\partial P$, $v_3$ is the midpoint of an edge of $P$,
we have the same sequence of length from $v_3$ to both sides,
and 
when we find the first pair $e$ and $e'$ of $\msize{e}>\msize{e'}$,
then $e'=(w,y')$ and $e=(y'',v_4,y''',x)$, where $y',y'',y'''$ are 
the three vertices forming $y$ on $Q$. 
Hence we can find the point $v_4$ on $\partial P$ uniquely.
Thus we can determine if $P$ can fold to a tetramonohedron $Q$ in Type 3.

\begin{algorithm}[h]
 \caption{Folding algorithm for Type 3}
 \label{alg:type3}
 \SetKwInOut{Input}{Input}
 \SetKwInOut{Output}{Output}
 \Input{A polygon $P=(p_0,p_1,\ldots,p_{n-1},p_0)$}
 \Output{A way of folding $P$ to a tetramonohedron in Type 3 (if one exists)}
 \ForEach{an edge $e$ of $P$}{
   Take the midpoint of $e$ as $v_3$\;
   Find $v_4$ from $v_3$\;
   Glue from $v_3$ and $v_4$ and obtain the points $x$, $y$, and $w$\;
   Continue gluing from $y$ to the last edge\;
   Check if $P$ is a net of a tetramonohedron by comparing $\msize{v_3v_4}$ and the last unglued edge\;
 }
\end{algorithm}

In the last step, when the algorithm obtains an open end of a cylinder by cutting along the edge $v_1v_2$, 
it concludes that if $P$ can fold to a tetramonohedron if $2\msize{v_3 v_4}$ is equal to 
the length of the open end (or cycle).
The correctness of Algorithm \ref{alg:type3},
or if $P$ is a polygon that can fold to a tetramonohedron $Q$ when $T(Q)$ is in Type 3,
follows the arguments in \cite{AkiyamaMatsunaga2020}.
Thus we give the running time:
\begin{lemma}
\label{lem:type3}
Algorithm \ref{alg:type3} runs in $O(n^2)$ time.
\end{lemma}
\begin{proof}
The number of edges of $P$ is $O(n)$.
The vertex $v_4$ can be found in linear time by just following both sides from $v_3$.
The gluing can be done in $O(n)$ time from $v_3$ and $v_4$.
Therefore, Algorithm \ref{alg:type3} runs in $O(n^2)$ time.
\qed 
\end{proof}

Before the proof of the main theorem, we give a technical lemma for Property $(*)$:
\begin{lemma}
\label{lem:property}
(1) For any leaf $\ell$, it satisfies Property $(*)$ if its parent $v$ has degree 2.
(2) If two leaves $\ell$ and $\ell'$ share their parent $v$ of degree 3
 and the angle between $e(\ell)$ and $e(\ell')$ is not $180^\circ$, 
 then at least one of $\ell$ and $\ell'$ satisfies Property $(*)$.
(3) When four leaves share their parent $r$, at least two leaves satisfy Property $(*)$.
\end{lemma}
\begin{proof}
(1) If the parent $v$ is a vertex of $Q$, we have $a_R(\ell)+a_L(\ell)=180^\circ$ 
and $0<a_R(\ell),a_L(\ell)<180^\circ$. Thus Property $(*)$ holds.

\noindent
(2) Since $\ell$ and $\ell'$ share the angle, $a_R(\ell)=a_L(\ell')$ or $a_R(\ell')=a_L(\ell)$.
Without loss of generality, we assume that $a_R(\ell)=a_L(\ell')$.
Then, $a_R(\ell)\neq 180^\circ$ by assumption. 
Since $a_R(\ell)+a_R(\ell')+a_L(\ell)=360^\circ$, $a_R(\ell')+a_L(\ell)\neq 180^\circ$.
Therefore, at least one of $a_R(\ell')$ and $a_L(\ell)$ is not equal to $180^\circ$.
Thus at least one of $\ell$ and $\ell'$ satisfies Property $(*)$.

\noindent
(3) Since $r$ has 4 children leaves, at most one angle can be equal to or greater than $180^\circ$, 
and the other three angles are consecutively less than $180^\circ$. 
Let $\ell$ and $\ell'$ be the leaves between these three consecutive angles less than  $180^\circ$.
Then these two leaves satisfy Property $(*)$.
\qed 
\end{proof}

We now turn to the proof of the main theorem in this section:
\begin{proof}%
(of Theorem \ref{th:tetra})
For a given polygon $P$, we perform three Algorithms \ref{alg:type1}, \ref{alg:type2}, and \ref{alg:type3} one
by one. By Lemmas \ref{lem:type1}, \ref{lem:type2}, and \ref{lem:type3},
the running time is $O(L(L+n)n^2)$ time in total.
Thus it is sufficient to show that any cut lines of $Q$ to obtain $P$ should be one of these types.

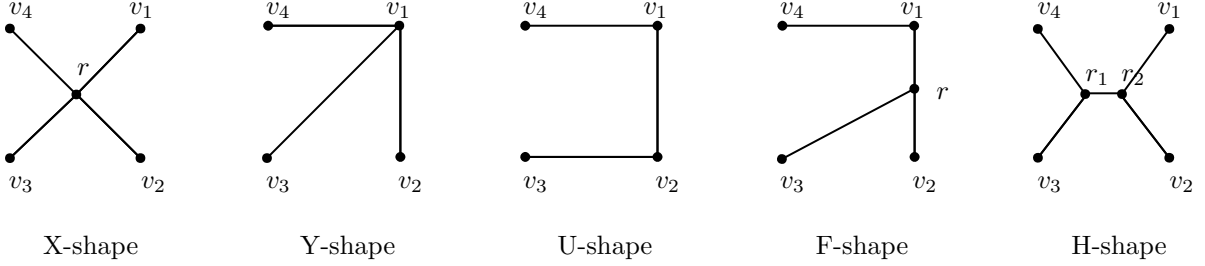
\begin{figure}\centering
\begin{minipage}[t]{0.15\textwidth}\centering
\input{figure/X.tex} 

\vspace*{-1cm}
X-shape
\end{minipage}
\hfill
\begin{minipage}[t]{0.15\textwidth}\centering
\input{figure/Y.tex}

\vspace*{-1cm}
Y-shape
\end{minipage}
\hfill
\begin{minipage}[t]{0.15\textwidth}\centering
\input{figure/U.tex}

\vspace*{-1cm}
U-shape
\end{minipage}
\hfill
\begin{minipage}[t]{0.15\textwidth}\centering
\input{figure/F.tex}

\vspace*{-1cm}
F-shape
\end{minipage}
\hfill
\begin{minipage}[t]{0.15\textwidth}\centering
\input{figure/H.tex}

\vspace*{-1cm}
H-shape
\end{minipage}
\caption{X-shape, Y-shape, U-shape, F-shape, and H-shape.}
\label{fig:Shapes}
\end{figure}

From now on, we consider the topological structure of $T(Q)$, 
and show that most cases are in Type 1 except two special cases, which imply Type 2 and Type 3.
According to the case analysis in \cite{AkiyamaMatsunaga2020},
$T(Q)$ has one of the following topological structures of X-shape, Y-shape, U-shape, F-shape, and 
H-shape (\figurename~\ref{fig:Shapes}):

\paragraph{X-shape:}
All of $v_1,v_2,v_3,v_4$ are leaves, and there is a vertex $r$ in $T(Q)$ with $\deg(r)=4$.
By Lemma \ref{lem:property}(3), this case is in Type 1.

\paragraph{Y-shape:}
Three of $v_1,v_2,v_3,v_4$ are leaves, and the last one is a vertex of degree 3.
Without loss of generality, we assume that $\deg(v_1)=3$ and $\deg(v_2)=\deg(v_3)=\deg(v_4)=1$.
Then none of three angles $\angle v_2 v_1 v_3$, $\angle v_3 v_1 v_4$ and $\angle v_4 v_1 v_2$
is equal to $180^\circ$ since $v_1$ is a vertex of $Q$ of curvature 
is $\angle v_2 v_1 v_3+\angle v_3 v_1 v_4+\angle v_4 v_1 v_2=180^\circ$,
and $0<\angle v_2 v_1 v_3,\angle v_3 v_1 v_4,\angle v_4 v_1 v_2$.
Hence this case is in Type 1.

\paragraph{U-shape:}
Two of $v_1,v_2,v_3,v_4$ are leaves, and the other two are vertices of degree 2.
Without loss of generality, we assume that $\deg(v_1)=\deg(v_2)=2$ and $\deg(v_3)=\deg(v_4)=1$, 
and $v_1$ is closer to $v_4$ than $v_2$.
If $T(Q)$ has a vertex $u$ of degree 2 between $v_1$ and $v_4$, $v_4$ has Property $(*)$ by Lemma \ref{lem:property}(1).
Thus we consider the case that $v_1$ is the parent of $v_4$. However, then 
$v_4$ has Property $(*)$ by Lemma \ref{lem:property}(1) for the parent $v_1$.
The leaf $v_3$ also satisfies Property $(*)$ by the same argument with $v_2$.
Thus this case is in Type 1.

\paragraph{F-shape:}
Three of $v_1,v_2,v_3,v_4$ are leaves, and the last one is a vertex of degree 2,
and $T(Q)$ has another vertex $r$ of $\deg(r)=3$.
Without loss of generality, $v_1$ is the vertex of degree 2,
and $v_2$ is the leaf reachable to $v_1$ without through $r$.
Then, by the same argument of U-shape, $v_2$ satisfies Property $(*)$.
In the same way, if $v_3$ or $v_4$ has other vertices of degree 2 on the way to $r$, 
it satisfies Property $(*)$. 
Thus, we consider the other case that both $v_3$ and $v_4$ are children of $r$.
If the angle $\angle v_3 r v_4\neq 180^\circ$, by Lemma \ref{lem:property}(2),
one of $v_3$ and $v_4$ satisfies Property $(*)$.
On the other hand, when $\angle v_3 r v_4 = 180^\circ$, 
we have two independent line segments $v_1 v_2$ and $v_3 v_4$.
This case is in Type 2. Intuitively, the cut lines $v_1 v_2$ and $v_3 v_4$ open $Q$ to a cylinder,
and the cylinder is open by cutting line segment(s) joining $v_1$ and $r$.

\paragraph{H-shape:}
All of $v_1,v_2,v_3,v_4$ are leaves, and there are two vertices $r_1$ and $r_2$ in $T(Q)$ with $\deg(r_1)=\deg(r_2)=3$.
We assume that $r_1$ has children $v_1$ and $v_2$, and $r_2$ has children $v_3$ and $v_4$.
(When the other vertices are between them, we can reduce to the other cases above.)
If $\angle v_1 r_1 v_2\neq 180^\circ$ and $\angle v_3 r_2 v_4\neq 180^\circ$,
by Lemma \ref{lem:property}(2), two vertices $r_1$ and $r_2$ have at least one leaf satisfying 
Property $(*)$.
On the other hand, when $\angle v_1 r_1 v_2 =180^\circ$ and $\angle v_3 r_2 v_4 =180^\circ$, we have the case in Type 2.
The last case is that, without loss of generality, 
 $\angle v_1 r_1 v_2 = 180^\circ$ and $\angle v_3 r_2 v_4\neq 180^\circ$.
Let $r'$ be the third neighbor of $r_2$ other than $v_3$ and $v_4$ (it can be $r'=r_2$).
If $\angle r' r_2 v_3\neq 180^\circ$ and $\angle r' r_2 v_4\neq 180^\circ$,
both of $v_3$ and $v_4$ satisfy Property $(*)$, which is in Type 1.
Type 3 is the last remaining case that $\angle v_1 r_1 v_2 =180^\circ$ and $\angle v_4 r_2 r'=180^\circ$.

Since all cases are covered by the results in \cite{AkiyamaMatsunaga2020},
we can conclude that three Algorithms \ref{alg:type1}, \ref{alg:type2}, and \ref{alg:type3} 
can determine whether $P$ can fold to $Q$ or not.
\qed
\end{proof}

\begin{cor}
\label{cor:tetra} 
Let $P$ be a polygon of $n$ vertices.
We can decide whether $P$ can fold to a regular tetrahedron or not in $O(L(L+n)n^2)$ time.
\end{cor}
\begin{proof}
When $Q$ is a regular tetrahedron, by letting $a=b=c$, 
we obtain the same results as for the tetramonohedron.
Thus the running time of Algorithm \ref{alg:type1} is $O(L(L+n)n^2)$.
Algorithms \ref{alg:type2} and \ref{alg:type3} run in this bound, which completes the proof.
\qed
\end{proof}

%% file: figure/Conway.tex
\ifx\du\undefined
  \newlength{\du}
\fi
\setlength{\du}{30\unitlength}
\begin{tikzpicture}
\pgftransformxscale{1.000000}
\pgftransformyscale{-1.000000}
\definecolor{dialinecolor}{rgb}{0.000000, 0.000000, 0.000000}
\pgfsetstrokecolor{dialinecolor}
\definecolor{dialinecolor}{rgb}{1.000000, 1.000000, 1.000000}
\pgfsetfillcolor{dialinecolor}
\definecolor{dialinecolor}{rgb}{0.000000, 0.000000, 0.000000}
\pgfsetstrokecolor{dialinecolor}
\node[anchor=west] at (4.800000\du,11.200000\du){$A$};
\definecolor{dialinecolor}{rgb}{0.000000, 0.000000, 0.000000}
\pgfsetstrokecolor{dialinecolor}
\node[anchor=west] at (2.200000\du,12.600000\du){$B$};
\definecolor{dialinecolor}{rgb}{0.000000, 0.000000, 0.000000}
\pgfsetstrokecolor{dialinecolor}
\node[anchor=west] at (2.800000\du,13.600000\du){$C$};
\definecolor{dialinecolor}{rgb}{0.000000, 0.000000, 0.000000}
\pgfsetstrokecolor{dialinecolor}
\node[anchor=west] at (4.000000\du,14.700000\du){$D$};
\definecolor{dialinecolor}{rgb}{0.000000, 0.000000, 0.000000}
\pgfsetstrokecolor{dialinecolor}
\node[anchor=west] at (6.200000\du,13.800000\du){$E$};
\definecolor{dialinecolor}{rgb}{0.000000, 0.000000, 0.000000}
\pgfsetstrokecolor{dialinecolor}
\node[anchor=west] at (5.675000\du,12.343800\du){$F$};
\pgfsetlinewidth{0.050000\du}
\pgfsetdash{}{0pt}
\pgfsetdash{}{0pt}
\pgfsetmiterjoin
\pgfsetbuttcap
{
\definecolor{dialinecolor}{rgb}{0.000000, 0.000000, 0.000000}
\pgfsetfillcolor{dialinecolor}
}
\definecolor{dialinecolor}{rgb}{0.000000, 0.000000, 0.000000}
\pgfsetstrokecolor{dialinecolor}
\pgfpathmoveto{\pgfpoint{5.000000\du}{11.600000\du}}
\pgfpathcurveto{\pgfpoint{4.400000\du}{11.000000\du}}{\pgfpoint{4.600000\du}{12.600000\du}}{\pgfpoint{4.000000\du}{12.000000\du}}
\pgfpathcurveto{\pgfpoint{3.400000\du}{11.400000\du}}{\pgfpoint{2.800000\du}{11.600000\du}}{\pgfpoint{2.800000\du}{12.600000\du}}
\pgfusepath{stroke}
\pgfsetlinewidth{0.050000\du}
\pgfsetdash{}{0pt}
\pgfsetdash{}{0pt}
\pgfsetmiterjoin
\pgfsetbuttcap
{
\definecolor{dialinecolor}{rgb}{0.000000, 0.000000, 0.000000}
\pgfsetfillcolor{dialinecolor}
}
\definecolor{dialinecolor}{rgb}{0.000000, 0.000000, 0.000000}
\pgfsetstrokecolor{dialinecolor}
\pgfpathmoveto{\pgfpoint{6.000000\du}{13.550000\du}}
\pgfpathcurveto{\pgfpoint{5.400000\du}{12.950000\du}}{\pgfpoint{5.600000\du}{14.550000\du}}{\pgfpoint{5.000000\du}{13.950000\du}}
\pgfpathcurveto{\pgfpoint{4.400000\du}{13.350000\du}}{\pgfpoint{3.800000\du}{13.550000\du}}{\pgfpoint{3.800000\du}{14.550000\du}}
\pgfusepath{stroke}
\pgfsetlinewidth{0.050000\du}
\pgfsetdash{}{0pt}
\pgfsetdash{}{0pt}
\pgfsetmiterjoin
\pgfsetbuttcap
{
\definecolor{dialinecolor}{rgb}{0.000000, 0.000000, 0.000000}
\pgfsetfillcolor{dialinecolor}
{\pgfsetcornersarced{\pgfpoint{0.000000\du}{0.000000\du}}\definecolor{dialinecolor}{rgb}{0.000000, 0.000000, 0.000000}
\pgfsetstrokecolor{dialinecolor}
\draw (2.806250\du,12.550000\du)--(3.306250\du,13.050000\du)--(3.806250\du,13.050000\du)--(3.306250\du,13.550000\du);
}}
{\pgfsetcornersarced{\pgfpoint{0.000000\du}{0.000000\du}}\definecolor{dialinecolor}{rgb}{0.000000, 0.000000, 0.000000}
\pgfsetstrokecolor{dialinecolor}
\draw (2.806250\du,12.550000\du)--(3.306250\du,13.050000\du)--(3.806250\du,13.050000\du)--(3.306250\du,13.550000\du);
}\pgfsetlinewidth{0.050000\du}
\pgfsetdash{}{0pt}
\pgfsetmiterjoin
\pgfsetbuttcap
\definecolor{dialinecolor}{rgb}{0.000000, 0.000000, 0.000000}
\pgfsetfillcolor{dialinecolor}
\pgfpathmoveto{\pgfpoint{3.306250\du}{13.550000\du}}
\pgfpathcurveto{\pgfpoint{3.288572\du}{13.532322\du}}{\pgfpoint{3.288572\du}{13.496967\du}}{\pgfpoint{3.306250\du}{13.479289\du}}
\pgfpathcurveto{\pgfpoint{3.323928\du}{13.461612\du}}{\pgfpoint{3.359283\du}{13.461612\du}}{\pgfpoint{3.376961\du}{13.479289\du}}
\pgfpathcurveto{\pgfpoint{3.394638\du}{13.496967\du}}{\pgfpoint{3.394638\du}{13.532322\du}}{\pgfpoint{3.376961\du}{13.550000\du}}
\pgfpathcurveto{\pgfpoint{3.359283\du}{13.567678\du}}{\pgfpoint{3.323928\du}{13.567678\du}}{\pgfpoint{3.306250\du}{13.550000\du}}
\pgfusepath{fill}
\definecolor{dialinecolor}{rgb}{0.000000, 0.000000, 0.000000}
\pgfsetstrokecolor{dialinecolor}
\pgfpathmoveto{\pgfpoint{3.306250\du}{13.550000\du}}
\pgfpathcurveto{\pgfpoint{3.288572\du}{13.532322\du}}{\pgfpoint{3.288572\du}{13.496967\du}}{\pgfpoint{3.306250\du}{13.479289\du}}
\pgfpathcurveto{\pgfpoint{3.323928\du}{13.461612\du}}{\pgfpoint{3.359283\du}{13.461612\du}}{\pgfpoint{3.376961\du}{13.479289\du}}
\pgfpathcurveto{\pgfpoint{3.394638\du}{13.496967\du}}{\pgfpoint{3.394638\du}{13.532322\du}}{\pgfpoint{3.376961\du}{13.550000\du}}
\pgfpathcurveto{\pgfpoint{3.359283\du}{13.567678\du}}{\pgfpoint{3.323928\du}{13.567678\du}}{\pgfpoint{3.306250\du}{13.550000\du}}
\pgfusepath{stroke}
\pgfsetlinewidth{0.050000\du}
\pgfsetdash{}{0pt}
\pgfsetdash{}{0pt}
\pgfsetmiterjoin
\pgfsetbuttcap
{
\definecolor{dialinecolor}{rgb}{0.000000, 0.000000, 0.000000}
\pgfsetfillcolor{dialinecolor}
{\pgfsetcornersarced{\pgfpoint{0.000000\du}{0.000000\du}}\definecolor{dialinecolor}{rgb}{0.000000, 0.000000, 0.000000}
\pgfsetstrokecolor{dialinecolor}
\draw (3.843750\du,14.506300\du)--(3.343750\du,14.006300\du)--(2.843750\du,14.006300\du)--(3.343750\du,13.506300\du);
}}
\pgfsetlinewidth{0.050000\du}
\pgfsetdash{}{0pt}
\pgfsetdash{}{0pt}
\pgfsetmiterjoin
\pgfsetbuttcap
{
\definecolor{dialinecolor}{rgb}{0.000000, 0.000000, 0.000000}
\pgfsetfillcolor{dialinecolor}
{\pgfsetcornersarced{\pgfpoint{0.000000\du}{0.000000\du}}\definecolor{dialinecolor}{rgb}{0.000000, 0.000000, 0.000000}
\pgfsetstrokecolor{dialinecolor}
\draw (4.975000\du,11.556300\du)--(4.975000\du,11.956300\du)--(5.475000\du,11.956300\du)--(5.475000\du,12.256300\du)--(4.975000\du,12.256300\du)--(4.975000\du,12.556300\du)--(5.531250\du,12.553100\du);
}}
{\pgfsetcornersarced{\pgfpoint{0.000000\du}{0.000000\du}}\definecolor{dialinecolor}{rgb}{0.000000, 0.000000, 0.000000}
\pgfsetstrokecolor{dialinecolor}
\draw (4.975000\du,11.556300\du)--(4.975000\du,11.956300\du)--(5.475000\du,11.956300\du)--(5.475000\du,12.256300\du)--(4.975000\du,12.256300\du)--(4.975000\du,12.556300\du)--(5.531250\du,12.553100\du);
}\pgfsetlinewidth{0.050000\du}
\pgfsetdash{}{0pt}
\pgfsetmiterjoin
\pgfsetbuttcap
\definecolor{dialinecolor}{rgb}{0.000000, 0.000000, 0.000000}
\pgfsetfillcolor{dialinecolor}
\pgfpathmoveto{\pgfpoint{5.531250\du}{12.553100\du}}
\pgfpathcurveto{\pgfpoint{5.531394\du}{12.578100\du}}{\pgfpoint{5.506538\du}{12.603243\du}}{\pgfpoint{5.481538\du}{12.603387\du}}
\pgfpathcurveto{\pgfpoint{5.456539\du}{12.603531\du}}{\pgfpoint{5.431395\du}{12.578675\du}}{\pgfpoint{5.431252\du}{12.553675\du}}
\pgfpathcurveto{\pgfpoint{5.431108\du}{12.528676\du}}{\pgfpoint{5.455964\du}{12.503532\du}}{\pgfpoint{5.480963\du}{12.503388\du}}
\pgfpathcurveto{\pgfpoint{5.505963\du}{12.503245\du}}{\pgfpoint{5.531106\du}{12.528100\du}}{\pgfpoint{5.531250\du}{12.553100\du}}
\pgfusepath{fill}
\definecolor{dialinecolor}{rgb}{0.000000, 0.000000, 0.000000}
\pgfsetstrokecolor{dialinecolor}
\pgfpathmoveto{\pgfpoint{5.531250\du}{12.553100\du}}
\pgfpathcurveto{\pgfpoint{5.531394\du}{12.578100\du}}{\pgfpoint{5.506538\du}{12.603243\du}}{\pgfpoint{5.481538\du}{12.603387\du}}
\pgfpathcurveto{\pgfpoint{5.456539\du}{12.603531\du}}{\pgfpoint{5.431395\du}{12.578675\du}}{\pgfpoint{5.431252\du}{12.553675\du}}
\pgfpathcurveto{\pgfpoint{5.431108\du}{12.528676\du}}{\pgfpoint{5.455964\du}{12.503532\du}}{\pgfpoint{5.480963\du}{12.503388\du}}
\pgfpathcurveto{\pgfpoint{5.505963\du}{12.503245\du}}{\pgfpoint{5.531106\du}{12.528100\du}}{\pgfpoint{5.531250\du}{12.553100\du}}
\pgfusepath{stroke}
\pgfsetlinewidth{0.050000\du}
\pgfsetdash{}{0pt}
\pgfsetdash{}{0pt}
\pgfsetmiterjoin
\pgfsetbuttcap
{
\definecolor{dialinecolor}{rgb}{0.000000, 0.000000, 0.000000}
\pgfsetfillcolor{dialinecolor}
{\pgfsetcornersarced{\pgfpoint{0.000000\du}{0.000000\du}}\definecolor{dialinecolor}{rgb}{0.000000, 0.000000, 0.000000}
\pgfsetstrokecolor{dialinecolor}
\draw (5.468750\du,12.550000\du)--(5.968750\du,12.550000\du)--(5.968750\du,12.850000\du)--(5.468750\du,12.850000\du)--(5.468750\du,13.150000\du)--(5.968750\du,13.150000\du)--(5.968750\du,13.550000\du);
}}
\end{tikzpicture}

%% file: figure/angles.tex
\ifx\du\undefined
  \newlength{\du}
\fi
\setlength{\du}{15\unitlength}
\begin{tikzpicture}
\pgftransformxscale{1.000000}
\pgftransformyscale{-1.000000}
\definecolor{dialinecolor}{rgb}{0.000000, 0.000000, 0.000000}
\pgfsetstrokecolor{dialinecolor}
\definecolor{dialinecolor}{rgb}{1.000000, 1.000000, 1.000000}
\pgfsetfillcolor{dialinecolor}
\pgfsetlinewidth{0.050000\du}
\pgfsetdash{}{0pt}
\pgfsetdash{}{0pt}
\pgfsetbuttcap
{
\definecolor{dialinecolor}{rgb}{0.000000, 0.000000, 0.000000}
\pgfsetfillcolor{dialinecolor}
}
\definecolor{dialinecolor}{rgb}{0.000000, 0.000000, 0.000000}
\pgfsetstrokecolor{dialinecolor}
\draw (1.400000\du,10.400000\du)--(4.600000\du,13.600000\du);
\pgfsetlinewidth{0.050000\du}
\pgfsetdash{}{0pt}
\pgfsetmiterjoin
\pgfsetbuttcap
\definecolor{dialinecolor}{rgb}{0.000000, 0.000000, 0.000000}
\pgfsetfillcolor{dialinecolor}
\pgfpathmoveto{\pgfpoint{1.400000\du}{10.400000\du}}
\pgfpathcurveto{\pgfpoint{1.435355\du}{10.364645\du}}{\pgfpoint{1.506066\du}{10.364645\du}}{\pgfpoint{1.541421\du}{10.400000\du}}
\pgfpathcurveto{\pgfpoint{1.576777\du}{10.435355\du}}{\pgfpoint{1.576777\du}{10.506066\du}}{\pgfpoint{1.541421\du}{10.541421\du}}
\pgfpathcurveto{\pgfpoint{1.506066\du}{10.576777\du}}{\pgfpoint{1.435355\du}{10.576777\du}}{\pgfpoint{1.400000\du}{10.541421\du}}
\pgfpathcurveto{\pgfpoint{1.364645\du}{10.506066\du}}{\pgfpoint{1.364645\du}{10.435355\du}}{\pgfpoint{1.400000\du}{10.400000\du}}
\pgfusepath{fill}
\definecolor{dialinecolor}{rgb}{0.000000, 0.000000, 0.000000}
\pgfsetstrokecolor{dialinecolor}
\pgfpathmoveto{\pgfpoint{1.400000\du}{10.400000\du}}
\pgfpathcurveto{\pgfpoint{1.435355\du}{10.364645\du}}{\pgfpoint{1.506066\du}{10.364645\du}}{\pgfpoint{1.541421\du}{10.400000\du}}
\pgfpathcurveto{\pgfpoint{1.576777\du}{10.435355\du}}{\pgfpoint{1.576777\du}{10.506066\du}}{\pgfpoint{1.541421\du}{10.541421\du}}
\pgfpathcurveto{\pgfpoint{1.506066\du}{10.576777\du}}{\pgfpoint{1.435355\du}{10.576777\du}}{\pgfpoint{1.400000\du}{10.541421\du}}
\pgfpathcurveto{\pgfpoint{1.364645\du}{10.506066\du}}{\pgfpoint{1.364645\du}{10.435355\du}}{\pgfpoint{1.400000\du}{10.400000\du}}
\pgfusepath{stroke}
\pgfsetlinewidth{0.050000\du}
\pgfsetdash{}{0pt}
\pgfsetmiterjoin
\pgfsetbuttcap
\definecolor{dialinecolor}{rgb}{0.000000, 0.000000, 0.000000}
\pgfsetfillcolor{dialinecolor}
\pgfpathmoveto{\pgfpoint{4.600000\du}{13.600000\du}}
\pgfpathcurveto{\pgfpoint{4.546967\du}{13.653033\du}}{\pgfpoint{4.440901\du}{13.653033\du}}{\pgfpoint{4.387868\du}{13.600000\du}}
\pgfpathcurveto{\pgfpoint{4.334835\du}{13.546967\du}}{\pgfpoint{4.334835\du}{13.440901\du}}{\pgfpoint{4.387868\du}{13.387868\du}}
\pgfpathcurveto{\pgfpoint{4.440901\du}{13.334835\du}}{\pgfpoint{4.546967\du}{13.334835\du}}{\pgfpoint{4.600000\du}{13.387868\du}}
\pgfpathcurveto{\pgfpoint{4.653033\du}{13.440901\du}}{\pgfpoint{4.653033\du}{13.546967\du}}{\pgfpoint{4.600000\du}{13.600000\du}}
\pgfusepath{fill}
\definecolor{dialinecolor}{rgb}{0.000000, 0.000000, 0.000000}
\pgfsetstrokecolor{dialinecolor}
\pgfpathmoveto{\pgfpoint{4.600000\du}{13.600000\du}}
\pgfpathcurveto{\pgfpoint{4.546967\du}{13.653033\du}}{\pgfpoint{4.440901\du}{13.653033\du}}{\pgfpoint{4.387868\du}{13.600000\du}}
\pgfpathcurveto{\pgfpoint{4.334835\du}{13.546967\du}}{\pgfpoint{4.334835\du}{13.440901\du}}{\pgfpoint{4.387868\du}{13.387868\du}}
\pgfpathcurveto{\pgfpoint{4.440901\du}{13.334835\du}}{\pgfpoint{4.546967\du}{13.334835\du}}{\pgfpoint{4.600000\du}{13.387868\du}}
\pgfpathcurveto{\pgfpoint{4.653033\du}{13.440901\du}}{\pgfpoint{4.653033\du}{13.546967\du}}{\pgfpoint{4.600000\du}{13.600000\du}}
\pgfusepath{stroke}
\pgfsetlinewidth{0.050000\du}
\pgfsetdash{}{0pt}
\pgfsetdash{}{0pt}
\pgfsetbuttcap
{
\definecolor{dialinecolor}{rgb}{0.000000, 0.000000, 0.000000}
\pgfsetfillcolor{dialinecolor}
\definecolor{dialinecolor}{rgb}{0.000000, 0.000000, 0.000000}
\pgfsetstrokecolor{dialinecolor}
\draw (4.400000\du,13.400000\du)--(7.000000\du,15.600000\du);
}
\definecolor{dialinecolor}{rgb}{0.000000, 0.000000, 0.000000}
\pgfsetstrokecolor{dialinecolor}
\draw (4.400000\du,13.400000\du)--(7.000000\du,15.600000\du);
\pgfsetlinewidth{0.050000\du}
\pgfsetdash{}{0pt}
\pgfsetmiterjoin
\pgfsetbuttcap
\definecolor{dialinecolor}{rgb}{0.000000, 0.000000, 0.000000}
\pgfsetfillcolor{dialinecolor}
\pgfpathmoveto{\pgfpoint{7.000000\du}{15.600000\du}}
\pgfpathcurveto{\pgfpoint{6.967703\du}{15.638169\du}}{\pgfpoint{6.897236\du}{15.644042\du}}{\pgfpoint{6.859067\du}{15.611744\du}}
\pgfpathcurveto{\pgfpoint{6.820898\du}{15.579447\du}}{\pgfpoint{6.815026\du}{15.508981\du}}{\pgfpoint{6.847323\du}{15.470812\du}}
\pgfpathcurveto{\pgfpoint{6.879620\du}{15.432642\du}}{\pgfpoint{6.950086\du}{15.426770\du}}{\pgfpoint{6.988256\du}{15.459067\du}}
\pgfpathcurveto{\pgfpoint{7.026425\du}{15.491364\du}}{\pgfpoint{7.032297\du}{15.561831\du}}{\pgfpoint{7.000000\du}{15.600000\du}}
\pgfusepath{fill}
\definecolor{dialinecolor}{rgb}{0.000000, 0.000000, 0.000000}
\pgfsetstrokecolor{dialinecolor}
\pgfpathmoveto{\pgfpoint{7.000000\du}{15.600000\du}}
\pgfpathcurveto{\pgfpoint{6.967703\du}{15.638169\du}}{\pgfpoint{6.897236\du}{15.644042\du}}{\pgfpoint{6.859067\du}{15.611744\du}}
\pgfpathcurveto{\pgfpoint{6.820898\du}{15.579447\du}}{\pgfpoint{6.815026\du}{15.508981\du}}{\pgfpoint{6.847323\du}{15.470812\du}}
\pgfpathcurveto{\pgfpoint{6.879620\du}{15.432642\du}}{\pgfpoint{6.950086\du}{15.426770\du}}{\pgfpoint{6.988256\du}{15.459067\du}}
\pgfpathcurveto{\pgfpoint{7.026425\du}{15.491364\du}}{\pgfpoint{7.032297\du}{15.561831\du}}{\pgfpoint{7.000000\du}{15.600000\du}}
\pgfusepath{stroke}
\pgfsetlinewidth{0.050000\du}
\pgfsetdash{}{0pt}
\pgfsetdash{}{0pt}
\pgfsetbuttcap
{
\definecolor{dialinecolor}{rgb}{0.000000, 0.000000, 0.000000}
\pgfsetfillcolor{dialinecolor}
\definecolor{dialinecolor}{rgb}{0.000000, 0.000000, 0.000000}
\pgfsetstrokecolor{dialinecolor}
\draw (4.400000\du,13.400000\du)--(3.600000\du,16.600000\du);
}
\definecolor{dialinecolor}{rgb}{0.000000, 0.000000, 0.000000}
\pgfsetstrokecolor{dialinecolor}
\draw (4.400000\du,13.400000\du)--(3.600000\du,16.600000\du);
\pgfsetlinewidth{0.050000\du}
\pgfsetdash{}{0pt}
\pgfsetmiterjoin
\pgfsetbuttcap
\definecolor{dialinecolor}{rgb}{0.000000, 0.000000, 0.000000}
\pgfsetfillcolor{dialinecolor}
\pgfpathmoveto{\pgfpoint{3.600000\du}{16.600000\du}}
\pgfpathcurveto{\pgfpoint{3.551493\du}{16.587873\du}}{\pgfpoint{3.515113\du}{16.527239\du}}{\pgfpoint{3.527239\du}{16.478732\du}}
\pgfpathcurveto{\pgfpoint{3.539366\du}{16.430225\du}}{\pgfpoint{3.600000\du}{16.393845\du}}{\pgfpoint{3.648507\du}{16.405971\du}}
\pgfpathcurveto{\pgfpoint{3.697014\du}{16.418098\du}}{\pgfpoint{3.733395\du}{16.478732\du}}{\pgfpoint{3.721268\du}{16.527239\du}}
\pgfpathcurveto{\pgfpoint{3.709141\du}{16.575746\du}}{\pgfpoint{3.648507\du}{16.612127\du}}{\pgfpoint{3.600000\du}{16.600000\du}}
\pgfusepath{fill}
\definecolor{dialinecolor}{rgb}{0.000000, 0.000000, 0.000000}
\pgfsetstrokecolor{dialinecolor}
\pgfpathmoveto{\pgfpoint{3.600000\du}{16.600000\du}}
\pgfpathcurveto{\pgfpoint{3.551493\du}{16.587873\du}}{\pgfpoint{3.515113\du}{16.527239\du}}{\pgfpoint{3.527239\du}{16.478732\du}}
\pgfpathcurveto{\pgfpoint{3.539366\du}{16.430225\du}}{\pgfpoint{3.600000\du}{16.393845\du}}{\pgfpoint{3.648507\du}{16.405971\du}}
\pgfpathcurveto{\pgfpoint{3.697014\du}{16.418098\du}}{\pgfpoint{3.733395\du}{16.478732\du}}{\pgfpoint{3.721268\du}{16.527239\du}}
\pgfpathcurveto{\pgfpoint{3.709141\du}{16.575746\du}}{\pgfpoint{3.648507\du}{16.612127\du}}{\pgfpoint{3.600000\du}{16.600000\du}}
\pgfusepath{stroke}
\definecolor{dialinecolor}{rgb}{0.000000, 0.000000, 0.000000}
\pgfsetstrokecolor{dialinecolor}
\node[anchor=west] at (1.700000\du,10.000000\du){$\ell$};
\pgfsetlinewidth{0.050000\du}
\pgfsetdash{{\pgflinewidth}{0.200000\du}}{0cm}
\pgfsetdash{{\pgflinewidth}{0.200000\du}}{0cm}
\pgfsetbuttcap
{
\definecolor{dialinecolor}{rgb}{0.000000, 0.000000, 0.000000}
\pgfsetfillcolor{dialinecolor}
\definecolor{dialinecolor}{rgb}{0.000000, 0.000000, 0.000000}
\pgfsetstrokecolor{dialinecolor}
\pgfpathmoveto{\pgfpoint{5.099964\du}{14.000075\du}}
\pgfpatharc{26}{-120}{0.851245\du and 0.851245\du}
\pgfusepath{stroke}
}
\pgfsetlinewidth{0.050000\du}
\pgfsetdash{{\pgflinewidth}{0.200000\du}}{0cm}
\pgfsetdash{{\pgflinewidth}{0.200000\du}}{0cm}
\pgfsetbuttcap
{
\definecolor{dialinecolor}{rgb}{0.000000, 0.000000, 0.000000}
\pgfsetfillcolor{dialinecolor}
\definecolor{dialinecolor}{rgb}{0.000000, 0.000000, 0.000000}
\pgfsetstrokecolor{dialinecolor}
\pgfpathmoveto{\pgfpoint{3.900029\du}{12.899971\du}}
\pgfpatharc{225}{122}{1.028525\du and 1.028525\du}
\pgfusepath{stroke}
}
\definecolor{dialinecolor}{rgb}{0.000000, 0.000000, 0.000000}
\pgfsetstrokecolor{dialinecolor}
\node[anchor=west] at (4.800000\du,16.000000\du){$u$};
\definecolor{dialinecolor}{rgb}{0.000000, 0.000000, 0.000000}
\pgfsetstrokecolor{dialinecolor}
\node[anchor=west] at (5.400000\du,16.000000\du){};
\pgfsetlinewidth{0.000000\du}
\pgfsetdash{}{0pt}
\pgfsetdash{}{0pt}
\pgfsetbuttcap
{
\definecolor{dialinecolor}{rgb}{0.000000, 0.000000, 0.000000}
\pgfsetfillcolor{dialinecolor}
\pgfsetarrowsstart{stealth}
\definecolor{dialinecolor}{rgb}{0.000000, 0.000000, 0.000000}
\pgfsetstrokecolor{dialinecolor}
\draw (4.600000\du,13.800000\du)--(5.200000\du,15.600000\du);
}
\definecolor{dialinecolor}{rgb}{0.000000, 0.000000, 0.000000}
\pgfsetstrokecolor{dialinecolor}
\node[anchor=west] at (5.000000\du,12.200000\du){$a_R(\ell)$};
\definecolor{dialinecolor}{rgb}{0.000000, 0.000000, 0.000000}
\pgfsetstrokecolor{dialinecolor}
\node[anchor=west] at (1.200000\du,14.400000\du){$a_L(\ell)$};
\definecolor{dialinecolor}{rgb}{0.000000, 0.000000, 0.000000}
\pgfsetstrokecolor{dialinecolor}
\node[anchor=west] at (5.600000\du,12.000000\du){};
\definecolor{dialinecolor}{rgb}{0.000000, 0.000000, 0.000000}
\pgfsetstrokecolor{dialinecolor}
\node[anchor=west] at (6.400000\du,12.000000\du){};
\definecolor{dialinecolor}{rgb}{0.000000, 0.000000, 0.000000}
\pgfsetstrokecolor{dialinecolor}
\node[anchor=west] at (2.700000\du,14.400000\du){};
\definecolor{dialinecolor}{rgb}{0.000000, 0.000000, 0.000000}
\pgfsetstrokecolor{dialinecolor}
\node[anchor=west] at (6.000000\du,12.000000\du){};
\end{tikzpicture}

%% file: figure/lattice.tex
\ifx\du\undefined
  \newlength{\du}
\fi
\setlength{\du}{15\unitlength}
\begin{tikzpicture}
\pgftransformxscale{1.000000}
\pgftransformyscale{-1.000000}
\definecolor{dialinecolor}{rgb}{0.000000, 0.000000, 0.000000}
\pgfsetstrokecolor{dialinecolor}
\definecolor{dialinecolor}{rgb}{1.000000, 1.000000, 1.000000}
\pgfsetfillcolor{dialinecolor}
\pgfsetlinewidth{0.000000\du}
\pgfsetdash{}{0pt}
\pgfsetdash{}{0pt}
\pgfsetbuttcap
{
\definecolor{dialinecolor}{rgb}{0.000000, 0.000000, 0.000000}
\pgfsetfillcolor{dialinecolor}
\definecolor{dialinecolor}{rgb}{0.000000, 0.000000, 0.000000}
\pgfsetstrokecolor{dialinecolor}
\draw (1.000000\du,17.000000\du)--(6.000000\du,12.000000\du);
}
\pgfsetlinewidth{0.000000\du}
\pgfsetdash{}{0pt}
\pgfsetdash{}{0pt}
\pgfsetbuttcap
{
\definecolor{dialinecolor}{rgb}{0.000000, 0.000000, 0.000000}
\pgfsetfillcolor{dialinecolor}
\definecolor{dialinecolor}{rgb}{0.000000, 0.000000, 0.000000}
\pgfsetstrokecolor{dialinecolor}
\draw (2.800000\du,17.000000\du)--(7.800000\du,12.000000\du);
}
\pgfsetlinewidth{0.000000\du}
\pgfsetdash{}{0pt}
\pgfsetdash{}{0pt}
\pgfsetbuttcap
{
\definecolor{dialinecolor}{rgb}{0.000000, 0.000000, 0.000000}
\pgfsetfillcolor{dialinecolor}
\definecolor{dialinecolor}{rgb}{0.000000, 0.000000, 0.000000}
\pgfsetstrokecolor{dialinecolor}
\draw (4.600000\du,17.000000\du)--(9.600000\du,12.000000\du);
}
\pgfsetlinewidth{0.000000\du}
\pgfsetdash{}{0pt}
\pgfsetdash{}{0pt}
\pgfsetbuttcap
{
\definecolor{dialinecolor}{rgb}{0.000000, 0.000000, 0.000000}
\pgfsetfillcolor{dialinecolor}
\definecolor{dialinecolor}{rgb}{0.000000, 0.000000, 0.000000}
\pgfsetstrokecolor{dialinecolor}
\draw (6.400000\du,17.000000\du)--(11.400000\du,12.000000\du);
}
\pgfsetlinewidth{0.000000\du}
\pgfsetdash{}{0pt}
\pgfsetdash{}{0pt}
\pgfsetbuttcap
{
\definecolor{dialinecolor}{rgb}{0.000000, 0.000000, 0.000000}
\pgfsetfillcolor{dialinecolor}
\definecolor{dialinecolor}{rgb}{0.000000, 0.000000, 0.000000}
\pgfsetstrokecolor{dialinecolor}
\draw (8.200000\du,17.000000\du)--(13.200000\du,12.000000\du);
}
\pgfsetlinewidth{0.000000\du}
\pgfsetdash{}{0pt}
\pgfsetdash{}{0pt}
\pgfsetbuttcap
{
\definecolor{dialinecolor}{rgb}{0.000000, 0.000000, 0.000000}
\pgfsetfillcolor{dialinecolor}
\definecolor{dialinecolor}{rgb}{0.000000, 0.000000, 0.000000}
\pgfsetstrokecolor{dialinecolor}
\draw (1.000000\du,17.000000\du)--(8.200000\du,17.000000\du);
}
\pgfsetlinewidth{0.000000\du}
\pgfsetdash{}{0pt}
\pgfsetdash{}{0pt}
\pgfsetbuttcap
{
\definecolor{dialinecolor}{rgb}{0.000000, 0.000000, 0.000000}
\pgfsetfillcolor{dialinecolor}
\definecolor{dialinecolor}{rgb}{0.000000, 0.000000, 0.000000}
\pgfsetstrokecolor{dialinecolor}
\draw (2.000000\du,16.000000\du)--(9.200000\du,16.000000\du);
}
\pgfsetlinewidth{0.000000\du}
\pgfsetdash{}{0pt}
\pgfsetdash{}{0pt}
\pgfsetbuttcap
{
\definecolor{dialinecolor}{rgb}{0.000000, 0.000000, 0.000000}
\pgfsetfillcolor{dialinecolor}
\definecolor{dialinecolor}{rgb}{0.000000, 0.000000, 0.000000}
\pgfsetstrokecolor{dialinecolor}
\draw (3.000000\du,15.000000\du)--(10.200000\du,15.000000\du);
}
\pgfsetlinewidth{0.000000\du}
\pgfsetdash{}{0pt}
\pgfsetdash{}{0pt}
\pgfsetbuttcap
{
\definecolor{dialinecolor}{rgb}{0.000000, 0.000000, 0.000000}
\pgfsetfillcolor{dialinecolor}
\definecolor{dialinecolor}{rgb}{0.000000, 0.000000, 0.000000}
\pgfsetstrokecolor{dialinecolor}
\draw (4.000000\du,14.000000\du)--(11.200000\du,14.000000\du);
}
\pgfsetlinewidth{0.000000\du}
\pgfsetdash{}{0pt}
\pgfsetdash{}{0pt}
\pgfsetbuttcap
{
\definecolor{dialinecolor}{rgb}{0.000000, 0.000000, 0.000000}
\pgfsetfillcolor{dialinecolor}
\definecolor{dialinecolor}{rgb}{0.000000, 0.000000, 0.000000}
\pgfsetstrokecolor{dialinecolor}
\draw (5.000000\du,13.000000\du)--(12.200000\du,13.000000\du);
}
\pgfsetlinewidth{0.000000\du}
\pgfsetdash{}{0pt}
\pgfsetdash{}{0pt}
\pgfsetbuttcap
{
\definecolor{dialinecolor}{rgb}{0.000000, 0.000000, 0.000000}
\pgfsetfillcolor{dialinecolor}
\definecolor{dialinecolor}{rgb}{0.000000, 0.000000, 0.000000}
\pgfsetstrokecolor{dialinecolor}
\draw (6.000000\du,12.000000\du)--(13.200000\du,12.000000\du);
}
\definecolor{dialinecolor}{rgb}{0.000000, 0.000000, 0.000000}
\pgfsetstrokecolor{dialinecolor}
\node[anchor=west] at (1.312500\du,16.425000\du){$a$};
\definecolor{dialinecolor}{rgb}{0.000000, 0.000000, 0.000000}
\pgfsetstrokecolor{dialinecolor}
\node[anchor=west] at (1.818750\du,17.331250\du){$b$};
\definecolor{dialinecolor}{rgb}{0.000000, 0.000000, 0.000000}
\pgfsetstrokecolor{dialinecolor}
\node[anchor=west] at (2.000000\du,17.400000\du){};
\definecolor{dialinecolor}{rgb}{0.000000, 0.000000, 0.000000}
\pgfsetstrokecolor{dialinecolor}
\node[anchor=west] at (2.481250\du,16.425000\du){$c$};
\definecolor{dialinecolor}{rgb}{0.000000, 0.000000, 0.000000}
\pgfsetstrokecolor{dialinecolor}
\node[anchor=west] at (3.000000\du,16.400000\du){};
\pgfsetlinewidth{0.000000\du}
\pgfsetdash{}{0pt}
\pgfsetdash{}{0pt}
\pgfsetbuttcap
{
\definecolor{dialinecolor}{rgb}{0.000000, 0.000000, 0.000000}
\pgfsetfillcolor{dialinecolor}
\definecolor{dialinecolor}{rgb}{0.000000, 0.000000, 0.000000}
\pgfsetstrokecolor{dialinecolor}
\draw (2.000000\du,16.000000\du)--(2.800000\du,17.000000\du);
}
\pgfsetlinewidth{0.000000\du}
\pgfsetdash{}{0pt}
\pgfsetdash{}{0pt}
\pgfsetbuttcap
{
\definecolor{dialinecolor}{rgb}{0.000000, 0.000000, 0.000000}
\pgfsetfillcolor{dialinecolor}
\definecolor{dialinecolor}{rgb}{0.000000, 0.000000, 0.000000}
\pgfsetstrokecolor{dialinecolor}
\draw (3.000000\du,15.000000\du)--(4.600000\du,17.000000\du);
}
\pgfsetlinewidth{0.000000\du}
\pgfsetdash{}{0pt}
\pgfsetdash{}{0pt}
\pgfsetbuttcap
{
\definecolor{dialinecolor}{rgb}{0.000000, 0.000000, 0.000000}
\pgfsetfillcolor{dialinecolor}
\definecolor{dialinecolor}{rgb}{0.000000, 0.000000, 0.000000}
\pgfsetstrokecolor{dialinecolor}
\draw (4.000000\du,14.000000\du)--(6.400000\du,17.000000\du);
}
\pgfsetlinewidth{0.000000\du}
\pgfsetdash{}{0pt}
\pgfsetdash{}{0pt}
\pgfsetbuttcap
{
\definecolor{dialinecolor}{rgb}{0.000000, 0.000000, 0.000000}
\pgfsetfillcolor{dialinecolor}
\definecolor{dialinecolor}{rgb}{0.000000, 0.000000, 0.000000}
\pgfsetstrokecolor{dialinecolor}
\draw (5.000000\du,13.000000\du)--(8.200000\du,17.000000\du);
}
\pgfsetlinewidth{0.000000\du}
\pgfsetdash{}{0pt}
\pgfsetdash{}{0pt}
\pgfsetbuttcap
{
\definecolor{dialinecolor}{rgb}{0.000000, 0.000000, 0.000000}
\pgfsetfillcolor{dialinecolor}
\definecolor{dialinecolor}{rgb}{0.000000, 0.000000, 0.000000}
\pgfsetstrokecolor{dialinecolor}
\draw (6.000000\du,12.000000\du)--(9.200000\du,16.000000\du);
}
\pgfsetlinewidth{0.000000\du}
\pgfsetdash{}{0pt}
\pgfsetdash{}{0pt}
\pgfsetbuttcap
{
\definecolor{dialinecolor}{rgb}{0.000000, 0.000000, 0.000000}
\pgfsetfillcolor{dialinecolor}
\definecolor{dialinecolor}{rgb}{0.000000, 0.000000, 0.000000}
\pgfsetstrokecolor{dialinecolor}
\draw (7.800000\du,12.000000\du)--(10.200000\du,15.000000\du);
}
\pgfsetlinewidth{0.000000\du}
\pgfsetdash{}{0pt}
\pgfsetdash{}{0pt}
\pgfsetbuttcap
{
\definecolor{dialinecolor}{rgb}{0.000000, 0.000000, 0.000000}
\pgfsetfillcolor{dialinecolor}
\definecolor{dialinecolor}{rgb}{0.000000, 0.000000, 0.000000}
\pgfsetstrokecolor{dialinecolor}
\draw (9.600000\du,12.000000\du)--(11.200000\du,14.000000\du);
}
\pgfsetlinewidth{0.000000\du}
\pgfsetdash{}{0pt}
\pgfsetdash{}{0pt}
\pgfsetbuttcap
{
\definecolor{dialinecolor}{rgb}{0.000000, 0.000000, 0.000000}
\pgfsetfillcolor{dialinecolor}
\definecolor{dialinecolor}{rgb}{0.000000, 0.000000, 0.000000}
\pgfsetstrokecolor{dialinecolor}
\draw (11.400000\du,12.000000\du)--(12.200000\du,13.000000\du);
}
\pgfsetlinewidth{0.050000\du}
\pgfsetdash{}{0pt}
\pgfsetdash{}{0pt}
\pgfsetmiterjoin
\pgfsetbuttcap
\definecolor{dialinecolor}{rgb}{0.847059, 0.898039, 0.898039}
\pgfsetfillcolor{dialinecolor}
\fill (3.600000\du,16.800000\du)--(4.800000\du,17.400000\du)--(5.800000\du,15.800000\du)--(6.000000\du,14.600000\du)--(7.800000\du,15.400000\du)--(9.000000\du,14.800000\du)--(9.800000\du,13.200000\du)--(8.800000\du,12.800000\du)--(7.400000\du,13.600000\du)--(5.400000\du,14.000000\du)--(4.800000\du,15.400000\du)--(4.000000\du,15.200000\du)--cycle;
\definecolor{dialinecolor}{rgb}{0.000000, 0.000000, 0.000000}
\pgfsetstrokecolor{dialinecolor}
\draw (3.600000\du,16.800000\du)--(4.800000\du,17.400000\du)--(5.800000\du,15.800000\du)--(6.000000\du,14.600000\du)--(7.800000\du,15.400000\du)--(9.000000\du,14.800000\du)--(9.800000\du,13.200000\du)--(8.800000\du,12.800000\du)--(7.400000\du,13.600000\du)--(5.400000\du,14.000000\du)--(4.800000\du,15.400000\du)--(4.000000\du,15.200000\du)--cycle;
\definecolor{dialinecolor}{rgb}{0.000000, 0.000000, 0.000000}
\pgfsetstrokecolor{dialinecolor}
\node[anchor=west] at (4.600000\du,16.000000\du){$m$};
\definecolor{dialinecolor}{rgb}{0.000000, 0.000000, 0.000000}
\pgfsetstrokecolor{dialinecolor}
\node[anchor=west] at (8.206250\du,14.231250\du){$m'$};
\pgfsetlinewidth{0.050000\du}
\pgfsetdash{{1.000000\du}{1.000000\du}}{0\du}
\pgfsetdash{{1.000000\du}{1.000000\du}}{0\du}
\pgfsetbuttcap
{
\definecolor{dialinecolor}{rgb}{0.000000, 0.000000, 0.000000}
\pgfsetfillcolor{dialinecolor}
\definecolor{dialinecolor}{rgb}{0.000000, 0.000000, 0.000000}
\pgfsetstrokecolor{dialinecolor}
\draw (6.700000\du,15.100000\du)--(6.700000\du,15.100000\du);
}
\pgfsetlinewidth{0.000000\du}
\pgfsetdash{}{0pt}
\pgfsetdash{}{0pt}
\pgfsetbuttcap
{
\definecolor{dialinecolor}{rgb}{0.000000, 0.000000, 0.000000}
\pgfsetfillcolor{dialinecolor}
\pgfsetarrowsend{to}
\definecolor{dialinecolor}{rgb}{0.000000, 0.000000, 0.000000}
\pgfsetstrokecolor{dialinecolor}
\draw (4.500000\du,15.900000\du)--(3.900000\du,16.000000\du);
}
\pgfsetlinewidth{0.000000\du}
\pgfsetdash{}{0pt}
\pgfsetdash{}{0pt}
\pgfsetbuttcap
\pgfsetmiterjoin
\pgfsetlinewidth{0.000000\du}
\pgfsetbuttcap
\pgfsetmiterjoin
\pgfsetdash{}{0pt}
\definecolor{dialinecolor}{rgb}{0.000000, 0.000000, 0.000000}
\pgfsetfillcolor{dialinecolor}
\pgfpathellipse{\pgfpoint{3.800000\du}{16.012500\du}}{\pgfpoint{0.050000\du}{0\du}}{\pgfpoint{0\du}{0.050000\du}}
\pgfusepath{fill}
\definecolor{dialinecolor}{rgb}{0.000000, 0.000000, 0.000000}
\pgfsetstrokecolor{dialinecolor}
\pgfpathellipse{\pgfpoint{3.800000\du}{16.012500\du}}{\pgfpoint{0.050000\du}{0\du}}{\pgfpoint{0\du}{0.050000\du}}
\pgfusepath{stroke}
\pgfsetbuttcap
\pgfsetmiterjoin
\pgfsetdash{}{0pt}
\definecolor{dialinecolor}{rgb}{0.000000, 0.000000, 0.000000}
\pgfsetstrokecolor{dialinecolor}
\pgfpathellipse{\pgfpoint{3.800000\du}{16.012500\du}}{\pgfpoint{0.050000\du}{0\du}}{\pgfpoint{0\du}{0.050000\du}}
\pgfusepath{stroke}
\pgfsetlinewidth{0.000000\du}
\pgfsetdash{}{0pt}
\pgfsetdash{}{0pt}
\pgfsetbuttcap
\pgfsetmiterjoin
\pgfsetlinewidth{0.000000\du}
\pgfsetbuttcap
\pgfsetmiterjoin
\pgfsetdash{}{0pt}
\definecolor{dialinecolor}{rgb}{0.000000, 0.000000, 0.000000}
\pgfsetfillcolor{dialinecolor}
\pgfpathellipse{\pgfpoint{9.397188\du}{14.010625\du}}{\pgfpoint{0.050000\du}{0\du}}{\pgfpoint{0\du}{0.050000\du}}
\pgfusepath{fill}
\definecolor{dialinecolor}{rgb}{0.000000, 0.000000, 0.000000}
\pgfsetstrokecolor{dialinecolor}
\pgfpathellipse{\pgfpoint{9.397188\du}{14.010625\du}}{\pgfpoint{0.050000\du}{0\du}}{\pgfpoint{0\du}{0.050000\du}}
\pgfusepath{stroke}
\pgfsetbuttcap
\pgfsetmiterjoin
\pgfsetdash{}{0pt}
\definecolor{dialinecolor}{rgb}{0.000000, 0.000000, 0.000000}
\pgfsetstrokecolor{dialinecolor}
\pgfpathellipse{\pgfpoint{9.397188\du}{14.010625\du}}{\pgfpoint{0.050000\du}{0\du}}{\pgfpoint{0\du}{0.050000\du}}
\pgfusepath{stroke}
\pgfsetlinewidth{0.000000\du}
\pgfsetdash{}{0pt}
\pgfsetdash{}{0pt}
\pgfsetbuttcap
{
\definecolor{dialinecolor}{rgb}{0.000000, 0.000000, 0.000000}
\pgfsetfillcolor{dialinecolor}
\pgfsetarrowsend{to}
\definecolor{dialinecolor}{rgb}{0.000000, 0.000000, 0.000000}
\pgfsetstrokecolor{dialinecolor}
\draw (8.698611\du,14.130010\du)--(9.234687\du,14.023125\du);
}
\definecolor{dialinecolor}{rgb}{0.000000, 0.000000, 0.000000}
\pgfsetstrokecolor{dialinecolor}
\node[anchor=west] at (3.278438\du,15.948125\du){$e$};
\definecolor{dialinecolor}{rgb}{0.000000, 0.000000, 0.000000}
\pgfsetstrokecolor{dialinecolor}
\node[anchor=west] at (9.728438\du,14.329375\du){$e'$};
\pgfsetlinewidth{0.020000\du}
\pgfsetdash{{1.000000\du}{1.000000\du}}{0\du}
\pgfsetdash{{1.000000\du}{1.000000\du}}{0\du}
\pgfsetmiterjoin
\pgfsetbuttcap
{
\definecolor{dialinecolor}{rgb}{0.000000, 0.000000, 0.000000}
\pgfsetfillcolor{dialinecolor}
\definecolor{dialinecolor}{rgb}{0.000000, 0.000000, 0.000000}
\pgfsetstrokecolor{dialinecolor}
\pgfpathmoveto{\pgfpoint{4.000000\du}{15.200000\du}}
\pgfpathcurveto{\pgfpoint{3.768750\du}{15.111260\du}}{\pgfpoint{3.724450\du}{15.948760\du}}{\pgfpoint{3.587500\du}{15.948760\du}}
\pgfusepath{stroke}
}
\pgfsetlinewidth{0.020000\du}
\pgfsetdash{{1.000000\du}{1.000000\du}}{0\du}
\pgfsetdash{{1.000000\du}{1.000000\du}}{0\du}
\pgfsetmiterjoin
\pgfsetbuttcap
{
\definecolor{dialinecolor}{rgb}{0.000000, 0.000000, 0.000000}
\pgfsetfillcolor{dialinecolor}
\definecolor{dialinecolor}{rgb}{0.000000, 0.000000, 0.000000}
\pgfsetstrokecolor{dialinecolor}
\pgfpathmoveto{\pgfpoint{3.600000\du}{16.800000\du}}
\pgfpathcurveto{\pgfpoint{3.412500\du}{16.648760\du}}{\pgfpoint{3.687500\du}{16.017510\du}}{\pgfpoint{3.587500\du}{15.955010\du}}
\pgfusepath{stroke}
}
\pgfsetlinewidth{0.020000\du}
\pgfsetdash{{1.000000\du}{1.000000\du}}{0\du}
\pgfsetdash{{1.000000\du}{1.000000\du}}{0\du}
\pgfsetmiterjoin
\pgfsetbuttcap
{
\definecolor{dialinecolor}{rgb}{0.000000, 0.000000, 0.000000}
\pgfsetfillcolor{dialinecolor}
\definecolor{dialinecolor}{rgb}{0.000000, 0.000000, 0.000000}
\pgfsetstrokecolor{dialinecolor}
\pgfpathmoveto{\pgfpoint{9.825000\du}{13.206250\du}}
\pgfpathcurveto{\pgfpoint{10.031038\du}{13.305010\du}}{\pgfpoint{9.412288\du}{13.992510\du}}{\pgfpoint{9.674788\du}{14.180010\du}}
\pgfusepath{stroke}
}
\pgfsetlinewidth{0.020000\du}
\pgfsetdash{{1.000000\du}{1.000000\du}}{0\du}
\pgfsetdash{{1.000000\du}{1.000000\du}}{0\du}
\pgfsetmiterjoin
\pgfsetbuttcap
{
\definecolor{dialinecolor}{rgb}{0.000000, 0.000000, 0.000000}
\pgfsetfillcolor{dialinecolor}
\definecolor{dialinecolor}{rgb}{0.000000, 0.000000, 0.000000}
\pgfsetstrokecolor{dialinecolor}
\pgfpathmoveto{\pgfpoint{9.024788\du}{14.780010\du}}
\pgfpathcurveto{\pgfpoint{9.362288\du}{14.823760\du}}{\pgfpoint{9.406038\du}{14.017510\du}}{\pgfpoint{9.681038\du}{14.136260\du}}
\pgfusepath{stroke}
}
\pgfsetlinewidth{0.030000\du}
\pgfsetdash{}{0pt}
\pgfsetdash{}{0pt}
\pgfsetbuttcap
{
\definecolor{dialinecolor}{rgb}{0.000000, 0.000000, 0.000000}
\pgfsetfillcolor{dialinecolor}
\pgfsetarrowsstart{stealth}
\pgfsetarrowsend{stealth}
\definecolor{dialinecolor}{rgb}{0.000000, 0.000000, 0.000000}
\pgfsetstrokecolor{dialinecolor}
\draw (2.000000\du,15.998760\du)--(0.993750\du,16.998760\du);
}
\pgfsetlinewidth{0.030000\du}
\pgfsetdash{}{0pt}
\pgfsetdash{}{0pt}
\pgfsetbuttcap
{
\definecolor{dialinecolor}{rgb}{0.000000, 0.000000, 0.000000}
\pgfsetfillcolor{dialinecolor}
\pgfsetarrowsstart{stealth}
\pgfsetarrowsend{stealth}
\definecolor{dialinecolor}{rgb}{0.000000, 0.000000, 0.000000}
\pgfsetstrokecolor{dialinecolor}
\draw (2.006250\du,16.011260\du)--(2.800000\du,16.992510\du);
}
\pgfsetlinewidth{0.030000\du}
\pgfsetdash{}{0pt}
\pgfsetdash{}{0pt}
\pgfsetbuttcap
{
\definecolor{dialinecolor}{rgb}{0.000000, 0.000000, 0.000000}
\pgfsetfillcolor{dialinecolor}
\pgfsetarrowsstart{stealth}
\pgfsetarrowsend{stealth}
\definecolor{dialinecolor}{rgb}{0.000000, 0.000000, 0.000000}
\pgfsetstrokecolor{dialinecolor}
\draw (2.793750\du,16.986260\du)--(1.012500\du,16.992510\du);
}
\definecolor{dialinecolor}{rgb}{0.000000, 0.000000, 0.000000}
\pgfsetstrokecolor{dialinecolor}
\node[anchor=west] at (1.337500\du,16.392510\du){};
\pgfsetlinewidth{0.050000\du}
\pgfsetdash{}{0pt}
\pgfsetdash{}{0pt}
\pgfsetbuttcap
{
\definecolor{dialinecolor}{rgb}{0.000000, 0.000000, 0.000000}
\pgfsetfillcolor{dialinecolor}
\pgfsetarrowsend{stealth}
\definecolor{dialinecolor}{rgb}{0.000000, 0.000000, 0.000000}
\pgfsetstrokecolor{dialinecolor}
\draw (6.402219\du,17.011260\du)--(7.402219\du,16.011260\du);
}
\pgfsetlinewidth{0.050000\du}
\pgfsetdash{}{0pt}
\pgfsetdash{}{0pt}
\pgfsetbuttcap
{
\definecolor{dialinecolor}{rgb}{0.000000, 0.000000, 0.000000}
\pgfsetfillcolor{dialinecolor}
\pgfsetarrowsend{stealth}
\definecolor{dialinecolor}{rgb}{0.000000, 0.000000, 0.000000}
\pgfsetstrokecolor{dialinecolor}
\draw (6.411906\du,17.009385\du)--(8.208469\du,17.005010\du);
}
\definecolor{dialinecolor}{rgb}{0.000000, 0.000000, 0.000000}
\pgfsetstrokecolor{dialinecolor}
\node[anchor=west] at (6.441287\du,16.423760\du){};
\definecolor{dialinecolor}{rgb}{0.000000, 0.000000, 0.000000}
\pgfsetstrokecolor{dialinecolor}
\node[anchor=west] at (6.040938\du,16.460625\du){$\vec\{v_a\}$};
\definecolor{dialinecolor}{rgb}{0.000000, 0.000000, 0.000000}
\pgfsetstrokecolor{dialinecolor}
\node[anchor=west] at (6.956875\du,17.352500\du){$\vec\{v_b\}$};
\definecolor{dialinecolor}{rgb}{0.000000, 0.000000, 0.000000}
\pgfsetstrokecolor{dialinecolor}
\node[anchor=west] at (5.722537\du,17.280010\du){};
\definecolor{dialinecolor}{rgb}{0.000000, 0.000000, 0.000000}
\pgfsetstrokecolor{dialinecolor}
\node[anchor=west] at (7.185037\du,17.311260\du){};
\end{tikzpicture}

%% file: figure/Type2.tex
\ifx\du\undefined
  \newlength{\du}
\fi
\setlength{\du}{15\unitlength}
\begin{tikzpicture}
\pgftransformxscale{1.000000}
\pgftransformyscale{-1.000000}
\definecolor{dialinecolor}{rgb}{0.000000, 0.000000, 0.000000}
\pgfsetstrokecolor{dialinecolor}
\definecolor{dialinecolor}{rgb}{1.000000, 1.000000, 1.000000}
\pgfsetfillcolor{dialinecolor}
\pgfsetlinewidth{0.050000\du}
\pgfsetdash{{\pgflinewidth}{0.200000\du}}{0cm}
\pgfsetdash{{\pgflinewidth}{0.200000\du}}{0cm}
\pgfsetbuttcap
{
\definecolor{dialinecolor}{rgb}{0.000000, 0.000000, 0.000000}
\pgfsetfillcolor{dialinecolor}
\definecolor{dialinecolor}{rgb}{0.000000, 0.000000, 0.000000}
\pgfsetstrokecolor{dialinecolor}
\pgfpathmoveto{\pgfpoint{7.511207\du}{16.104691\du}}
\pgfpatharc{321}{156}{0.396224\du and 0.396224\du}
\pgfusepath{stroke}
}
\pgfsetlinewidth{0.050000\du}
\pgfsetdash{{\pgflinewidth}{0.200000\du}}{0cm}
\pgfsetdash{{\pgflinewidth}{0.200000\du}}{0cm}
\pgfsetbuttcap
{
\definecolor{dialinecolor}{rgb}{0.000000, 0.000000, 0.000000}
\pgfsetfillcolor{dialinecolor}
\definecolor{dialinecolor}{rgb}{0.000000, 0.000000, 0.000000}
\pgfsetstrokecolor{dialinecolor}
\pgfpathmoveto{\pgfpoint{8.254923\du}{16.879703\du}}
\pgfpatharc{64}{-112}{0.402975\du and 0.402975\du}
\pgfusepath{stroke}
}
\definecolor{dialinecolor}{rgb}{0.000000, 0.000000, 0.000000}
\pgfsetstrokecolor{dialinecolor}
\node[anchor=west] at (3.925000\du,18.637500\du){$v_1$};
\definecolor{dialinecolor}{rgb}{0.000000, 0.000000, 0.000000}
\pgfsetstrokecolor{dialinecolor}
\node[anchor=west] at (2.550000\du,17.375000\du){};
\pgfsetlinewidth{0.000000\du}
\pgfsetdash{}{0pt}
\pgfsetdash{}{0pt}
\pgfsetbuttcap
{
\definecolor{dialinecolor}{rgb}{0.000000, 0.000000, 0.000000}
\pgfsetfillcolor{dialinecolor}
\pgfsetarrowsstart{stealth}
\definecolor{dialinecolor}{rgb}{0.000000, 0.000000, 0.000000}
\pgfsetstrokecolor{dialinecolor}
\draw (8.370579\du,16.207812\du)--(8.214329\du,11.951562\du);
}
\definecolor{dialinecolor}{rgb}{0.000000, 0.000000, 0.000000}
\pgfsetstrokecolor{dialinecolor}
\node[anchor=west] at (4.200000\du,10.800000\du){$a_R(v_1)=a_L(v_2)=180^\circ$};
\definecolor{dialinecolor}{rgb}{0.000000, 0.000000, 0.000000}
\pgfsetstrokecolor{dialinecolor}
\node[anchor=west] at (2.625000\du,12.000000\du){};
\definecolor{dialinecolor}{rgb}{0.000000, 0.000000, 0.000000}
\pgfsetstrokecolor{dialinecolor}
\node[anchor=west] at (3.425000\du,12.000000\du){};
\definecolor{dialinecolor}{rgb}{0.000000, 0.000000, 0.000000}
\pgfsetstrokecolor{dialinecolor}
\node[anchor=west] at (-0.275000\du,14.400000\du){};
\pgfsetlinewidth{0.100000\du}
\pgfsetdash{{1.000000\du}{1.000000\du}}{0\du}
\pgfsetdash{{0.100000\du}{0.100000\du}}{0\du}
\pgfsetbuttcap
{
\definecolor{dialinecolor}{rgb}{0.000000, 0.000000, 0.000000}
\pgfsetfillcolor{dialinecolor}
}
\definecolor{dialinecolor}{rgb}{0.000000, 0.000000, 0.000000}
\pgfsetstrokecolor{dialinecolor}
\draw (6.025000\du,13.000000\du)--(9.025000\du,18.200000\du);
\pgfsetlinewidth{0.100000\du}
\pgfsetdash{}{0pt}
\pgfsetmiterjoin
\pgfsetbuttcap
\definecolor{dialinecolor}{rgb}{0.000000, 0.000000, 0.000000}
\pgfsetfillcolor{dialinecolor}
\pgfpathmoveto{\pgfpoint{6.025000\du}{13.000000\du}}
\pgfpathcurveto{\pgfpoint{6.068309\du}{12.975014\du}}{\pgfpoint{6.136605\du}{12.993337\du}}{\pgfpoint{6.161591\du}{13.036646\du}}
\pgfpathcurveto{\pgfpoint{6.186577\du}{13.079956\du}}{\pgfpoint{6.168254\du}{13.148251\du}}{\pgfpoint{6.124944\du}{13.173237\du}}
\pgfpathcurveto{\pgfpoint{6.081635\du}{13.198223\du}}{\pgfpoint{6.013340\du}{13.179900\du}}{\pgfpoint{5.988354\du}{13.136591\du}}
\pgfpathcurveto{\pgfpoint{5.963368\du}{13.093282\du}}{\pgfpoint{5.981691\du}{13.024986\du}}{\pgfpoint{6.025000\du}{13.000000\du}}
\pgfusepath{fill}
\definecolor{dialinecolor}{rgb}{0.000000, 0.000000, 0.000000}
\pgfsetstrokecolor{dialinecolor}
\pgfpathmoveto{\pgfpoint{6.025000\du}{13.000000\du}}
\pgfpathcurveto{\pgfpoint{6.068309\du}{12.975014\du}}{\pgfpoint{6.136605\du}{12.993337\du}}{\pgfpoint{6.161591\du}{13.036646\du}}
\pgfpathcurveto{\pgfpoint{6.186577\du}{13.079956\du}}{\pgfpoint{6.168254\du}{13.148251\du}}{\pgfpoint{6.124944\du}{13.173237\du}}
\pgfpathcurveto{\pgfpoint{6.081635\du}{13.198223\du}}{\pgfpoint{6.013340\du}{13.179900\du}}{\pgfpoint{5.988354\du}{13.136591\du}}
\pgfpathcurveto{\pgfpoint{5.963368\du}{13.093282\du}}{\pgfpoint{5.981691\du}{13.024986\du}}{\pgfpoint{6.025000\du}{13.000000\du}}
\pgfusepath{stroke}
\pgfsetlinewidth{0.100000\du}
\pgfsetdash{}{0pt}
\pgfsetmiterjoin
\pgfsetbuttcap
\definecolor{dialinecolor}{rgb}{0.000000, 0.000000, 0.000000}
\pgfsetfillcolor{dialinecolor}
\pgfpathmoveto{\pgfpoint{9.025000\du}{18.200000\du}}
\pgfpathcurveto{\pgfpoint{8.981691\du}{18.224986\du}}{\pgfpoint{8.913395\du}{18.206663\du}}{\pgfpoint{8.888409\du}{18.163354\du}}
\pgfpathcurveto{\pgfpoint{8.863423\du}{18.120044\du}}{\pgfpoint{8.881746\du}{18.051749\du}}{\pgfpoint{8.925056\du}{18.026763\du}}
\pgfpathcurveto{\pgfpoint{8.968365\du}{18.001777\du}}{\pgfpoint{9.036660\du}{18.020100\du}}{\pgfpoint{9.061646\du}{18.063409\du}}
\pgfpathcurveto{\pgfpoint{9.086632\du}{18.106718\du}}{\pgfpoint{9.068309\du}{18.175014\du}}{\pgfpoint{9.025000\du}{18.200000\du}}
\pgfusepath{fill}
\definecolor{dialinecolor}{rgb}{0.000000, 0.000000, 0.000000}
\pgfsetstrokecolor{dialinecolor}
\pgfpathmoveto{\pgfpoint{9.025000\du}{18.200000\du}}
\pgfpathcurveto{\pgfpoint{8.981691\du}{18.224986\du}}{\pgfpoint{8.913395\du}{18.206663\du}}{\pgfpoint{8.888409\du}{18.163354\du}}
\pgfpathcurveto{\pgfpoint{8.863423\du}{18.120044\du}}{\pgfpoint{8.881746\du}{18.051749\du}}{\pgfpoint{8.925056\du}{18.026763\du}}
\pgfpathcurveto{\pgfpoint{8.968365\du}{18.001777\du}}{\pgfpoint{9.036660\du}{18.020100\du}}{\pgfpoint{9.061646\du}{18.063409\du}}
\pgfpathcurveto{\pgfpoint{9.086632\du}{18.106718\du}}{\pgfpoint{9.068309\du}{18.175014\du}}{\pgfpoint{9.025000\du}{18.200000\du}}
\pgfusepath{stroke}
\pgfsetlinewidth{0.100000\du}
\pgfsetdash{}{0pt}
\pgfsetdash{}{0pt}
\pgfsetbuttcap
{
\definecolor{dialinecolor}{rgb}{0.000000, 0.000000, 0.000000}
\pgfsetfillcolor{dialinecolor}
}
\definecolor{dialinecolor}{rgb}{0.000000, 0.000000, 0.000000}
\pgfsetstrokecolor{dialinecolor}
\draw (4.625000\du,18.000000\du)--(9.825000\du,14.600000\du);
\pgfsetlinewidth{0.100000\du}
\pgfsetdash{}{0pt}
\pgfsetmiterjoin
\pgfsetbuttcap
\definecolor{dialinecolor}{rgb}{0.000000, 0.000000, 0.000000}
\pgfsetfillcolor{dialinecolor}
\pgfpathmoveto{\pgfpoint{4.625000\du}{18.000000\du}}
\pgfpathcurveto{\pgfpoint{4.597638\du}{17.958152\du}}{\pgfpoint{4.612124\du}{17.888941\du}}{\pgfpoint{4.653972\du}{17.861578\du}}
\pgfpathcurveto{\pgfpoint{4.695821\du}{17.834216\du}}{\pgfpoint{4.765031\du}{17.848702\du}}{\pgfpoint{4.792394\du}{17.890550\du}}
\pgfpathcurveto{\pgfpoint{4.819756\du}{17.932399\du}}{\pgfpoint{4.805270\du}{18.001610\du}}{\pgfpoint{4.763422\du}{18.028972\du}}
\pgfpathcurveto{\pgfpoint{4.721573\du}{18.056334\du}}{\pgfpoint{4.652362\du}{18.041848\du}}{\pgfpoint{4.625000\du}{18.000000\du}}
\pgfusepath{fill}
\definecolor{dialinecolor}{rgb}{0.000000, 0.000000, 0.000000}
\pgfsetstrokecolor{dialinecolor}
\pgfpathmoveto{\pgfpoint{4.625000\du}{18.000000\du}}
\pgfpathcurveto{\pgfpoint{4.597638\du}{17.958152\du}}{\pgfpoint{4.612124\du}{17.888941\du}}{\pgfpoint{4.653972\du}{17.861578\du}}
\pgfpathcurveto{\pgfpoint{4.695821\du}{17.834216\du}}{\pgfpoint{4.765031\du}{17.848702\du}}{\pgfpoint{4.792394\du}{17.890550\du}}
\pgfpathcurveto{\pgfpoint{4.819756\du}{17.932399\du}}{\pgfpoint{4.805270\du}{18.001610\du}}{\pgfpoint{4.763422\du}{18.028972\du}}
\pgfpathcurveto{\pgfpoint{4.721573\du}{18.056334\du}}{\pgfpoint{4.652362\du}{18.041848\du}}{\pgfpoint{4.625000\du}{18.000000\du}}
\pgfusepath{stroke}
\pgfsetlinewidth{0.100000\du}
\pgfsetdash{}{0pt}
\pgfsetmiterjoin
\pgfsetbuttcap
\definecolor{dialinecolor}{rgb}{0.000000, 0.000000, 0.000000}
\pgfsetfillcolor{dialinecolor}
\pgfpathmoveto{\pgfpoint{9.825000\du}{14.600000\du}}
\pgfpathcurveto{\pgfpoint{9.852362\du}{14.641848\du}}{\pgfpoint{9.837876\du}{14.711059\du}}{\pgfpoint{9.796028\du}{14.738422\du}}
\pgfpathcurveto{\pgfpoint{9.754179\du}{14.765784\du}}{\pgfpoint{9.684969\du}{14.751298\du}}{\pgfpoint{9.657606\du}{14.709450\du}}
\pgfpathcurveto{\pgfpoint{9.630244\du}{14.667601\du}}{\pgfpoint{9.644730\du}{14.598390\du}}{\pgfpoint{9.686578\du}{14.571028\du}}
\pgfpathcurveto{\pgfpoint{9.728427\du}{14.543666\du}}{\pgfpoint{9.797638\du}{14.558152\du}}{\pgfpoint{9.825000\du}{14.600000\du}}
\pgfusepath{fill}
\definecolor{dialinecolor}{rgb}{0.000000, 0.000000, 0.000000}
\pgfsetstrokecolor{dialinecolor}
\pgfpathmoveto{\pgfpoint{9.825000\du}{14.600000\du}}
\pgfpathcurveto{\pgfpoint{9.852362\du}{14.641848\du}}{\pgfpoint{9.837876\du}{14.711059\du}}{\pgfpoint{9.796028\du}{14.738422\du}}
\pgfpathcurveto{\pgfpoint{9.754179\du}{14.765784\du}}{\pgfpoint{9.684969\du}{14.751298\du}}{\pgfpoint{9.657606\du}{14.709450\du}}
\pgfpathcurveto{\pgfpoint{9.630244\du}{14.667601\du}}{\pgfpoint{9.644730\du}{14.598390\du}}{\pgfpoint{9.686578\du}{14.571028\du}}
\pgfpathcurveto{\pgfpoint{9.728427\du}{14.543666\du}}{\pgfpoint{9.797638\du}{14.558152\du}}{\pgfpoint{9.825000\du}{14.600000\du}}
\pgfusepath{stroke}
\pgfsetlinewidth{0.100000\du}
\pgfsetdash{{1.000000\du}{1.000000\du}}{0\du}
\pgfsetdash{{0.100000\du}{0.100000\du}}{0\du}
\pgfsetbuttcap
{
\definecolor{dialinecolor}{rgb}{0.000000, 0.000000, 0.000000}
\pgfsetfillcolor{dialinecolor}
\definecolor{dialinecolor}{rgb}{0.000000, 0.000000, 0.000000}
\pgfsetstrokecolor{dialinecolor}
\draw (8.025000\du,16.400000\du)--(7.561204\du,18.087500\du);
}
\definecolor{dialinecolor}{rgb}{0.000000, 0.000000, 0.000000}
\pgfsetstrokecolor{dialinecolor}
\draw (8.025000\du,16.400000\du)--(7.561204\du,18.087500\du);
\pgfsetlinewidth{0.100000\du}
\pgfsetdash{}{0pt}
\pgfsetmiterjoin
\pgfsetbuttcap
\definecolor{dialinecolor}{rgb}{0.000000, 0.000000, 0.000000}
\pgfsetfillcolor{dialinecolor}
\pgfpathmoveto{\pgfpoint{8.025000\du}{16.400000\du}}
\pgfpathcurveto{\pgfpoint{8.073212\du}{16.413251\du}}{\pgfpoint{8.108174\du}{16.474714\du}}{\pgfpoint{8.094923\du}{16.522926\du}}
\pgfpathcurveto{\pgfpoint{8.081672\du}{16.571138\du}}{\pgfpoint{8.020209\du}{16.606100\du}}{\pgfpoint{7.971997\du}{16.592849\du}}
\pgfpathcurveto{\pgfpoint{7.923785\du}{16.579598\du}}{\pgfpoint{7.888823\du}{16.518135\du}}{\pgfpoint{7.902074\du}{16.469923\du}}
\pgfpathcurveto{\pgfpoint{7.915325\du}{16.421711\du}}{\pgfpoint{7.976788\du}{16.386749\du}}{\pgfpoint{8.025000\du}{16.400000\du}}
\pgfusepath{fill}
\definecolor{dialinecolor}{rgb}{0.000000, 0.000000, 0.000000}
\pgfsetstrokecolor{dialinecolor}
\pgfpathmoveto{\pgfpoint{8.025000\du}{16.400000\du}}
\pgfpathcurveto{\pgfpoint{8.073212\du}{16.413251\du}}{\pgfpoint{8.108174\du}{16.474714\du}}{\pgfpoint{8.094923\du}{16.522926\du}}
\pgfpathcurveto{\pgfpoint{8.081672\du}{16.571138\du}}{\pgfpoint{8.020209\du}{16.606100\du}}{\pgfpoint{7.971997\du}{16.592849\du}}
\pgfpathcurveto{\pgfpoint{7.923785\du}{16.579598\du}}{\pgfpoint{7.888823\du}{16.518135\du}}{\pgfpoint{7.902074\du}{16.469923\du}}
\pgfpathcurveto{\pgfpoint{7.915325\du}{16.421711\du}}{\pgfpoint{7.976788\du}{16.386749\du}}{\pgfpoint{8.025000\du}{16.400000\du}}
\pgfusepath{stroke}
\pgfsetlinewidth{0.050000\du}
\pgfsetdash{}{0pt}
\pgfsetdash{}{0pt}
\pgfsetbuttcap
{
\definecolor{dialinecolor}{rgb}{0.000000, 0.000000, 0.000000}
\pgfsetfillcolor{dialinecolor}
\definecolor{dialinecolor}{rgb}{0.000000, 0.000000, 0.000000}
\pgfsetstrokecolor{dialinecolor}
\draw (8.992454\du,18.120313\du)--(4.736204\du,17.970313\du);
}
\pgfsetlinewidth{0.050000\du}
\pgfsetdash{}{0pt}
\pgfsetdash{}{0pt}
\pgfsetbuttcap
{
\definecolor{dialinecolor}{rgb}{0.000000, 0.000000, 0.000000}
\pgfsetfillcolor{dialinecolor}
\definecolor{dialinecolor}{rgb}{0.000000, 0.000000, 0.000000}
\pgfsetstrokecolor{dialinecolor}
\draw (6.061204\du,13.089063\du)--(4.729954\du,17.939063\du);
}
\pgfsetlinewidth{0.100000\du}
\pgfsetdash{}{0pt}
\pgfsetdash{}{0pt}
\pgfsetbuttcap
{
\definecolor{dialinecolor}{rgb}{0.000000, 0.000000, 0.000000}
\pgfsetfillcolor{dialinecolor}
\definecolor{dialinecolor}{rgb}{0.000000, 0.000000, 0.000000}
\pgfsetstrokecolor{dialinecolor}
\draw (7.187500\du,16.218750\du)--(7.542454\du,18.093750\du);
}
\definecolor{dialinecolor}{rgb}{0.000000, 0.000000, 0.000000}
\pgfsetstrokecolor{dialinecolor}
\draw (7.187500\du,16.218750\du)--(7.542454\du,18.093750\du);
\pgfsetlinewidth{0.100000\du}
\pgfsetdash{}{0pt}
\pgfsetmiterjoin
\pgfsetbuttcap
\definecolor{dialinecolor}{rgb}{0.000000, 0.000000, 0.000000}
\pgfsetfillcolor{dialinecolor}
\pgfpathmoveto{\pgfpoint{7.187500\du}{16.218750\du}}
\pgfpathcurveto{\pgfpoint{7.236627\du}{16.209450\du}}{\pgfpoint{7.295055\du}{16.249277\du}}{\pgfpoint{7.304355\du}{16.298404\du}}
\pgfpathcurveto{\pgfpoint{7.313656\du}{16.347532\du}}{\pgfpoint{7.273829\du}{16.405959\du}}{\pgfpoint{7.224701\du}{16.415260\du}}
\pgfpathcurveto{\pgfpoint{7.175574\du}{16.424560\du}}{\pgfpoint{7.117146\du}{16.384733\du}}{\pgfpoint{7.107846\du}{16.335605\du}}
\pgfpathcurveto{\pgfpoint{7.098545\du}{16.286478\du}}{\pgfpoint{7.138373\du}{16.228050\du}}{\pgfpoint{7.187500\du}{16.218750\du}}
\pgfusepath{fill}
\definecolor{dialinecolor}{rgb}{0.000000, 0.000000, 0.000000}
\pgfsetstrokecolor{dialinecolor}
\pgfpathmoveto{\pgfpoint{7.187500\du}{16.218750\du}}
\pgfpathcurveto{\pgfpoint{7.236627\du}{16.209450\du}}{\pgfpoint{7.295055\du}{16.249277\du}}{\pgfpoint{7.304355\du}{16.298404\du}}
\pgfpathcurveto{\pgfpoint{7.313656\du}{16.347532\du}}{\pgfpoint{7.273829\du}{16.405959\du}}{\pgfpoint{7.224701\du}{16.415260\du}}
\pgfpathcurveto{\pgfpoint{7.175574\du}{16.424560\du}}{\pgfpoint{7.117146\du}{16.384733\du}}{\pgfpoint{7.107846\du}{16.335605\du}}
\pgfpathcurveto{\pgfpoint{7.098545\du}{16.286478\du}}{\pgfpoint{7.138373\du}{16.228050\du}}{\pgfpoint{7.187500\du}{16.218750\du}}
\pgfusepath{stroke}
\pgfsetlinewidth{0.050000\du}
\pgfsetdash{}{0pt}
\pgfsetdash{}{0pt}
\pgfsetbuttcap
{
\definecolor{dialinecolor}{rgb}{0.000000, 0.000000, 0.000000}
\pgfsetfillcolor{dialinecolor}
\definecolor{dialinecolor}{rgb}{0.000000, 0.000000, 0.000000}
\pgfsetstrokecolor{dialinecolor}
\draw (9.754954\du,14.648438\du)--(9.025000\du,18.000000\du);
}
\pgfsetlinewidth{0.050000\du}
\pgfsetdash{}{0pt}
\pgfsetdash{}{0pt}
\pgfsetbuttcap
{
\definecolor{dialinecolor}{rgb}{0.000000, 0.000000, 0.000000}
\pgfsetfillcolor{dialinecolor}
\definecolor{dialinecolor}{rgb}{0.000000, 0.000000, 0.000000}
\pgfsetstrokecolor{dialinecolor}
\draw (9.748704\du,14.642188\du)--(6.086204\du,13.095313\du);
}
\definecolor{dialinecolor}{rgb}{0.000000, 0.000000, 0.000000}
\pgfsetstrokecolor{dialinecolor}
\node[anchor=west] at (3.708079\du,18.096875\du){};
\definecolor{dialinecolor}{rgb}{0.000000, 0.000000, 0.000000}
\pgfsetstrokecolor{dialinecolor}
\node[anchor=west] at (10.066250\du,14.773750\du){$v_2$};
\definecolor{dialinecolor}{rgb}{0.000000, 0.000000, 0.000000}
\pgfsetstrokecolor{dialinecolor}
\node[anchor=west] at (11.383079\du,14.246875\du){};
\definecolor{dialinecolor}{rgb}{0.000000, 0.000000, 0.000000}
\pgfsetstrokecolor{dialinecolor}
\node[anchor=west] at (5.820000\du,12.672500\du){$v_3$};
\definecolor{dialinecolor}{rgb}{0.000000, 0.000000, 0.000000}
\pgfsetstrokecolor{dialinecolor}
\node[anchor=west] at (7.383079\du,12.121875\du){};
\definecolor{dialinecolor}{rgb}{0.000000, 0.000000, 0.000000}
\pgfsetstrokecolor{dialinecolor}
\node[anchor=west] at (8.823750\du,18.896250\du){$v_4$};
\definecolor{dialinecolor}{rgb}{0.000000, 0.000000, 0.000000}
\pgfsetstrokecolor{dialinecolor}
\node[anchor=west] at (10.133079\du,19.171875\du){};
\definecolor{dialinecolor}{rgb}{0.000000, 0.000000, 0.000000}
\pgfsetstrokecolor{dialinecolor}
\node[anchor=west] at (8.433079\du,9.557812\du){};
\definecolor{dialinecolor}{rgb}{0.000000, 0.000000, 0.000000}
\pgfsetstrokecolor{dialinecolor}
\node[anchor=west] at (6.016250\du,11.686250\du){$a_L(v_3)=a_R(v_4)=180^\circ$};
\definecolor{dialinecolor}{rgb}{0.000000, 0.000000, 0.000000}
\pgfsetstrokecolor{dialinecolor}
\node[anchor=west] at (5.208079\du,20.332812\du){};
\pgfsetlinewidth{0.000000\du}
\pgfsetdash{}{0pt}
\pgfsetdash{}{0pt}
\pgfsetbuttcap
{
\definecolor{dialinecolor}{rgb}{0.000000, 0.000000, 0.000000}
\pgfsetfillcolor{dialinecolor}
\pgfsetarrowsstart{stealth}
\definecolor{dialinecolor}{rgb}{0.000000, 0.000000, 0.000000}
\pgfsetstrokecolor{dialinecolor}
\draw (6.891204\du,15.875937\du)--(4.689329\du,11.185937\du);
}
\end{tikzpicture}

%% file: figure/Type3.tex
\ifx\du\undefined
  \newlength{\du}
\fi
\setlength{\du}{15\unitlength}
\begin{tikzpicture}
\pgftransformxscale{1.000000}
\pgftransformyscale{-1.000000}
\definecolor{dialinecolor}{rgb}{0.000000, 0.000000, 0.000000}
\pgfsetstrokecolor{dialinecolor}
\definecolor{dialinecolor}{rgb}{1.000000, 1.000000, 1.000000}
\pgfsetfillcolor{dialinecolor}
\pgfsetlinewidth{0.050000\du}
\pgfsetdash{}{0pt}
\pgfsetdash{}{0pt}
\pgfsetbuttcap
{
\definecolor{dialinecolor}{rgb}{0.000000, 0.000000, 0.000000}
\pgfsetfillcolor{dialinecolor}
}
\definecolor{dialinecolor}{rgb}{0.000000, 0.000000, 0.000000}
\pgfsetstrokecolor{dialinecolor}
\draw (5.000000\du,-6.000000\du)--(5.000000\du,-1.600000\du);
\pgfsetlinewidth{0.050000\du}
\pgfsetdash{}{0pt}
\pgfsetmiterjoin
\pgfsetbuttcap
\definecolor{dialinecolor}{rgb}{0.000000, 0.000000, 0.000000}
\pgfsetfillcolor{dialinecolor}
\pgfpathmoveto{\pgfpoint{5.000000\du}{-6.000000\du}}
\pgfpathcurveto{\pgfpoint{5.050000\du}{-6.000000\du}}{\pgfpoint{5.100000\du}{-5.950000\du}}{\pgfpoint{5.100000\du}{-5.900000\du}}
\pgfpathcurveto{\pgfpoint{5.100000\du}{-5.850000\du}}{\pgfpoint{5.050000\du}{-5.800000\du}}{\pgfpoint{5.000000\du}{-5.800000\du}}
\pgfpathcurveto{\pgfpoint{4.950000\du}{-5.800000\du}}{\pgfpoint{4.900000\du}{-5.850000\du}}{\pgfpoint{4.900000\du}{-5.900000\du}}
\pgfpathcurveto{\pgfpoint{4.900000\du}{-5.950000\du}}{\pgfpoint{4.950000\du}{-6.000000\du}}{\pgfpoint{5.000000\du}{-6.000000\du}}
\pgfusepath{fill}
\definecolor{dialinecolor}{rgb}{0.000000, 0.000000, 0.000000}
\pgfsetstrokecolor{dialinecolor}
\pgfpathmoveto{\pgfpoint{5.000000\du}{-6.000000\du}}
\pgfpathcurveto{\pgfpoint{5.050000\du}{-6.000000\du}}{\pgfpoint{5.100000\du}{-5.950000\du}}{\pgfpoint{5.100000\du}{-5.900000\du}}
\pgfpathcurveto{\pgfpoint{5.100000\du}{-5.850000\du}}{\pgfpoint{5.050000\du}{-5.800000\du}}{\pgfpoint{5.000000\du}{-5.800000\du}}
\pgfpathcurveto{\pgfpoint{4.950000\du}{-5.800000\du}}{\pgfpoint{4.900000\du}{-5.850000\du}}{\pgfpoint{4.900000\du}{-5.900000\du}}
\pgfpathcurveto{\pgfpoint{4.900000\du}{-5.950000\du}}{\pgfpoint{4.950000\du}{-6.000000\du}}{\pgfpoint{5.000000\du}{-6.000000\du}}
\pgfusepath{stroke}
\pgfsetlinewidth{0.050000\du}
\pgfsetdash{}{0pt}
\pgfsetmiterjoin
\pgfsetbuttcap
\definecolor{dialinecolor}{rgb}{0.000000, 0.000000, 0.000000}
\pgfsetfillcolor{dialinecolor}
\pgfpathmoveto{\pgfpoint{5.000000\du}{-1.600000\du}}
\pgfpathcurveto{\pgfpoint{4.950000\du}{-1.600000\du}}{\pgfpoint{4.900000\du}{-1.650000\du}}{\pgfpoint{4.900000\du}{-1.700000\du}}
\pgfpathcurveto{\pgfpoint{4.900000\du}{-1.750000\du}}{\pgfpoint{4.950000\du}{-1.800000\du}}{\pgfpoint{5.000000\du}{-1.800000\du}}
\pgfpathcurveto{\pgfpoint{5.050000\du}{-1.800000\du}}{\pgfpoint{5.100000\du}{-1.750000\du}}{\pgfpoint{5.100000\du}{-1.700000\du}}
\pgfpathcurveto{\pgfpoint{5.100000\du}{-1.650000\du}}{\pgfpoint{5.050000\du}{-1.600000\du}}{\pgfpoint{5.000000\du}{-1.600000\du}}
\pgfusepath{fill}
\definecolor{dialinecolor}{rgb}{0.000000, 0.000000, 0.000000}
\pgfsetstrokecolor{dialinecolor}
\pgfpathmoveto{\pgfpoint{5.000000\du}{-1.600000\du}}
\pgfpathcurveto{\pgfpoint{4.950000\du}{-1.600000\du}}{\pgfpoint{4.900000\du}{-1.650000\du}}{\pgfpoint{4.900000\du}{-1.700000\du}}
\pgfpathcurveto{\pgfpoint{4.900000\du}{-1.750000\du}}{\pgfpoint{4.950000\du}{-1.800000\du}}{\pgfpoint{5.000000\du}{-1.800000\du}}
\pgfpathcurveto{\pgfpoint{5.050000\du}{-1.800000\du}}{\pgfpoint{5.100000\du}{-1.750000\du}}{\pgfpoint{5.100000\du}{-1.700000\du}}
\pgfpathcurveto{\pgfpoint{5.100000\du}{-1.650000\du}}{\pgfpoint{5.050000\du}{-1.600000\du}}{\pgfpoint{5.000000\du}{-1.600000\du}}
\pgfusepath{stroke}
\definecolor{dialinecolor}{rgb}{0.000000, 0.000000, 0.000000}
\pgfsetstrokecolor{dialinecolor}
\node[anchor=west] at (6.000000\du,-2.000000\du){};
\pgfsetlinewidth{0.050000\du}
\pgfsetdash{}{0pt}
\pgfsetdash{}{0pt}
\pgfsetbuttcap
{
\definecolor{dialinecolor}{rgb}{0.000000, 0.000000, 0.000000}
\pgfsetfillcolor{dialinecolor}
\definecolor{dialinecolor}{rgb}{0.000000, 0.000000, 0.000000}
\pgfsetstrokecolor{dialinecolor}
\draw (5.000000\du,-3.800000\du)--(4.400000\du,-4.400000\du);
}
\pgfsetlinewidth{0.050000\du}
\pgfsetdash{}{0pt}
\pgfsetdash{}{0pt}
\pgfsetbuttcap
{
\definecolor{dialinecolor}{rgb}{0.000000, 0.000000, 0.000000}
\pgfsetfillcolor{dialinecolor}
\definecolor{dialinecolor}{rgb}{0.000000, 0.000000, 0.000000}
\pgfsetstrokecolor{dialinecolor}
\draw (4.200000\du,-3.800000\du)--(4.400000\du,-4.400000\du);
}
\pgfsetlinewidth{0.050000\du}
\pgfsetdash{}{0pt}
\pgfsetdash{}{0pt}
\pgfsetbuttcap
{
\definecolor{dialinecolor}{rgb}{0.000000, 0.000000, 0.000000}
\pgfsetfillcolor{dialinecolor}
\definecolor{dialinecolor}{rgb}{0.000000, 0.000000, 0.000000}
\pgfsetstrokecolor{dialinecolor}
\draw (4.200000\du,-3.800000\du)--(4.000000\du,-4.400000\du);
}
\pgfsetlinewidth{0.050000\du}
\pgfsetdash{}{0pt}
\pgfsetdash{}{0pt}
\pgfsetbuttcap
{
\definecolor{dialinecolor}{rgb}{0.000000, 0.000000, 0.000000}
\pgfsetfillcolor{dialinecolor}
\definecolor{dialinecolor}{rgb}{0.000000, 0.000000, 0.000000}
\pgfsetstrokecolor{dialinecolor}
\draw (3.800000\du,-3.800000\du)--(4.000000\du,-4.400000\du);
}
\pgfsetlinewidth{0.050000\du}
\pgfsetdash{}{0pt}
\pgfsetdash{}{0pt}
\pgfsetbuttcap
{
\definecolor{dialinecolor}{rgb}{0.000000, 0.000000, 0.000000}
\pgfsetfillcolor{dialinecolor}
}
\definecolor{dialinecolor}{rgb}{0.000000, 0.000000, 0.000000}
\pgfsetstrokecolor{dialinecolor}
\draw (3.800000\du,-3.800000\du)--(1.800000\du,-5.400000\du);
\pgfsetlinewidth{0.050000\du}
\pgfsetdash{}{0pt}
\pgfsetmiterjoin
\pgfsetbuttcap
\definecolor{dialinecolor}{rgb}{0.000000, 0.000000, 0.000000}
\pgfsetfillcolor{dialinecolor}
\pgfpathmoveto{\pgfpoint{3.800000\du}{-3.800000\du}}
\pgfpathcurveto{\pgfpoint{3.768765\du}{-3.760957\du}}{\pgfpoint{3.698487\du}{-3.753148\du}}{\pgfpoint{3.659444\du}{-3.784383\du}}
\pgfpathcurveto{\pgfpoint{3.620400\du}{-3.815617\du}}{\pgfpoint{3.612591\du}{-3.885896\du}}{\pgfpoint{3.643826\du}{-3.924939\du}}
\pgfpathcurveto{\pgfpoint{3.675061\du}{-3.963982\du}}{\pgfpoint{3.745339\du}{-3.971791\du}}{\pgfpoint{3.784383\du}{-3.940556\du}}
\pgfpathcurveto{\pgfpoint{3.823426\du}{-3.909322\du}}{\pgfpoint{3.831235\du}{-3.839043\du}}{\pgfpoint{3.800000\du}{-3.800000\du}}
\pgfusepath{fill}
\definecolor{dialinecolor}{rgb}{0.000000, 0.000000, 0.000000}
\pgfsetstrokecolor{dialinecolor}
\pgfpathmoveto{\pgfpoint{3.800000\du}{-3.800000\du}}
\pgfpathcurveto{\pgfpoint{3.768765\du}{-3.760957\du}}{\pgfpoint{3.698487\du}{-3.753148\du}}{\pgfpoint{3.659444\du}{-3.784383\du}}
\pgfpathcurveto{\pgfpoint{3.620400\du}{-3.815617\du}}{\pgfpoint{3.612591\du}{-3.885896\du}}{\pgfpoint{3.643826\du}{-3.924939\du}}
\pgfpathcurveto{\pgfpoint{3.675061\du}{-3.963982\du}}{\pgfpoint{3.745339\du}{-3.971791\du}}{\pgfpoint{3.784383\du}{-3.940556\du}}
\pgfpathcurveto{\pgfpoint{3.823426\du}{-3.909322\du}}{\pgfpoint{3.831235\du}{-3.839043\du}}{\pgfpoint{3.800000\du}{-3.800000\du}}
\pgfusepath{stroke}
\pgfsetlinewidth{0.050000\du}
\pgfsetdash{}{0pt}
\pgfsetmiterjoin
\pgfsetbuttcap
\definecolor{dialinecolor}{rgb}{0.000000, 0.000000, 0.000000}
\pgfsetfillcolor{dialinecolor}
\pgfpathmoveto{\pgfpoint{1.800000\du}{-5.400000\du}}
\pgfpathcurveto{\pgfpoint{1.831235\du}{-5.439043\du}}{\pgfpoint{1.901513\du}{-5.446852\du}}{\pgfpoint{1.940556\du}{-5.415617\du}}
\pgfpathcurveto{\pgfpoint{1.979600\du}{-5.384383\du}}{\pgfpoint{1.987409\du}{-5.314104\du}}{\pgfpoint{1.956174\du}{-5.275061\du}}
\pgfpathcurveto{\pgfpoint{1.924939\du}{-5.236018\du}}{\pgfpoint{1.854661\du}{-5.228209\du}}{\pgfpoint{1.815617\du}{-5.259444\du}}
\pgfpathcurveto{\pgfpoint{1.776574\du}{-5.290678\du}}{\pgfpoint{1.768765\du}{-5.360957\du}}{\pgfpoint{1.800000\du}{-5.400000\du}}
\pgfusepath{fill}
\definecolor{dialinecolor}{rgb}{0.000000, 0.000000, 0.000000}
\pgfsetstrokecolor{dialinecolor}
\pgfpathmoveto{\pgfpoint{1.800000\du}{-5.400000\du}}
\pgfpathcurveto{\pgfpoint{1.831235\du}{-5.439043\du}}{\pgfpoint{1.901513\du}{-5.446852\du}}{\pgfpoint{1.940556\du}{-5.415617\du}}
\pgfpathcurveto{\pgfpoint{1.979600\du}{-5.384383\du}}{\pgfpoint{1.987409\du}{-5.314104\du}}{\pgfpoint{1.956174\du}{-5.275061\du}}
\pgfpathcurveto{\pgfpoint{1.924939\du}{-5.236018\du}}{\pgfpoint{1.854661\du}{-5.228209\du}}{\pgfpoint{1.815617\du}{-5.259444\du}}
\pgfpathcurveto{\pgfpoint{1.776574\du}{-5.290678\du}}{\pgfpoint{1.768765\du}{-5.360957\du}}{\pgfpoint{1.800000\du}{-5.400000\du}}
\pgfusepath{stroke}
\pgfsetlinewidth{0.050000\du}
\pgfsetdash{}{0pt}
\pgfsetdash{}{0pt}
\pgfsetbuttcap
{
\definecolor{dialinecolor}{rgb}{0.000000, 0.000000, 0.000000}
\pgfsetfillcolor{dialinecolor}
}
\definecolor{dialinecolor}{rgb}{0.000000, 0.000000, 0.000000}
\pgfsetstrokecolor{dialinecolor}
\draw (2.800000\du,-3.600000\du)--(3.000000\du,-4.600000\du);
\pgfsetlinewidth{0.050000\du}
\pgfsetdash{}{0pt}
\pgfsetmiterjoin
\pgfsetbuttcap
\definecolor{dialinecolor}{rgb}{0.000000, 0.000000, 0.000000}
\pgfsetfillcolor{dialinecolor}
\pgfpathmoveto{\pgfpoint{2.800000\du}{-3.600000\du}}
\pgfpathcurveto{\pgfpoint{2.750971\du}{-3.609806\du}}{\pgfpoint{2.711748\du}{-3.668641\du}}{\pgfpoint{2.721554\du}{-3.717670\du}}
\pgfpathcurveto{\pgfpoint{2.731359\du}{-3.766699\du}}{\pgfpoint{2.790194\du}{-3.805922\du}}{\pgfpoint{2.839223\du}{-3.796116\du}}
\pgfpathcurveto{\pgfpoint{2.888252\du}{-3.786310\du}}{\pgfpoint{2.927475\du}{-3.727475\du}}{\pgfpoint{2.917670\du}{-3.678446\du}}
\pgfpathcurveto{\pgfpoint{2.907864\du}{-3.629417\du}}{\pgfpoint{2.849029\du}{-3.590194\du}}{\pgfpoint{2.800000\du}{-3.600000\du}}
\pgfusepath{fill}
\definecolor{dialinecolor}{rgb}{0.000000, 0.000000, 0.000000}
\pgfsetstrokecolor{dialinecolor}
\pgfpathmoveto{\pgfpoint{2.800000\du}{-3.600000\du}}
\pgfpathcurveto{\pgfpoint{2.750971\du}{-3.609806\du}}{\pgfpoint{2.711748\du}{-3.668641\du}}{\pgfpoint{2.721554\du}{-3.717670\du}}
\pgfpathcurveto{\pgfpoint{2.731359\du}{-3.766699\du}}{\pgfpoint{2.790194\du}{-3.805922\du}}{\pgfpoint{2.839223\du}{-3.796116\du}}
\pgfpathcurveto{\pgfpoint{2.888252\du}{-3.786310\du}}{\pgfpoint{2.927475\du}{-3.727475\du}}{\pgfpoint{2.917670\du}{-3.678446\du}}
\pgfpathcurveto{\pgfpoint{2.907864\du}{-3.629417\du}}{\pgfpoint{2.849029\du}{-3.590194\du}}{\pgfpoint{2.800000\du}{-3.600000\du}}
\pgfusepath{stroke}
\pgfsetlinewidth{0.050000\du}
\pgfsetdash{}{0pt}
\pgfsetmiterjoin
\pgfsetbuttcap
\definecolor{dialinecolor}{rgb}{0.000000, 0.000000, 0.000000}
\pgfsetfillcolor{dialinecolor}
\pgfpathmoveto{\pgfpoint{3.000000\du}{-4.600000\du}}
\pgfpathcurveto{\pgfpoint{3.049029\du}{-4.590194\du}}{\pgfpoint{3.088252\du}{-4.531359\du}}{\pgfpoint{3.078446\du}{-4.482330\du}}
\pgfpathcurveto{\pgfpoint{3.068641\du}{-4.433301\du}}{\pgfpoint{3.009806\du}{-4.394078\du}}{\pgfpoint{2.960777\du}{-4.403884\du}}
\pgfpathcurveto{\pgfpoint{2.911748\du}{-4.413690\du}}{\pgfpoint{2.872525\du}{-4.472525\du}}{\pgfpoint{2.882330\du}{-4.521554\du}}
\pgfpathcurveto{\pgfpoint{2.892136\du}{-4.570583\du}}{\pgfpoint{2.950971\du}{-4.609806\du}}{\pgfpoint{3.000000\du}{-4.600000\du}}
\pgfusepath{fill}
\definecolor{dialinecolor}{rgb}{0.000000, 0.000000, 0.000000}
\pgfsetstrokecolor{dialinecolor}
\pgfpathmoveto{\pgfpoint{3.000000\du}{-4.600000\du}}
\pgfpathcurveto{\pgfpoint{3.049029\du}{-4.590194\du}}{\pgfpoint{3.088252\du}{-4.531359\du}}{\pgfpoint{3.078446\du}{-4.482330\du}}
\pgfpathcurveto{\pgfpoint{3.068641\du}{-4.433301\du}}{\pgfpoint{3.009806\du}{-4.394078\du}}{\pgfpoint{2.960777\du}{-4.403884\du}}
\pgfpathcurveto{\pgfpoint{2.911748\du}{-4.413690\du}}{\pgfpoint{2.872525\du}{-4.472525\du}}{\pgfpoint{2.882330\du}{-4.521554\du}}
\pgfpathcurveto{\pgfpoint{2.892136\du}{-4.570583\du}}{\pgfpoint{2.950971\du}{-4.609806\du}}{\pgfpoint{3.000000\du}{-4.600000\du}}
\pgfusepath{stroke}
\definecolor{dialinecolor}{rgb}{0.000000, 0.000000, 0.000000}
\pgfsetstrokecolor{dialinecolor}
\node[anchor=west] at (1.800000\du,-5.800000\du){};
\pgfsetlinewidth{0.050000\du}
\pgfsetdash{}{0pt}
\pgfsetdash{}{0pt}
\pgfsetbuttcap
{
\definecolor{dialinecolor}{rgb}{0.000000, 0.000000, 0.000000}
\pgfsetfillcolor{dialinecolor}
\definecolor{dialinecolor}{rgb}{0.000000, 0.000000, 0.000000}
\pgfsetstrokecolor{dialinecolor}
\draw (1.600000\du,-2.200000\du)--(2.800000\du,-3.600000\du);
}
\definecolor{dialinecolor}{rgb}{0.000000, 0.000000, 0.000000}
\pgfsetstrokecolor{dialinecolor}
\draw (1.600000\du,-2.200000\du)--(2.800000\du,-3.600000\du);
\pgfsetlinewidth{0.050000\du}
\pgfsetdash{}{0pt}
\pgfsetmiterjoin
\pgfsetbuttcap
\definecolor{dialinecolor}{rgb}{0.000000, 0.000000, 0.000000}
\pgfsetfillcolor{dialinecolor}
\pgfpathmoveto{\pgfpoint{1.600000\du}{-2.200000\du}}
\pgfpathcurveto{\pgfpoint{1.562037\du}{-2.232540\du}}{\pgfpoint{1.556614\du}{-2.303042\du}}{\pgfpoint{1.589153\du}{-2.341005\du}}
\pgfpathcurveto{\pgfpoint{1.621693\du}{-2.378968\du}}{\pgfpoint{1.692195\du}{-2.384391\du}}{\pgfpoint{1.730158\du}{-2.351851\du}}
\pgfpathcurveto{\pgfpoint{1.768121\du}{-2.319312\du}}{\pgfpoint{1.773544\du}{-2.248809\du}}{\pgfpoint{1.741005\du}{-2.210847\du}}
\pgfpathcurveto{\pgfpoint{1.708465\du}{-2.172884\du}}{\pgfpoint{1.637963\du}{-2.167460\du}}{\pgfpoint{1.600000\du}{-2.200000\du}}
\pgfusepath{fill}
\definecolor{dialinecolor}{rgb}{0.000000, 0.000000, 0.000000}
\pgfsetstrokecolor{dialinecolor}
\pgfpathmoveto{\pgfpoint{1.600000\du}{-2.200000\du}}
\pgfpathcurveto{\pgfpoint{1.562037\du}{-2.232540\du}}{\pgfpoint{1.556614\du}{-2.303042\du}}{\pgfpoint{1.589153\du}{-2.341005\du}}
\pgfpathcurveto{\pgfpoint{1.621693\du}{-2.378968\du}}{\pgfpoint{1.692195\du}{-2.384391\du}}{\pgfpoint{1.730158\du}{-2.351851\du}}
\pgfpathcurveto{\pgfpoint{1.768121\du}{-2.319312\du}}{\pgfpoint{1.773544\du}{-2.248809\du}}{\pgfpoint{1.741005\du}{-2.210847\du}}
\pgfpathcurveto{\pgfpoint{1.708465\du}{-2.172884\du}}{\pgfpoint{1.637963\du}{-2.167460\du}}{\pgfpoint{1.600000\du}{-2.200000\du}}
\pgfusepath{stroke}
\definecolor{dialinecolor}{rgb}{0.000000, 0.000000, 0.000000}
\pgfsetstrokecolor{dialinecolor}
\node[anchor=west] at (2.315000\du,-5.382500\du){};
\definecolor{dialinecolor}{rgb}{0.000000, 0.000000, 0.000000}
\pgfsetstrokecolor{dialinecolor}
\node[anchor=west] at (1.600000\du,-1.400000\du){};
\definecolor{dialinecolor}{rgb}{0.000000, 0.000000, 0.000000}
\pgfsetstrokecolor{dialinecolor}
\node[anchor=west] at (5.400000\du,-6.200000\du){};
\definecolor{dialinecolor}{rgb}{0.000000, 0.000000, 0.000000}
\pgfsetstrokecolor{dialinecolor}
\node[anchor=west] at (6.200000\du,-6.600000\du){};
\definecolor{dialinecolor}{rgb}{0.000000, 0.000000, 0.000000}
\pgfsetstrokecolor{dialinecolor}
\node[anchor=west] at (3.475000\du,-3.250000\du){$x$};
\definecolor{dialinecolor}{rgb}{0.000000, 0.000000, 0.000000}
\pgfsetstrokecolor{dialinecolor}
\node[anchor=west] at (3.050000\du,-4.825000\du){$y$};
\definecolor{dialinecolor}{rgb}{0.000000, 0.000000, 0.000000}
\pgfsetstrokecolor{dialinecolor}
\node[anchor=west] at (2.000000\du,-3.800000\du){$w$};
\definecolor{dialinecolor}{rgb}{0.000000, 0.000000, 0.000000}
\pgfsetstrokecolor{dialinecolor}
\node[anchor=west] at (3.000000\du,-5.000000\du){};
\pgfsetlinewidth{0.050000\du}
\pgfsetdash{{\pgflinewidth}{0.200000\du}}{0cm}
\pgfsetdash{{\pgflinewidth}{0.200000\du}}{0cm}
\pgfsetbuttcap
{
\definecolor{dialinecolor}{rgb}{0.000000, 0.000000, 0.000000}
\pgfsetfillcolor{dialinecolor}
\definecolor{dialinecolor}{rgb}{0.000000, 0.000000, 0.000000}
\pgfsetstrokecolor{dialinecolor}
\pgfpathmoveto{\pgfpoint{10.418785\du}{-3.378120\du}}
\pgfpatharc{143}{-26}{0.394689\du and 0.394689\du}
\pgfusepath{stroke}
}
\definecolor{dialinecolor}{rgb}{0.000000, 0.000000, 0.000000}
\pgfsetstrokecolor{dialinecolor}
\node[anchor=west] at (7.532550\du,-1.282800\du){$v_1$};
\definecolor{dialinecolor}{rgb}{0.000000, 0.000000, 0.000000}
\pgfsetstrokecolor{dialinecolor}
\node[anchor=west] at (6.157550\du,-2.545300\du){};
\definecolor{dialinecolor}{rgb}{0.000000, 0.000000, 0.000000}
\pgfsetstrokecolor{dialinecolor}
\node[anchor=west] at (6.232550\du,-7.920300\du){};
\definecolor{dialinecolor}{rgb}{0.000000, 0.000000, 0.000000}
\pgfsetstrokecolor{dialinecolor}
\node[anchor=west] at (7.032550\du,-7.920300\du){};
\pgfsetlinewidth{0.100000\du}
\pgfsetdash{}{0pt}
\pgfsetdash{}{0pt}
\pgfsetbuttcap
{
\definecolor{dialinecolor}{rgb}{0.000000, 0.000000, 0.000000}
\pgfsetfillcolor{dialinecolor}
}
\definecolor{dialinecolor}{rgb}{0.000000, 0.000000, 0.000000}
\pgfsetstrokecolor{dialinecolor}
\draw (8.232550\du,-1.920300\du)--(13.432600\du,-5.320300\du);
\pgfsetlinewidth{0.100000\du}
\pgfsetdash{}{0pt}
\pgfsetmiterjoin
\pgfsetbuttcap
\definecolor{dialinecolor}{rgb}{0.000000, 0.000000, 0.000000}
\pgfsetfillcolor{dialinecolor}
\pgfpathmoveto{\pgfpoint{8.232550\du}{-1.920300\du}}
\pgfpathcurveto{\pgfpoint{8.205188\du}{-1.962149\du}}{\pgfpoint{8.219674\du}{-2.031359\du}}{\pgfpoint{8.261523\du}{-2.058722\du}}
\pgfpathcurveto{\pgfpoint{8.303371\du}{-2.086084\du}}{\pgfpoint{8.372582\du}{-2.071598\du}}{\pgfpoint{8.399944\du}{-2.029749\du}}
\pgfpathcurveto{\pgfpoint{8.427307\du}{-1.987901\du}}{\pgfpoint{8.412820\du}{-1.918690\du}}{\pgfpoint{8.370972\du}{-1.891327\du}}
\pgfpathcurveto{\pgfpoint{8.329123\du}{-1.863965\du}}{\pgfpoint{8.259912\du}{-1.878451\du}}{\pgfpoint{8.232550\du}{-1.920300\du}}
\pgfusepath{fill}
\definecolor{dialinecolor}{rgb}{0.000000, 0.000000, 0.000000}
\pgfsetstrokecolor{dialinecolor}
\pgfpathmoveto{\pgfpoint{8.232550\du}{-1.920300\du}}
\pgfpathcurveto{\pgfpoint{8.205188\du}{-1.962149\du}}{\pgfpoint{8.219674\du}{-2.031359\du}}{\pgfpoint{8.261523\du}{-2.058722\du}}
\pgfpathcurveto{\pgfpoint{8.303371\du}{-2.086084\du}}{\pgfpoint{8.372582\du}{-2.071598\du}}{\pgfpoint{8.399944\du}{-2.029749\du}}
\pgfpathcurveto{\pgfpoint{8.427307\du}{-1.987901\du}}{\pgfpoint{8.412820\du}{-1.918690\du}}{\pgfpoint{8.370972\du}{-1.891327\du}}
\pgfpathcurveto{\pgfpoint{8.329123\du}{-1.863965\du}}{\pgfpoint{8.259912\du}{-1.878451\du}}{\pgfpoint{8.232550\du}{-1.920300\du}}
\pgfusepath{stroke}
\pgfsetlinewidth{0.100000\du}
\pgfsetdash{}{0pt}
\pgfsetmiterjoin
\pgfsetbuttcap
\definecolor{dialinecolor}{rgb}{0.000000, 0.000000, 0.000000}
\pgfsetfillcolor{dialinecolor}
\pgfpathmoveto{\pgfpoint{13.432600\du}{-5.320300\du}}
\pgfpathcurveto{\pgfpoint{13.459962\du}{-5.278451\du}}{\pgfpoint{13.445476\du}{-5.209241\du}}{\pgfpoint{13.403627\du}{-5.181878\du}}
\pgfpathcurveto{\pgfpoint{13.361779\du}{-5.154516\du}}{\pgfpoint{13.292568\du}{-5.169002\du}}{\pgfpoint{13.265206\du}{-5.210851\du}}
\pgfpathcurveto{\pgfpoint{13.237843\du}{-5.252699\du}}{\pgfpoint{13.252330\du}{-5.321910\du}}{\pgfpoint{13.294178\du}{-5.349273\du}}
\pgfpathcurveto{\pgfpoint{13.336027\du}{-5.376635\du}}{\pgfpoint{13.405238\du}{-5.362149\du}}{\pgfpoint{13.432600\du}{-5.320300\du}}
\pgfusepath{fill}
\definecolor{dialinecolor}{rgb}{0.000000, 0.000000, 0.000000}
\pgfsetstrokecolor{dialinecolor}
\pgfpathmoveto{\pgfpoint{13.432600\du}{-5.320300\du}}
\pgfpathcurveto{\pgfpoint{13.459962\du}{-5.278451\du}}{\pgfpoint{13.445476\du}{-5.209241\du}}{\pgfpoint{13.403627\du}{-5.181878\du}}
\pgfpathcurveto{\pgfpoint{13.361779\du}{-5.154516\du}}{\pgfpoint{13.292568\du}{-5.169002\du}}{\pgfpoint{13.265206\du}{-5.210851\du}}
\pgfpathcurveto{\pgfpoint{13.237843\du}{-5.252699\du}}{\pgfpoint{13.252330\du}{-5.321910\du}}{\pgfpoint{13.294178\du}{-5.349273\du}}
\pgfpathcurveto{\pgfpoint{13.336027\du}{-5.376635\du}}{\pgfpoint{13.405238\du}{-5.362149\du}}{\pgfpoint{13.432600\du}{-5.320300\du}}
\pgfusepath{stroke}
\pgfsetlinewidth{0.050000\du}
\pgfsetdash{}{0pt}
\pgfsetdash{}{0pt}
\pgfsetbuttcap
{
\definecolor{dialinecolor}{rgb}{0.000000, 0.000000, 0.000000}
\pgfsetfillcolor{dialinecolor}
\definecolor{dialinecolor}{rgb}{0.000000, 0.000000, 0.000000}
\pgfsetstrokecolor{dialinecolor}
\draw (12.600000\du,-1.800000\du)--(8.343750\du,-1.950000\du);
}
\pgfsetlinewidth{0.050000\du}
\pgfsetdash{}{0pt}
\pgfsetdash{}{0pt}
\pgfsetbuttcap
{
\definecolor{dialinecolor}{rgb}{0.000000, 0.000000, 0.000000}
\pgfsetfillcolor{dialinecolor}
\definecolor{dialinecolor}{rgb}{0.000000, 0.000000, 0.000000}
\pgfsetstrokecolor{dialinecolor}
\draw (9.668750\du,-6.831200\du)--(8.337500\du,-1.981200\du);
}
\pgfsetlinewidth{0.100000\du}
\pgfsetdash{}{0pt}
\pgfsetdash{}{0pt}
\pgfsetbuttcap
{
\definecolor{dialinecolor}{rgb}{0.000000, 0.000000, 0.000000}
\pgfsetfillcolor{dialinecolor}
\definecolor{dialinecolor}{rgb}{0.000000, 0.000000, 0.000000}
\pgfsetstrokecolor{dialinecolor}
\draw (10.868700\du,-3.537780\du)--(9.962500\du,-3.868750\du);
}
\definecolor{dialinecolor}{rgb}{0.000000, 0.000000, 0.000000}
\pgfsetstrokecolor{dialinecolor}
\draw (10.868700\du,-3.537780\du)--(9.962500\du,-3.868750\du);
\pgfsetlinewidth{0.100000\du}
\pgfsetdash{}{0pt}
\pgfsetmiterjoin
\pgfsetbuttcap
\definecolor{dialinecolor}{rgb}{0.000000, 0.000000, 0.000000}
\pgfsetfillcolor{dialinecolor}
\pgfpathmoveto{\pgfpoint{10.868700\du}{-3.537780\du}}
\pgfpathcurveto{\pgfpoint{10.851547\du}{-3.490814\du}}{\pgfpoint{10.787428\du}{-3.461002\du}}{\pgfpoint{10.740462\du}{-3.478155\du}}
\pgfpathcurveto{\pgfpoint{10.693497\du}{-3.495308\du}}{\pgfpoint{10.663684\du}{-3.559427\du}}{\pgfpoint{10.680838\du}{-3.606393\du}}
\pgfpathcurveto{\pgfpoint{10.697991\du}{-3.653358\du}}{\pgfpoint{10.762110\du}{-3.683171\du}}{\pgfpoint{10.809075\du}{-3.666018\du}}
\pgfpathcurveto{\pgfpoint{10.856041\du}{-3.648864\du}}{\pgfpoint{10.885853\du}{-3.584746\du}}{\pgfpoint{10.868700\du}{-3.537780\du}}
\pgfusepath{fill}
\definecolor{dialinecolor}{rgb}{0.000000, 0.000000, 0.000000}
\pgfsetstrokecolor{dialinecolor}
\pgfpathmoveto{\pgfpoint{10.868700\du}{-3.537780\du}}
\pgfpathcurveto{\pgfpoint{10.851547\du}{-3.490814\du}}{\pgfpoint{10.787428\du}{-3.461002\du}}{\pgfpoint{10.740462\du}{-3.478155\du}}
\pgfpathcurveto{\pgfpoint{10.693497\du}{-3.495308\du}}{\pgfpoint{10.663684\du}{-3.559427\du}}{\pgfpoint{10.680838\du}{-3.606393\du}}
\pgfpathcurveto{\pgfpoint{10.697991\du}{-3.653358\du}}{\pgfpoint{10.762110\du}{-3.683171\du}}{\pgfpoint{10.809075\du}{-3.666018\du}}
\pgfpathcurveto{\pgfpoint{10.856041\du}{-3.648864\du}}{\pgfpoint{10.885853\du}{-3.584746\du}}{\pgfpoint{10.868700\du}{-3.537780\du}}
\pgfusepath{stroke}
\pgfsetlinewidth{0.050000\du}
\pgfsetdash{}{0pt}
\pgfsetdash{}{0pt}
\pgfsetbuttcap
{
\definecolor{dialinecolor}{rgb}{0.000000, 0.000000, 0.000000}
\pgfsetfillcolor{dialinecolor}
\definecolor{dialinecolor}{rgb}{0.000000, 0.000000, 0.000000}
\pgfsetstrokecolor{dialinecolor}
\draw (13.362500\du,-5.271900\du)--(12.632600\du,-1.920300\du);
}
\pgfsetlinewidth{0.050000\du}
\pgfsetdash{}{0pt}
\pgfsetdash{}{0pt}
\pgfsetbuttcap
{
\definecolor{dialinecolor}{rgb}{0.000000, 0.000000, 0.000000}
\pgfsetfillcolor{dialinecolor}
\definecolor{dialinecolor}{rgb}{0.000000, 0.000000, 0.000000}
\pgfsetstrokecolor{dialinecolor}
\draw (13.356300\du,-5.278100\du)--(9.693750\du,-6.825000\du);
}
\definecolor{dialinecolor}{rgb}{0.000000, 0.000000, 0.000000}
\pgfsetstrokecolor{dialinecolor}
\node[anchor=west] at (7.315630\du,-1.823400\du){};
\definecolor{dialinecolor}{rgb}{0.000000, 0.000000, 0.000000}
\pgfsetstrokecolor{dialinecolor}
\node[anchor=west] at (13.673900\du,-5.146600\du){$v_2$};
\definecolor{dialinecolor}{rgb}{0.000000, 0.000000, 0.000000}
\pgfsetstrokecolor{dialinecolor}
\node[anchor=west] at (14.990700\du,-5.673400\du){};
\definecolor{dialinecolor}{rgb}{0.000000, 0.000000, 0.000000}
\pgfsetstrokecolor{dialinecolor}
\node[anchor=west] at (9.427550\du,-7.247800\du){$v_3$};
\definecolor{dialinecolor}{rgb}{0.000000, 0.000000, 0.000000}
\pgfsetstrokecolor{dialinecolor}
\node[anchor=west] at (10.990600\du,-7.798400\du){};
\definecolor{dialinecolor}{rgb}{0.000000, 0.000000, 0.000000}
\pgfsetstrokecolor{dialinecolor}
\node[anchor=west] at (12.400000\du,-1.200000\du){$v_4$};
\definecolor{dialinecolor}{rgb}{0.000000, 0.000000, 0.000000}
\pgfsetstrokecolor{dialinecolor}
\node[anchor=west] at (13.740700\du,-0.748400\du){};
\definecolor{dialinecolor}{rgb}{0.000000, 0.000000, 0.000000}
\pgfsetstrokecolor{dialinecolor}
\node[anchor=west] at (12.040600\du,-10.362500\du){};
\definecolor{dialinecolor}{rgb}{0.000000, 0.000000, 0.000000}
\pgfsetstrokecolor{dialinecolor}
\node[anchor=west] at (8.815630\du,0.412500\du){};
\pgfsetlinewidth{0.100000\du}
\pgfsetdash{{1.000000\du}{1.000000\du}}{0\du}
\pgfsetdash{{0.100000\du}{0.100000\du}}{0\du}
\pgfsetbuttcap
{
\definecolor{dialinecolor}{rgb}{0.000000, 0.000000, 0.000000}
\pgfsetfillcolor{dialinecolor}
\definecolor{dialinecolor}{rgb}{0.000000, 0.000000, 0.000000}
\pgfsetstrokecolor{dialinecolor}
\draw (12.600000\du,-1.800000\du)--(11.100000\du,-6.231530\du);
}
\definecolor{dialinecolor}{rgb}{0.000000, 0.000000, 0.000000}
\pgfsetstrokecolor{dialinecolor}
\draw (12.600000\du,-1.800000\du)--(11.100000\du,-6.231530\du);
\pgfsetlinewidth{0.100000\du}
\pgfsetdash{}{0pt}
\pgfsetmiterjoin
\pgfsetbuttcap
\definecolor{dialinecolor}{rgb}{0.000000, 0.000000, 0.000000}
\pgfsetfillcolor{dialinecolor}
\pgfpathmoveto{\pgfpoint{12.600000\du}{-1.800000\du}}
\pgfpathcurveto{\pgfpoint{12.552640\du}{-1.783969\du}}{\pgfpoint{12.489248\du}{-1.815299\du}}{\pgfpoint{12.473218\du}{-1.862659\du}}
\pgfpathcurveto{\pgfpoint{12.457187\du}{-1.910020\du}}{\pgfpoint{12.488517\du}{-1.973411\du}}{\pgfpoint{12.535877\du}{-1.989442\du}}
\pgfpathcurveto{\pgfpoint{12.583238\du}{-2.005473\du}}{\pgfpoint{12.646629\du}{-1.974143\du}}{\pgfpoint{12.662659\du}{-1.926782\du}}
\pgfpathcurveto{\pgfpoint{12.678690\du}{-1.879422\du}}{\pgfpoint{12.647360\du}{-1.816031\du}}{\pgfpoint{12.600000\du}{-1.800000\du}}
\pgfusepath{fill}
\definecolor{dialinecolor}{rgb}{0.000000, 0.000000, 0.000000}
\pgfsetstrokecolor{dialinecolor}
\pgfpathmoveto{\pgfpoint{12.600000\du}{-1.800000\du}}
\pgfpathcurveto{\pgfpoint{12.552640\du}{-1.783969\du}}{\pgfpoint{12.489248\du}{-1.815299\du}}{\pgfpoint{12.473218\du}{-1.862659\du}}
\pgfpathcurveto{\pgfpoint{12.457187\du}{-1.910020\du}}{\pgfpoint{12.488517\du}{-1.973411\du}}{\pgfpoint{12.535877\du}{-1.989442\du}}
\pgfpathcurveto{\pgfpoint{12.583238\du}{-2.005473\du}}{\pgfpoint{12.646629\du}{-1.974143\du}}{\pgfpoint{12.662659\du}{-1.926782\du}}
\pgfpathcurveto{\pgfpoint{12.678690\du}{-1.879422\du}}{\pgfpoint{12.647360\du}{-1.816031\du}}{\pgfpoint{12.600000\du}{-1.800000\du}}
\pgfusepath{stroke}
\pgfsetlinewidth{0.100000\du}
\pgfsetdash{}{0pt}
\pgfsetdash{}{0pt}
\pgfsetbuttcap
{
\definecolor{dialinecolor}{rgb}{0.000000, 0.000000, 0.000000}
\pgfsetfillcolor{dialinecolor}
\definecolor{dialinecolor}{rgb}{0.000000, 0.000000, 0.000000}
\pgfsetstrokecolor{dialinecolor}
\draw (9.900000\du,-3.800000\du)--(11.131200\du,-6.231530\du);
}
\definecolor{dialinecolor}{rgb}{0.000000, 0.000000, 0.000000}
\pgfsetstrokecolor{dialinecolor}
\draw (9.900000\du,-3.800000\du)--(11.131200\du,-6.231530\du);
\pgfsetlinewidth{0.100000\du}
\pgfsetdash{}{0pt}
\pgfsetmiterjoin
\pgfsetbuttcap
\definecolor{dialinecolor}{rgb}{0.000000, 0.000000, 0.000000}
\pgfsetfillcolor{dialinecolor}
\pgfpathmoveto{\pgfpoint{9.900000\du}{-3.800000\du}}
\pgfpathcurveto{\pgfpoint{9.855392\du}{-3.822587\du}}{\pgfpoint{9.833372\du}{-3.889781\du}}{\pgfpoint{9.855959\du}{-3.934389\du}}
\pgfpathcurveto{\pgfpoint{9.878546\du}{-3.978996\du}}{\pgfpoint{9.945740\du}{-4.001017\du}}{\pgfpoint{9.990348\du}{-3.978430\du}}
\pgfpathcurveto{\pgfpoint{10.034955\du}{-3.955843\du}}{\pgfpoint{10.056976\du}{-3.888649\du}}{\pgfpoint{10.034389\du}{-3.844041\du}}
\pgfpathcurveto{\pgfpoint{10.011802\du}{-3.799434\du}}{\pgfpoint{9.944608\du}{-3.777413\du}}{\pgfpoint{9.900000\du}{-3.800000\du}}
\pgfusepath{fill}
\definecolor{dialinecolor}{rgb}{0.000000, 0.000000, 0.000000}
\pgfsetstrokecolor{dialinecolor}
\pgfpathmoveto{\pgfpoint{9.900000\du}{-3.800000\du}}
\pgfpathcurveto{\pgfpoint{9.855392\du}{-3.822587\du}}{\pgfpoint{9.833372\du}{-3.889781\du}}{\pgfpoint{9.855959\du}{-3.934389\du}}
\pgfpathcurveto{\pgfpoint{9.878546\du}{-3.978996\du}}{\pgfpoint{9.945740\du}{-4.001017\du}}{\pgfpoint{9.990348\du}{-3.978430\du}}
\pgfpathcurveto{\pgfpoint{10.034955\du}{-3.955843\du}}{\pgfpoint{10.056976\du}{-3.888649\du}}{\pgfpoint{10.034389\du}{-3.844041\du}}
\pgfpathcurveto{\pgfpoint{10.011802\du}{-3.799434\du}}{\pgfpoint{9.944608\du}{-3.777413\du}}{\pgfpoint{9.900000\du}{-3.800000\du}}
\pgfusepath{stroke}
\pgfsetlinewidth{0.050000\du}
\pgfsetdash{{1.000000\du}{1.000000\du}}{0\du}
\pgfsetdash{{0.100000\du}{0.100000\du}}{0\du}
\pgfsetbuttcap
{
\definecolor{dialinecolor}{rgb}{0.000000, 0.000000, 0.000000}
\pgfsetfillcolor{dialinecolor}
\definecolor{dialinecolor}{rgb}{0.000000, 0.000000, 0.000000}
\pgfsetstrokecolor{dialinecolor}
\draw (12.600000\du,-1.800000\du)--(9.700000\du,-6.800000\du);
}
\pgfsetlinewidth{0.100000\du}
\pgfsetdash{}{0pt}
\pgfsetdash{}{0pt}
\pgfsetbuttcap
{
\definecolor{dialinecolor}{rgb}{0.000000, 0.000000, 0.000000}
\pgfsetfillcolor{dialinecolor}
\definecolor{dialinecolor}{rgb}{0.000000, 0.000000, 0.000000}
\pgfsetstrokecolor{dialinecolor}
\draw (9.700000\du,-6.900000\du)--(9.843750\du,-5.069030\du);
}
\definecolor{dialinecolor}{rgb}{0.000000, 0.000000, 0.000000}
\pgfsetstrokecolor{dialinecolor}
\draw (9.700000\du,-6.900000\du)--(9.843750\du,-5.069030\du);
\pgfsetlinewidth{0.100000\du}
\pgfsetdash{}{0pt}
\pgfsetmiterjoin
\pgfsetbuttcap
\definecolor{dialinecolor}{rgb}{0.000000, 0.000000, 0.000000}
\pgfsetfillcolor{dialinecolor}
\pgfpathmoveto{\pgfpoint{9.700000\du}{-6.900000\du}}
\pgfpathcurveto{\pgfpoint{9.749847\du}{-6.903913\du}}{\pgfpoint{9.803607\du}{-6.857980\du}}{\pgfpoint{9.807520\du}{-6.808134\du}}
\pgfpathcurveto{\pgfpoint{9.811434\du}{-6.758287\du}}{\pgfpoint{9.765501\du}{-6.704527\du}}{\pgfpoint{9.715654\du}{-6.700614\du}}
\pgfpathcurveto{\pgfpoint{9.665807\du}{-6.696700\du}}{\pgfpoint{9.612047\du}{-6.742633\du}}{\pgfpoint{9.608134\du}{-6.792480\du}}
\pgfpathcurveto{\pgfpoint{9.604220\du}{-6.842326\du}}{\pgfpoint{9.650153\du}{-6.896087\du}}{\pgfpoint{9.700000\du}{-6.900000\du}}
\pgfusepath{fill}
\definecolor{dialinecolor}{rgb}{0.000000, 0.000000, 0.000000}
\pgfsetstrokecolor{dialinecolor}
\pgfpathmoveto{\pgfpoint{9.700000\du}{-6.900000\du}}
\pgfpathcurveto{\pgfpoint{9.749847\du}{-6.903913\du}}{\pgfpoint{9.803607\du}{-6.857980\du}}{\pgfpoint{9.807520\du}{-6.808134\du}}
\pgfpathcurveto{\pgfpoint{9.811434\du}{-6.758287\du}}{\pgfpoint{9.765501\du}{-6.704527\du}}{\pgfpoint{9.715654\du}{-6.700614\du}}
\pgfpathcurveto{\pgfpoint{9.665807\du}{-6.696700\du}}{\pgfpoint{9.612047\du}{-6.742633\du}}{\pgfpoint{9.608134\du}{-6.792480\du}}
\pgfpathcurveto{\pgfpoint{9.604220\du}{-6.842326\du}}{\pgfpoint{9.650153\du}{-6.896087\du}}{\pgfpoint{9.700000\du}{-6.900000\du}}
\pgfusepath{stroke}
\pgfsetlinewidth{0.100000\du}
\pgfsetdash{}{0pt}
\pgfsetdash{}{0pt}
\pgfsetbuttcap
{
\definecolor{dialinecolor}{rgb}{0.000000, 0.000000, 0.000000}
\pgfsetfillcolor{dialinecolor}
}
\definecolor{dialinecolor}{rgb}{0.000000, 0.000000, 0.000000}
\pgfsetstrokecolor{dialinecolor}
\draw (9.775000\du,-5.137500\du)--(10.375000\du,-4.537500\du);
\pgfsetlinewidth{0.100000\du}
\pgfsetdash{}{0pt}
\pgfsetmiterjoin
\pgfsetbuttcap
\definecolor{dialinecolor}{rgb}{0.000000, 0.000000, 0.000000}
\pgfsetfillcolor{dialinecolor}
\pgfpathmoveto{\pgfpoint{9.775000\du}{-5.137500\du}}
\pgfpathcurveto{\pgfpoint{9.810355\du}{-5.172855\du}}{\pgfpoint{9.881066\du}{-5.172855\du}}{\pgfpoint{9.916421\du}{-5.137500\du}}
\pgfpathcurveto{\pgfpoint{9.951777\du}{-5.102145\du}}{\pgfpoint{9.951777\du}{-5.031434\du}}{\pgfpoint{9.916421\du}{-4.996079\du}}
\pgfpathcurveto{\pgfpoint{9.881066\du}{-4.960723\du}}{\pgfpoint{9.810355\du}{-4.960723\du}}{\pgfpoint{9.775000\du}{-4.996079\du}}
\pgfpathcurveto{\pgfpoint{9.739645\du}{-5.031434\du}}{\pgfpoint{9.739645\du}{-5.102145\du}}{\pgfpoint{9.775000\du}{-5.137500\du}}
\pgfusepath{fill}
\definecolor{dialinecolor}{rgb}{0.000000, 0.000000, 0.000000}
\pgfsetstrokecolor{dialinecolor}
\pgfpathmoveto{\pgfpoint{9.775000\du}{-5.137500\du}}
\pgfpathcurveto{\pgfpoint{9.810355\du}{-5.172855\du}}{\pgfpoint{9.881066\du}{-5.172855\du}}{\pgfpoint{9.916421\du}{-5.137500\du}}
\pgfpathcurveto{\pgfpoint{9.951777\du}{-5.102145\du}}{\pgfpoint{9.951777\du}{-5.031434\du}}{\pgfpoint{9.916421\du}{-4.996079\du}}
\pgfpathcurveto{\pgfpoint{9.881066\du}{-4.960723\du}}{\pgfpoint{9.810355\du}{-4.960723\du}}{\pgfpoint{9.775000\du}{-4.996079\du}}
\pgfpathcurveto{\pgfpoint{9.739645\du}{-5.031434\du}}{\pgfpoint{9.739645\du}{-5.102145\du}}{\pgfpoint{9.775000\du}{-5.137500\du}}
\pgfusepath{stroke}
\pgfsetlinewidth{0.100000\du}
\pgfsetdash{}{0pt}
\pgfsetmiterjoin
\pgfsetbuttcap
\definecolor{dialinecolor}{rgb}{0.000000, 0.000000, 0.000000}
\pgfsetfillcolor{dialinecolor}
\pgfpathmoveto{\pgfpoint{10.375000\du}{-4.537500\du}}
\pgfpathcurveto{\pgfpoint{10.339645\du}{-4.502145\du}}{\pgfpoint{10.268934\du}{-4.502145\du}}{\pgfpoint{10.233579\du}{-4.537500\du}}
\pgfpathcurveto{\pgfpoint{10.198223\du}{-4.572855\du}}{\pgfpoint{10.198223\du}{-4.643566\du}}{\pgfpoint{10.233579\du}{-4.678921\du}}
\pgfpathcurveto{\pgfpoint{10.268934\du}{-4.714277\du}}{\pgfpoint{10.339645\du}{-4.714277\du}}{\pgfpoint{10.375000\du}{-4.678921\du}}
\pgfpathcurveto{\pgfpoint{10.410355\du}{-4.643566\du}}{\pgfpoint{10.410355\du}{-4.572855\du}}{\pgfpoint{10.375000\du}{-4.537500\du}}
\pgfusepath{fill}
\definecolor{dialinecolor}{rgb}{0.000000, 0.000000, 0.000000}
\pgfsetstrokecolor{dialinecolor}
\pgfpathmoveto{\pgfpoint{10.375000\du}{-4.537500\du}}
\pgfpathcurveto{\pgfpoint{10.339645\du}{-4.502145\du}}{\pgfpoint{10.268934\du}{-4.502145\du}}{\pgfpoint{10.233579\du}{-4.537500\du}}
\pgfpathcurveto{\pgfpoint{10.198223\du}{-4.572855\du}}{\pgfpoint{10.198223\du}{-4.643566\du}}{\pgfpoint{10.233579\du}{-4.678921\du}}
\pgfpathcurveto{\pgfpoint{10.268934\du}{-4.714277\du}}{\pgfpoint{10.339645\du}{-4.714277\du}}{\pgfpoint{10.375000\du}{-4.678921\du}}
\pgfpathcurveto{\pgfpoint{10.410355\du}{-4.643566\du}}{\pgfpoint{10.410355\du}{-4.572855\du}}{\pgfpoint{10.375000\du}{-4.537500\du}}
\pgfusepath{stroke}
\pgfsetlinewidth{0.050000\du}
\pgfsetdash{{\pgflinewidth}{0.200000\du}}{0cm}
\pgfsetdash{{\pgflinewidth}{0.200000\du}}{0cm}
\pgfsetbuttcap
{
\definecolor{dialinecolor}{rgb}{0.000000, 0.000000, 0.000000}
\pgfsetfillcolor{dialinecolor}
\definecolor{dialinecolor}{rgb}{0.000000, 0.000000, 0.000000}
\pgfsetstrokecolor{dialinecolor}
\pgfpathmoveto{\pgfpoint{10.143780\du}{-4.300012\du}}
\pgfpatharc{121}{-65}{0.338820\du and 0.338820\du}
\pgfusepath{stroke}
}
\definecolor{dialinecolor}{rgb}{0.000000, 0.000000, 0.000000}
\pgfsetstrokecolor{dialinecolor}
\node[anchor=west] at (9.600000\du,-4.400000\du){$y$};
\definecolor{dialinecolor}{rgb}{0.000000, 0.000000, 0.000000}
\pgfsetstrokecolor{dialinecolor}
\node[anchor=west] at (9.500000\du,-3.500000\du){$x$};
\definecolor{dialinecolor}{rgb}{0.000000, 0.000000, 0.000000}
\pgfsetstrokecolor{dialinecolor}
\node[anchor=west] at (9.256250\du,-4.798750\du){$w$};
\definecolor{dialinecolor}{rgb}{0.000000, 0.000000, 0.000000}
\pgfsetstrokecolor{dialinecolor}
\node[anchor=west] at (4.800000\du,-6.400000\du){$v_1$};
\definecolor{dialinecolor}{rgb}{0.000000, 0.000000, 0.000000}
\pgfsetstrokecolor{dialinecolor}
\node[anchor=west] at (4.800000\du,-1.000000\du){$v_2$};
\definecolor{dialinecolor}{rgb}{0.000000, 0.000000, 0.000000}
\pgfsetstrokecolor{dialinecolor}
\node[anchor=west] at (1.400000\du,-1.600000\du){$v_3$};
\definecolor{dialinecolor}{rgb}{0.000000, 0.000000, 0.000000}
\pgfsetstrokecolor{dialinecolor}
\node[anchor=west] at (1.600000\du,-5.800000\du){$v_4$};
\definecolor{dialinecolor}{rgb}{0.000000, 0.000000, 0.000000}
\pgfsetstrokecolor{dialinecolor}
\node[anchor=west] at (3.000000\du,-4.800000\du){};
\definecolor{dialinecolor}{rgb}{0.000000, 0.000000, 0.000000}
\pgfsetstrokecolor{dialinecolor}
\node[anchor=west] at (9.400000\du,-5.000000\du){};
\end{tikzpicture}

%% file: figure/X.tex
\ifx\du\undefined
  \newlength{\du}
\fi
\setlength{\du}{15\unitlength}
\begin{tikzpicture}
\pgftransformxscale{1.000000}
\pgftransformyscale{-1.000000}
\definecolor{dialinecolor}{rgb}{0.000000, 0.000000, 0.000000}
\pgfsetstrokecolor{dialinecolor}
\definecolor{dialinecolor}{rgb}{1.000000, 1.000000, 1.000000}
\pgfsetfillcolor{dialinecolor}
\definecolor{dialinecolor}{rgb}{0.000000, 0.000000, 0.000000}
\pgfsetstrokecolor{dialinecolor}
\node[anchor=west] at (6.000000\du,-2.000000\du){};
\pgfsetlinewidth{0.050000\du}
\pgfsetdash{}{0pt}
\pgfsetdash{}{0pt}
\pgfsetbuttcap
{
\definecolor{dialinecolor}{rgb}{0.000000, 0.000000, 0.000000}
\pgfsetfillcolor{dialinecolor}
}
\definecolor{dialinecolor}{rgb}{0.000000, 0.000000, 0.000000}
\pgfsetstrokecolor{dialinecolor}
\draw (3.600000\du,-3.600000\du)--(1.800000\du,-5.400000\du);
\pgfsetlinewidth{0.050000\du}
\pgfsetdash{}{0pt}
\pgfsetmiterjoin
\pgfsetbuttcap
\definecolor{dialinecolor}{rgb}{0.000000, 0.000000, 0.000000}
\pgfsetfillcolor{dialinecolor}
\pgfpathmoveto{\pgfpoint{3.600000\du}{-3.600000\du}}
\pgfpathcurveto{\pgfpoint{3.564645\du}{-3.564645\du}}{\pgfpoint{3.493934\du}{-3.564645\du}}{\pgfpoint{3.458579\du}{-3.600000\du}}
\pgfpathcurveto{\pgfpoint{3.423223\du}{-3.635355\du}}{\pgfpoint{3.423223\du}{-3.706066\du}}{\pgfpoint{3.458579\du}{-3.741421\du}}
\pgfpathcurveto{\pgfpoint{3.493934\du}{-3.776777\du}}{\pgfpoint{3.564645\du}{-3.776777\du}}{\pgfpoint{3.600000\du}{-3.741421\du}}
\pgfpathcurveto{\pgfpoint{3.635355\du}{-3.706066\du}}{\pgfpoint{3.635355\du}{-3.635355\du}}{\pgfpoint{3.600000\du}{-3.600000\du}}
\pgfusepath{fill}
\definecolor{dialinecolor}{rgb}{0.000000, 0.000000, 0.000000}
\pgfsetstrokecolor{dialinecolor}
\pgfpathmoveto{\pgfpoint{3.600000\du}{-3.600000\du}}
\pgfpathcurveto{\pgfpoint{3.564645\du}{-3.564645\du}}{\pgfpoint{3.493934\du}{-3.564645\du}}{\pgfpoint{3.458579\du}{-3.600000\du}}
\pgfpathcurveto{\pgfpoint{3.423223\du}{-3.635355\du}}{\pgfpoint{3.423223\du}{-3.706066\du}}{\pgfpoint{3.458579\du}{-3.741421\du}}
\pgfpathcurveto{\pgfpoint{3.493934\du}{-3.776777\du}}{\pgfpoint{3.564645\du}{-3.776777\du}}{\pgfpoint{3.600000\du}{-3.741421\du}}
\pgfpathcurveto{\pgfpoint{3.635355\du}{-3.706066\du}}{\pgfpoint{3.635355\du}{-3.635355\du}}{\pgfpoint{3.600000\du}{-3.600000\du}}
\pgfusepath{stroke}
\pgfsetlinewidth{0.050000\du}
\pgfsetdash{}{0pt}
\pgfsetmiterjoin
\pgfsetbuttcap
\definecolor{dialinecolor}{rgb}{0.000000, 0.000000, 0.000000}
\pgfsetfillcolor{dialinecolor}
\pgfpathmoveto{\pgfpoint{1.800000\du}{-5.400000\du}}
\pgfpathcurveto{\pgfpoint{1.835355\du}{-5.435355\du}}{\pgfpoint{1.906066\du}{-5.435355\du}}{\pgfpoint{1.941421\du}{-5.400000\du}}
\pgfpathcurveto{\pgfpoint{1.976777\du}{-5.364645\du}}{\pgfpoint{1.976777\du}{-5.293934\du}}{\pgfpoint{1.941421\du}{-5.258579\du}}
\pgfpathcurveto{\pgfpoint{1.906066\du}{-5.223223\du}}{\pgfpoint{1.835355\du}{-5.223223\du}}{\pgfpoint{1.800000\du}{-5.258579\du}}
\pgfpathcurveto{\pgfpoint{1.764645\du}{-5.293934\du}}{\pgfpoint{1.764645\du}{-5.364645\du}}{\pgfpoint{1.800000\du}{-5.400000\du}}
\pgfusepath{fill}
\definecolor{dialinecolor}{rgb}{0.000000, 0.000000, 0.000000}
\pgfsetstrokecolor{dialinecolor}
\pgfpathmoveto{\pgfpoint{1.800000\du}{-5.400000\du}}
\pgfpathcurveto{\pgfpoint{1.835355\du}{-5.435355\du}}{\pgfpoint{1.906066\du}{-5.435355\du}}{\pgfpoint{1.941421\du}{-5.400000\du}}
\pgfpathcurveto{\pgfpoint{1.976777\du}{-5.364645\du}}{\pgfpoint{1.976777\du}{-5.293934\du}}{\pgfpoint{1.941421\du}{-5.258579\du}}
\pgfpathcurveto{\pgfpoint{1.906066\du}{-5.223223\du}}{\pgfpoint{1.835355\du}{-5.223223\du}}{\pgfpoint{1.800000\du}{-5.258579\du}}
\pgfpathcurveto{\pgfpoint{1.764645\du}{-5.293934\du}}{\pgfpoint{1.764645\du}{-5.364645\du}}{\pgfpoint{1.800000\du}{-5.400000\du}}
\pgfusepath{stroke}
\definecolor{dialinecolor}{rgb}{0.000000, 0.000000, 0.000000}
\pgfsetstrokecolor{dialinecolor}
\node[anchor=west] at (1.800000\du,-5.800000\du){};
\pgfsetlinewidth{0.050000\du}
\pgfsetdash{}{0pt}
\pgfsetdash{}{0pt}
\pgfsetbuttcap
{
\definecolor{dialinecolor}{rgb}{0.000000, 0.000000, 0.000000}
\pgfsetfillcolor{dialinecolor}
\definecolor{dialinecolor}{rgb}{0.000000, 0.000000, 0.000000}
\pgfsetstrokecolor{dialinecolor}
\draw (1.800000\du,-2.000000\du)--(3.537500\du,-3.669525\du);
}
\definecolor{dialinecolor}{rgb}{0.000000, 0.000000, 0.000000}
\pgfsetstrokecolor{dialinecolor}
\draw (1.800000\du,-2.000000\du)--(3.537500\du,-3.669525\du);
\pgfsetlinewidth{0.050000\du}
\pgfsetdash{}{0pt}
\pgfsetmiterjoin
\pgfsetbuttcap
\definecolor{dialinecolor}{rgb}{0.000000, 0.000000, 0.000000}
\pgfsetfillcolor{dialinecolor}
\pgfpathmoveto{\pgfpoint{1.800000\du}{-2.000000\du}}
\pgfpathcurveto{\pgfpoint{1.765357\du}{-2.036054\du}}{\pgfpoint{1.766767\du}{-2.106750\du}}{\pgfpoint{1.802821\du}{-2.141393\du}}
\pgfpathcurveto{\pgfpoint{1.838875\du}{-2.176036\du}}{\pgfpoint{1.909571\du}{-2.174626\du}}{\pgfpoint{1.944214\du}{-2.138572\du}}
\pgfpathcurveto{\pgfpoint{1.978857\du}{-2.102519\du}}{\pgfpoint{1.977447\du}{-2.031822\du}}{\pgfpoint{1.941393\du}{-1.997179\du}}
\pgfpathcurveto{\pgfpoint{1.905340\du}{-1.962536\du}}{\pgfpoint{1.834643\du}{-1.963946\du}}{\pgfpoint{1.800000\du}{-2.000000\du}}
\pgfusepath{fill}
\definecolor{dialinecolor}{rgb}{0.000000, 0.000000, 0.000000}
\pgfsetstrokecolor{dialinecolor}
\pgfpathmoveto{\pgfpoint{1.800000\du}{-2.000000\du}}
\pgfpathcurveto{\pgfpoint{1.765357\du}{-2.036054\du}}{\pgfpoint{1.766767\du}{-2.106750\du}}{\pgfpoint{1.802821\du}{-2.141393\du}}
\pgfpathcurveto{\pgfpoint{1.838875\du}{-2.176036\du}}{\pgfpoint{1.909571\du}{-2.174626\du}}{\pgfpoint{1.944214\du}{-2.138572\du}}
\pgfpathcurveto{\pgfpoint{1.978857\du}{-2.102519\du}}{\pgfpoint{1.977447\du}{-2.031822\du}}{\pgfpoint{1.941393\du}{-1.997179\du}}
\pgfpathcurveto{\pgfpoint{1.905340\du}{-1.962536\du}}{\pgfpoint{1.834643\du}{-1.963946\du}}{\pgfpoint{1.800000\du}{-2.000000\du}}
\pgfusepath{stroke}
\definecolor{dialinecolor}{rgb}{0.000000, 0.000000, 0.000000}
\pgfsetstrokecolor{dialinecolor}
\node[anchor=west] at (2.315000\du,-5.382500\du){};
\definecolor{dialinecolor}{rgb}{0.000000, 0.000000, 0.000000}
\pgfsetstrokecolor{dialinecolor}
\node[anchor=west] at (1.600000\du,-1.400000\du){};
\definecolor{dialinecolor}{rgb}{0.000000, 0.000000, 0.000000}
\pgfsetstrokecolor{dialinecolor}
\node[anchor=west] at (5.400000\du,-6.200000\du){};
\definecolor{dialinecolor}{rgb}{0.000000, 0.000000, 0.000000}
\pgfsetstrokecolor{dialinecolor}
\node[anchor=west] at (6.200000\du,-6.600000\du){};
\definecolor{dialinecolor}{rgb}{0.000000, 0.000000, 0.000000}
\pgfsetstrokecolor{dialinecolor}
\node[anchor=west] at (3.300000\du,-4.300000\du){$r$};
\definecolor{dialinecolor}{rgb}{0.000000, 0.000000, 0.000000}
\pgfsetstrokecolor{dialinecolor}
\node[anchor=west] at (3.000000\du,-5.000000\du){};
\definecolor{dialinecolor}{rgb}{0.000000, 0.000000, 0.000000}
\pgfsetstrokecolor{dialinecolor}
\node[anchor=west] at (6.157550\du,-2.545300\du){};
\definecolor{dialinecolor}{rgb}{0.000000, 0.000000, 0.000000}
\pgfsetstrokecolor{dialinecolor}
\node[anchor=west] at (6.232550\du,-7.920300\du){};
\definecolor{dialinecolor}{rgb}{0.000000, 0.000000, 0.000000}
\pgfsetstrokecolor{dialinecolor}
\node[anchor=west] at (7.032550\du,-7.920300\du){};
\definecolor{dialinecolor}{rgb}{0.000000, 0.000000, 0.000000}
\pgfsetstrokecolor{dialinecolor}
\node[anchor=west] at (14.990700\du,-5.673400\du){};
\definecolor{dialinecolor}{rgb}{0.000000, 0.000000, 0.000000}
\pgfsetstrokecolor{dialinecolor}
\node[anchor=west] at (10.990600\du,-7.798400\du){};
\definecolor{dialinecolor}{rgb}{0.000000, 0.000000, 0.000000}
\pgfsetstrokecolor{dialinecolor}
\node[anchor=west] at (13.740700\du,-0.748400\du){};
\definecolor{dialinecolor}{rgb}{0.000000, 0.000000, 0.000000}
\pgfsetstrokecolor{dialinecolor}
\node[anchor=west] at (12.040600\du,-10.362500\du){};
\definecolor{dialinecolor}{rgb}{0.000000, 0.000000, 0.000000}
\pgfsetstrokecolor{dialinecolor}
\node[anchor=west] at (8.815630\du,0.412500\du){};
\definecolor{dialinecolor}{rgb}{0.000000, 0.000000, 0.000000}
\pgfsetstrokecolor{dialinecolor}
\node[anchor=west] at (4.600000\du,-5.800000\du){$v_1$};
\definecolor{dialinecolor}{rgb}{0.000000, 0.000000, 0.000000}
\pgfsetstrokecolor{dialinecolor}
\node[anchor=west] at (4.900000\du,-1.400000\du){$v_2$};
\definecolor{dialinecolor}{rgb}{0.000000, 0.000000, 0.000000}
\pgfsetstrokecolor{dialinecolor}
\node[anchor=west] at (1.600000\du,-1.400000\du){$v_3$};
\definecolor{dialinecolor}{rgb}{0.000000, 0.000000, 0.000000}
\pgfsetstrokecolor{dialinecolor}
\node[anchor=west] at (1.600000\du,-5.800000\du){$v_4$};
\definecolor{dialinecolor}{rgb}{0.000000, 0.000000, 0.000000}
\pgfsetstrokecolor{dialinecolor}
\node[anchor=west] at (3.000000\du,-4.800000\du){};
\definecolor{dialinecolor}{rgb}{0.000000, 0.000000, 0.000000}
\pgfsetstrokecolor{dialinecolor}
\node[anchor=west] at (2.400000\du,-4.000000\du){};
\pgfsetlinewidth{0.050000\du}
\pgfsetdash{}{0pt}
\pgfsetdash{}{0pt}
\pgfsetbuttcap
{
\definecolor{dialinecolor}{rgb}{0.000000, 0.000000, 0.000000}
\pgfsetfillcolor{dialinecolor}
\definecolor{dialinecolor}{rgb}{0.000000, 0.000000, 0.000000}
\pgfsetstrokecolor{dialinecolor}
\draw (5.200000\du,-2.000000\du)--(3.537500\du,-3.669525\du);
}
\definecolor{dialinecolor}{rgb}{0.000000, 0.000000, 0.000000}
\pgfsetstrokecolor{dialinecolor}
\draw (5.200000\du,-2.000000\du)--(3.537500\du,-3.669525\du);
\pgfsetlinewidth{0.050000\du}
\pgfsetdash{}{0pt}
\pgfsetmiterjoin
\pgfsetbuttcap
\definecolor{dialinecolor}{rgb}{0.000000, 0.000000, 0.000000}
\pgfsetfillcolor{dialinecolor}
\pgfpathmoveto{\pgfpoint{5.200000\du}{-2.000000\du}}
\pgfpathcurveto{\pgfpoint{5.164570\du}{-1.964719\du}}{\pgfpoint{5.093860\du}{-1.964868\du}}{\pgfpoint{5.058579\du}{-2.000298\du}}
\pgfpathcurveto{\pgfpoint{5.023298\du}{-2.035728\du}}{\pgfpoint{5.023447\du}{-2.106438\du}}{\pgfpoint{5.058877\du}{-2.141719\du}}
\pgfpathcurveto{\pgfpoint{5.094307\du}{-2.177000\du}}{\pgfpoint{5.165017\du}{-2.176851\du}}{\pgfpoint{5.200298\du}{-2.141421\du}}
\pgfpathcurveto{\pgfpoint{5.235579\du}{-2.105991\du}}{\pgfpoint{5.235430\du}{-2.035281\du}}{\pgfpoint{5.200000\du}{-2.000000\du}}
\pgfusepath{fill}
\definecolor{dialinecolor}{rgb}{0.000000, 0.000000, 0.000000}
\pgfsetstrokecolor{dialinecolor}
\pgfpathmoveto{\pgfpoint{5.200000\du}{-2.000000\du}}
\pgfpathcurveto{\pgfpoint{5.164570\du}{-1.964719\du}}{\pgfpoint{5.093860\du}{-1.964868\du}}{\pgfpoint{5.058579\du}{-2.000298\du}}
\pgfpathcurveto{\pgfpoint{5.023298\du}{-2.035728\du}}{\pgfpoint{5.023447\du}{-2.106438\du}}{\pgfpoint{5.058877\du}{-2.141719\du}}
\pgfpathcurveto{\pgfpoint{5.094307\du}{-2.177000\du}}{\pgfpoint{5.165017\du}{-2.176851\du}}{\pgfpoint{5.200298\du}{-2.141421\du}}
\pgfpathcurveto{\pgfpoint{5.235579\du}{-2.105991\du}}{\pgfpoint{5.235430\du}{-2.035281\du}}{\pgfpoint{5.200000\du}{-2.000000\du}}
\pgfusepath{stroke}
\pgfsetlinewidth{0.050000\du}
\pgfsetdash{}{0pt}
\pgfsetdash{}{0pt}
\pgfsetbuttcap
{
\definecolor{dialinecolor}{rgb}{0.000000, 0.000000, 0.000000}
\pgfsetfillcolor{dialinecolor}
\definecolor{dialinecolor}{rgb}{0.000000, 0.000000, 0.000000}
\pgfsetstrokecolor{dialinecolor}
\draw (5.200000\du,-5.400000\du)--(3.537500\du,-3.675775\du);
}
\definecolor{dialinecolor}{rgb}{0.000000, 0.000000, 0.000000}
\pgfsetstrokecolor{dialinecolor}
\draw (5.200000\du,-5.400000\du)--(3.537500\du,-3.675775\du);
\pgfsetlinewidth{0.050000\du}
\pgfsetdash{}{0pt}
\pgfsetmiterjoin
\pgfsetbuttcap
\definecolor{dialinecolor}{rgb}{0.000000, 0.000000, 0.000000}
\pgfsetfillcolor{dialinecolor}
\pgfpathmoveto{\pgfpoint{5.200000\du}{-5.400000\du}}
\pgfpathcurveto{\pgfpoint{5.235994\du}{-5.365295\du}}{\pgfpoint{5.237282\du}{-5.294596\du}}{\pgfpoint{5.202577\du}{-5.258602\du}}
\pgfpathcurveto{\pgfpoint{5.167872\du}{-5.222608\du}}{\pgfpoint{5.097173\du}{-5.221320\du}}{\pgfpoint{5.061179\du}{-5.256025\du}}
\pgfpathcurveto{\pgfpoint{5.025185\du}{-5.290730\du}}{\pgfpoint{5.023897\du}{-5.361429\du}}{\pgfpoint{5.058602\du}{-5.397423\du}}
\pgfpathcurveto{\pgfpoint{5.093307\du}{-5.433417\du}}{\pgfpoint{5.164006\du}{-5.434705\du}}{\pgfpoint{5.200000\du}{-5.400000\du}}
\pgfusepath{fill}
\definecolor{dialinecolor}{rgb}{0.000000, 0.000000, 0.000000}
\pgfsetstrokecolor{dialinecolor}
\pgfpathmoveto{\pgfpoint{5.200000\du}{-5.400000\du}}
\pgfpathcurveto{\pgfpoint{5.235994\du}{-5.365295\du}}{\pgfpoint{5.237282\du}{-5.294596\du}}{\pgfpoint{5.202577\du}{-5.258602\du}}
\pgfpathcurveto{\pgfpoint{5.167872\du}{-5.222608\du}}{\pgfpoint{5.097173\du}{-5.221320\du}}{\pgfpoint{5.061179\du}{-5.256025\du}}
\pgfpathcurveto{\pgfpoint{5.025185\du}{-5.290730\du}}{\pgfpoint{5.023897\du}{-5.361429\du}}{\pgfpoint{5.058602\du}{-5.397423\du}}
\pgfpathcurveto{\pgfpoint{5.093307\du}{-5.433417\du}}{\pgfpoint{5.164006\du}{-5.434705\du}}{\pgfpoint{5.200000\du}{-5.400000\du}}
\pgfusepath{stroke}
\end{tikzpicture}

%% file: figure/Y.tex
\ifx\du\undefined
  \newlength{\du}
\fi
\setlength{\du}{15\unitlength}
\begin{tikzpicture}
\pgftransformxscale{1.000000}
\pgftransformyscale{-1.000000}
\definecolor{dialinecolor}{rgb}{0.000000, 0.000000, 0.000000}
\pgfsetstrokecolor{dialinecolor}
\definecolor{dialinecolor}{rgb}{1.000000, 1.000000, 1.000000}
\pgfsetfillcolor{dialinecolor}
\definecolor{dialinecolor}{rgb}{0.000000, 0.000000, 0.000000}
\pgfsetstrokecolor{dialinecolor}
\node[anchor=west] at (6.000000\du,-2.000000\du){};
\pgfsetlinewidth{0.050000\du}
\pgfsetdash{}{0pt}
\pgfsetdash{}{0pt}
\pgfsetbuttcap
{
\definecolor{dialinecolor}{rgb}{0.000000, 0.000000, 0.000000}
\pgfsetfillcolor{dialinecolor}
\definecolor{dialinecolor}{rgb}{0.000000, 0.000000, 0.000000}
\pgfsetstrokecolor{dialinecolor}
\draw (5.168750\du,-5.399213\du)--(1.800000\du,-5.400000\du);
}
\definecolor{dialinecolor}{rgb}{0.000000, 0.000000, 0.000000}
\pgfsetstrokecolor{dialinecolor}
\draw (5.168750\du,-5.399213\du)--(1.800000\du,-5.400000\du);
\pgfsetlinewidth{0.050000\du}
\pgfsetdash{}{0pt}
\pgfsetmiterjoin
\pgfsetbuttcap
\definecolor{dialinecolor}{rgb}{0.000000, 0.000000, 0.000000}
\pgfsetfillcolor{dialinecolor}
\pgfpathmoveto{\pgfpoint{1.800000\du}{-5.400000\du}}
\pgfpathcurveto{\pgfpoint{1.800012\du}{-5.450000\du}}{\pgfpoint{1.850023\du}{-5.499988\du}}{\pgfpoint{1.900023\du}{-5.499977\du}}
\pgfpathcurveto{\pgfpoint{1.950023\du}{-5.499965\du}}{\pgfpoint{2.000012\du}{-5.449953\du}}{\pgfpoint{2.000000\du}{-5.399953\du}}
\pgfpathcurveto{\pgfpoint{1.999988\du}{-5.349953\du}}{\pgfpoint{1.949977\du}{-5.299965\du}}{\pgfpoint{1.899977\du}{-5.299977\du}}
\pgfpathcurveto{\pgfpoint{1.849977\du}{-5.299988\du}}{\pgfpoint{1.799988\du}{-5.350000\du}}{\pgfpoint{1.800000\du}{-5.400000\du}}
\pgfusepath{fill}
\definecolor{dialinecolor}{rgb}{0.000000, 0.000000, 0.000000}
\pgfsetstrokecolor{dialinecolor}
\pgfpathmoveto{\pgfpoint{1.800000\du}{-5.400000\du}}
\pgfpathcurveto{\pgfpoint{1.800012\du}{-5.450000\du}}{\pgfpoint{1.850023\du}{-5.499988\du}}{\pgfpoint{1.900023\du}{-5.499977\du}}
\pgfpathcurveto{\pgfpoint{1.950023\du}{-5.499965\du}}{\pgfpoint{2.000012\du}{-5.449953\du}}{\pgfpoint{2.000000\du}{-5.399953\du}}
\pgfpathcurveto{\pgfpoint{1.999988\du}{-5.349953\du}}{\pgfpoint{1.949977\du}{-5.299965\du}}{\pgfpoint{1.899977\du}{-5.299977\du}}
\pgfpathcurveto{\pgfpoint{1.849977\du}{-5.299988\du}}{\pgfpoint{1.799988\du}{-5.350000\du}}{\pgfpoint{1.800000\du}{-5.400000\du}}
\pgfusepath{stroke}
\definecolor{dialinecolor}{rgb}{0.000000, 0.000000, 0.000000}
\pgfsetstrokecolor{dialinecolor}
\node[anchor=west] at (1.800000\du,-5.800000\du){};
\pgfsetlinewidth{0.050000\du}
\pgfsetdash{}{0pt}
\pgfsetdash{}{0pt}
\pgfsetbuttcap
{
\definecolor{dialinecolor}{rgb}{0.000000, 0.000000, 0.000000}
\pgfsetfillcolor{dialinecolor}
}
\definecolor{dialinecolor}{rgb}{0.000000, 0.000000, 0.000000}
\pgfsetstrokecolor{dialinecolor}
\draw (1.800000\du,-2.000000\du)--(5.243750\du,-5.474213\du);
\pgfsetlinewidth{0.050000\du}
\pgfsetdash{}{0pt}
\pgfsetmiterjoin
\pgfsetbuttcap
\definecolor{dialinecolor}{rgb}{0.000000, 0.000000, 0.000000}
\pgfsetfillcolor{dialinecolor}
\pgfpathmoveto{\pgfpoint{1.800000\du}{-2.000000\du}}
\pgfpathcurveto{\pgfpoint{1.764489\du}{-2.035199\du}}{\pgfpoint{1.764178\du}{-2.105909\du}}{\pgfpoint{1.799377\du}{-2.141420\du}}
\pgfpathcurveto{\pgfpoint{1.834577\du}{-2.176931\du}}{\pgfpoint{1.905287\du}{-2.177242\du}}{\pgfpoint{1.940797\du}{-2.142043\du}}
\pgfpathcurveto{\pgfpoint{1.976308\du}{-2.106843\du}}{\pgfpoint{1.976619\du}{-2.036133\du}}{\pgfpoint{1.941420\du}{-2.000623\du}}
\pgfpathcurveto{\pgfpoint{1.906221\du}{-1.965112\du}}{\pgfpoint{1.835511\du}{-1.964801\du}}{\pgfpoint{1.800000\du}{-2.000000\du}}
\pgfusepath{fill}
\definecolor{dialinecolor}{rgb}{0.000000, 0.000000, 0.000000}
\pgfsetstrokecolor{dialinecolor}
\pgfpathmoveto{\pgfpoint{1.800000\du}{-2.000000\du}}
\pgfpathcurveto{\pgfpoint{1.764489\du}{-2.035199\du}}{\pgfpoint{1.764178\du}{-2.105909\du}}{\pgfpoint{1.799377\du}{-2.141420\du}}
\pgfpathcurveto{\pgfpoint{1.834577\du}{-2.176931\du}}{\pgfpoint{1.905287\du}{-2.177242\du}}{\pgfpoint{1.940797\du}{-2.142043\du}}
\pgfpathcurveto{\pgfpoint{1.976308\du}{-2.106843\du}}{\pgfpoint{1.976619\du}{-2.036133\du}}{\pgfpoint{1.941420\du}{-2.000623\du}}
\pgfpathcurveto{\pgfpoint{1.906221\du}{-1.965112\du}}{\pgfpoint{1.835511\du}{-1.964801\du}}{\pgfpoint{1.800000\du}{-2.000000\du}}
\pgfusepath{stroke}
\pgfsetlinewidth{0.050000\du}
\pgfsetdash{}{0pt}
\pgfsetmiterjoin
\pgfsetbuttcap
\definecolor{dialinecolor}{rgb}{0.000000, 0.000000, 0.000000}
\pgfsetfillcolor{dialinecolor}
\pgfpathmoveto{\pgfpoint{5.243750\du}{-5.474213\du}}
\pgfpathcurveto{\pgfpoint{5.279261\du}{-5.439013\du}}{\pgfpoint{5.279572\du}{-5.368303\du}}{\pgfpoint{5.244373\du}{-5.332793\du}}
\pgfpathcurveto{\pgfpoint{5.209173\du}{-5.297282\du}}{\pgfpoint{5.138463\du}{-5.296970\du}}{\pgfpoint{5.102953\du}{-5.332170\du}}
\pgfpathcurveto{\pgfpoint{5.067442\du}{-5.367369\du}}{\pgfpoint{5.067131\du}{-5.438079\du}}{\pgfpoint{5.102330\du}{-5.473590\du}}
\pgfpathcurveto{\pgfpoint{5.137529\du}{-5.509100\du}}{\pgfpoint{5.208239\du}{-5.509412\du}}{\pgfpoint{5.243750\du}{-5.474213\du}}
\pgfusepath{fill}
\definecolor{dialinecolor}{rgb}{0.000000, 0.000000, 0.000000}
\pgfsetstrokecolor{dialinecolor}
\pgfpathmoveto{\pgfpoint{5.243750\du}{-5.474213\du}}
\pgfpathcurveto{\pgfpoint{5.279261\du}{-5.439013\du}}{\pgfpoint{5.279572\du}{-5.368303\du}}{\pgfpoint{5.244373\du}{-5.332793\du}}
\pgfpathcurveto{\pgfpoint{5.209173\du}{-5.297282\du}}{\pgfpoint{5.138463\du}{-5.296970\du}}{\pgfpoint{5.102953\du}{-5.332170\du}}
\pgfpathcurveto{\pgfpoint{5.067442\du}{-5.367369\du}}{\pgfpoint{5.067131\du}{-5.438079\du}}{\pgfpoint{5.102330\du}{-5.473590\du}}
\pgfpathcurveto{\pgfpoint{5.137529\du}{-5.509100\du}}{\pgfpoint{5.208239\du}{-5.509412\du}}{\pgfpoint{5.243750\du}{-5.474213\du}}
\pgfusepath{stroke}
\definecolor{dialinecolor}{rgb}{0.000000, 0.000000, 0.000000}
\pgfsetstrokecolor{dialinecolor}
\node[anchor=west] at (2.315000\du,-5.382500\du){};
\definecolor{dialinecolor}{rgb}{0.000000, 0.000000, 0.000000}
\pgfsetstrokecolor{dialinecolor}
\node[anchor=west] at (1.600000\du,-1.400000\du){};
\definecolor{dialinecolor}{rgb}{0.000000, 0.000000, 0.000000}
\pgfsetstrokecolor{dialinecolor}
\node[anchor=west] at (5.400000\du,-6.200000\du){};
\definecolor{dialinecolor}{rgb}{0.000000, 0.000000, 0.000000}
\pgfsetstrokecolor{dialinecolor}
\node[anchor=west] at (6.200000\du,-6.600000\du){};
\definecolor{dialinecolor}{rgb}{0.000000, 0.000000, 0.000000}
\pgfsetstrokecolor{dialinecolor}
\node[anchor=west] at (3.000000\du,-5.000000\du){};
\definecolor{dialinecolor}{rgb}{0.000000, 0.000000, 0.000000}
\pgfsetstrokecolor{dialinecolor}
\node[anchor=west] at (6.157550\du,-2.545300\du){};
\definecolor{dialinecolor}{rgb}{0.000000, 0.000000, 0.000000}
\pgfsetstrokecolor{dialinecolor}
\node[anchor=west] at (6.232550\du,-7.920300\du){};
\definecolor{dialinecolor}{rgb}{0.000000, 0.000000, 0.000000}
\pgfsetstrokecolor{dialinecolor}
\node[anchor=west] at (7.032550\du,-7.920300\du){};
\definecolor{dialinecolor}{rgb}{0.000000, 0.000000, 0.000000}
\pgfsetstrokecolor{dialinecolor}
\node[anchor=west] at (14.990700\du,-5.673400\du){};
\definecolor{dialinecolor}{rgb}{0.000000, 0.000000, 0.000000}
\pgfsetstrokecolor{dialinecolor}
\node[anchor=west] at (10.990600\du,-7.798400\du){};
\definecolor{dialinecolor}{rgb}{0.000000, 0.000000, 0.000000}
\pgfsetstrokecolor{dialinecolor}
\node[anchor=west] at (13.740700\du,-0.748400\du){};
\definecolor{dialinecolor}{rgb}{0.000000, 0.000000, 0.000000}
\pgfsetstrokecolor{dialinecolor}
\node[anchor=west] at (12.040600\du,-10.362500\du){};
\definecolor{dialinecolor}{rgb}{0.000000, 0.000000, 0.000000}
\pgfsetstrokecolor{dialinecolor}
\node[anchor=west] at (8.815630\du,0.412500\du){};
\definecolor{dialinecolor}{rgb}{0.000000, 0.000000, 0.000000}
\pgfsetstrokecolor{dialinecolor}
\node[anchor=west] at (4.600000\du,-5.800000\du){$v_1$};
\definecolor{dialinecolor}{rgb}{0.000000, 0.000000, 0.000000}
\pgfsetstrokecolor{dialinecolor}
\node[anchor=west] at (4.900000\du,-1.400000\du){$v_2$};
\definecolor{dialinecolor}{rgb}{0.000000, 0.000000, 0.000000}
\pgfsetstrokecolor{dialinecolor}
\node[anchor=west] at (1.600000\du,-1.400000\du){$v_3$};
\definecolor{dialinecolor}{rgb}{0.000000, 0.000000, 0.000000}
\pgfsetstrokecolor{dialinecolor}
\node[anchor=west] at (1.600000\du,-5.800000\du){$v_4$};
\definecolor{dialinecolor}{rgb}{0.000000, 0.000000, 0.000000}
\pgfsetstrokecolor{dialinecolor}
\node[anchor=west] at (3.000000\du,-4.800000\du){};
\definecolor{dialinecolor}{rgb}{0.000000, 0.000000, 0.000000}
\pgfsetstrokecolor{dialinecolor}
\node[anchor=west] at (2.400000\du,-4.000000\du){};
\pgfsetlinewidth{0.050000\du}
\pgfsetdash{}{0pt}
\pgfsetdash{}{0pt}
\pgfsetbuttcap
{
\definecolor{dialinecolor}{rgb}{0.000000, 0.000000, 0.000000}
\pgfsetfillcolor{dialinecolor}
\definecolor{dialinecolor}{rgb}{0.000000, 0.000000, 0.000000}
\pgfsetstrokecolor{dialinecolor}
\draw (5.200000\du,-2.000000\du)--(5.200000\du,-5.400000\du);
}
\definecolor{dialinecolor}{rgb}{0.000000, 0.000000, 0.000000}
\pgfsetstrokecolor{dialinecolor}
\draw (5.200000\du,-2.000000\du)--(5.200000\du,-5.400000\du);
\pgfsetlinewidth{0.050000\du}
\pgfsetdash{}{0pt}
\pgfsetmiterjoin
\pgfsetbuttcap
\definecolor{dialinecolor}{rgb}{0.000000, 0.000000, 0.000000}
\pgfsetfillcolor{dialinecolor}
\pgfpathmoveto{\pgfpoint{5.200000\du}{-2.000000\du}}
\pgfpathcurveto{\pgfpoint{5.150000\du}{-2.000000\du}}{\pgfpoint{5.100000\du}{-2.050000\du}}{\pgfpoint{5.100000\du}{-2.100000\du}}
\pgfpathcurveto{\pgfpoint{5.100000\du}{-2.150000\du}}{\pgfpoint{5.150000\du}{-2.200000\du}}{\pgfpoint{5.200000\du}{-2.200000\du}}
\pgfpathcurveto{\pgfpoint{5.250000\du}{-2.200000\du}}{\pgfpoint{5.300000\du}{-2.150000\du}}{\pgfpoint{5.300000\du}{-2.100000\du}}
\pgfpathcurveto{\pgfpoint{5.300000\du}{-2.050000\du}}{\pgfpoint{5.250000\du}{-2.000000\du}}{\pgfpoint{5.200000\du}{-2.000000\du}}
\pgfusepath{fill}
\definecolor{dialinecolor}{rgb}{0.000000, 0.000000, 0.000000}
\pgfsetstrokecolor{dialinecolor}
\pgfpathmoveto{\pgfpoint{5.200000\du}{-2.000000\du}}
\pgfpathcurveto{\pgfpoint{5.150000\du}{-2.000000\du}}{\pgfpoint{5.100000\du}{-2.050000\du}}{\pgfpoint{5.100000\du}{-2.100000\du}}
\pgfpathcurveto{\pgfpoint{5.100000\du}{-2.150000\du}}{\pgfpoint{5.150000\du}{-2.200000\du}}{\pgfpoint{5.200000\du}{-2.200000\du}}
\pgfpathcurveto{\pgfpoint{5.250000\du}{-2.200000\du}}{\pgfpoint{5.300000\du}{-2.150000\du}}{\pgfpoint{5.300000\du}{-2.100000\du}}
\pgfpathcurveto{\pgfpoint{5.300000\du}{-2.050000\du}}{\pgfpoint{5.250000\du}{-2.000000\du}}{\pgfpoint{5.200000\du}{-2.000000\du}}
\pgfusepath{stroke}
\definecolor{dialinecolor}{rgb}{0.000000, 0.000000, 0.000000}
\pgfsetstrokecolor{dialinecolor}
\node[anchor=west] at (3.200000\du,-4.200000\du){};
\end{tikzpicture}

%% file: figure/U.tex
\ifx\du\undefined
  \newlength{\du}
\fi
\setlength{\du}{15\unitlength}
\begin{tikzpicture}
\pgftransformxscale{1.000000}
\pgftransformyscale{-1.000000}
\definecolor{dialinecolor}{rgb}{0.000000, 0.000000, 0.000000}
\pgfsetstrokecolor{dialinecolor}
\definecolor{dialinecolor}{rgb}{1.000000, 1.000000, 1.000000}
\pgfsetfillcolor{dialinecolor}
\definecolor{dialinecolor}{rgb}{0.000000, 0.000000, 0.000000}
\pgfsetstrokecolor{dialinecolor}
\node[anchor=west] at (6.000000\du,-2.000000\du){};
\pgfsetlinewidth{0.050000\du}
\pgfsetdash{}{0pt}
\pgfsetdash{}{0pt}
\pgfsetbuttcap
{
\definecolor{dialinecolor}{rgb}{0.000000, 0.000000, 0.000000}
\pgfsetfillcolor{dialinecolor}
}
\definecolor{dialinecolor}{rgb}{0.000000, 0.000000, 0.000000}
\pgfsetstrokecolor{dialinecolor}
\draw (5.300000\du,-5.400000\du)--(1.800000\du,-5.400000\du);
\pgfsetlinewidth{0.050000\du}
\pgfsetdash{}{0pt}
\pgfsetmiterjoin
\pgfsetbuttcap
\definecolor{dialinecolor}{rgb}{0.000000, 0.000000, 0.000000}
\pgfsetfillcolor{dialinecolor}
\pgfpathmoveto{\pgfpoint{5.300000\du}{-5.400000\du}}
\pgfpathcurveto{\pgfpoint{5.300000\du}{-5.350000\du}}{\pgfpoint{5.250000\du}{-5.300000\du}}{\pgfpoint{5.200000\du}{-5.300000\du}}
\pgfpathcurveto{\pgfpoint{5.150000\du}{-5.300000\du}}{\pgfpoint{5.100000\du}{-5.350000\du}}{\pgfpoint{5.100000\du}{-5.400000\du}}
\pgfpathcurveto{\pgfpoint{5.100000\du}{-5.450000\du}}{\pgfpoint{5.150000\du}{-5.500000\du}}{\pgfpoint{5.200000\du}{-5.500000\du}}
\pgfpathcurveto{\pgfpoint{5.250000\du}{-5.500000\du}}{\pgfpoint{5.300000\du}{-5.450000\du}}{\pgfpoint{5.300000\du}{-5.400000\du}}
\pgfusepath{fill}
\definecolor{dialinecolor}{rgb}{0.000000, 0.000000, 0.000000}
\pgfsetstrokecolor{dialinecolor}
\pgfpathmoveto{\pgfpoint{5.300000\du}{-5.400000\du}}
\pgfpathcurveto{\pgfpoint{5.300000\du}{-5.350000\du}}{\pgfpoint{5.250000\du}{-5.300000\du}}{\pgfpoint{5.200000\du}{-5.300000\du}}
\pgfpathcurveto{\pgfpoint{5.150000\du}{-5.300000\du}}{\pgfpoint{5.100000\du}{-5.350000\du}}{\pgfpoint{5.100000\du}{-5.400000\du}}
\pgfpathcurveto{\pgfpoint{5.100000\du}{-5.450000\du}}{\pgfpoint{5.150000\du}{-5.500000\du}}{\pgfpoint{5.200000\du}{-5.500000\du}}
\pgfpathcurveto{\pgfpoint{5.250000\du}{-5.500000\du}}{\pgfpoint{5.300000\du}{-5.450000\du}}{\pgfpoint{5.300000\du}{-5.400000\du}}
\pgfusepath{stroke}
\pgfsetlinewidth{0.050000\du}
\pgfsetdash{}{0pt}
\pgfsetmiterjoin
\pgfsetbuttcap
\definecolor{dialinecolor}{rgb}{0.000000, 0.000000, 0.000000}
\pgfsetfillcolor{dialinecolor}
\pgfpathmoveto{\pgfpoint{1.800000\du}{-5.400000\du}}
\pgfpathcurveto{\pgfpoint{1.800000\du}{-5.450000\du}}{\pgfpoint{1.850000\du}{-5.500000\du}}{\pgfpoint{1.900000\du}{-5.500000\du}}
\pgfpathcurveto{\pgfpoint{1.950000\du}{-5.500000\du}}{\pgfpoint{2.000000\du}{-5.450000\du}}{\pgfpoint{2.000000\du}{-5.400000\du}}
\pgfpathcurveto{\pgfpoint{2.000000\du}{-5.350000\du}}{\pgfpoint{1.950000\du}{-5.300000\du}}{\pgfpoint{1.900000\du}{-5.300000\du}}
\pgfpathcurveto{\pgfpoint{1.850000\du}{-5.300000\du}}{\pgfpoint{1.800000\du}{-5.350000\du}}{\pgfpoint{1.800000\du}{-5.400000\du}}
\pgfusepath{fill}
\definecolor{dialinecolor}{rgb}{0.000000, 0.000000, 0.000000}
\pgfsetstrokecolor{dialinecolor}
\pgfpathmoveto{\pgfpoint{1.800000\du}{-5.400000\du}}
\pgfpathcurveto{\pgfpoint{1.800000\du}{-5.450000\du}}{\pgfpoint{1.850000\du}{-5.500000\du}}{\pgfpoint{1.900000\du}{-5.500000\du}}
\pgfpathcurveto{\pgfpoint{1.950000\du}{-5.500000\du}}{\pgfpoint{2.000000\du}{-5.450000\du}}{\pgfpoint{2.000000\du}{-5.400000\du}}
\pgfpathcurveto{\pgfpoint{2.000000\du}{-5.350000\du}}{\pgfpoint{1.950000\du}{-5.300000\du}}{\pgfpoint{1.900000\du}{-5.300000\du}}
\pgfpathcurveto{\pgfpoint{1.850000\du}{-5.300000\du}}{\pgfpoint{1.800000\du}{-5.350000\du}}{\pgfpoint{1.800000\du}{-5.400000\du}}
\pgfusepath{stroke}
\definecolor{dialinecolor}{rgb}{0.000000, 0.000000, 0.000000}
\pgfsetstrokecolor{dialinecolor}
\node[anchor=west] at (1.800000\du,-5.800000\du){};
\pgfsetlinewidth{0.050000\du}
\pgfsetdash{}{0pt}
\pgfsetdash{}{0pt}
\pgfsetbuttcap
{
\definecolor{dialinecolor}{rgb}{0.000000, 0.000000, 0.000000}
\pgfsetfillcolor{dialinecolor}
\definecolor{dialinecolor}{rgb}{0.000000, 0.000000, 0.000000}
\pgfsetstrokecolor{dialinecolor}
\draw (1.800000\du,-2.100000\du)--(5.200000\du,-2.100000\du);
}
\definecolor{dialinecolor}{rgb}{0.000000, 0.000000, 0.000000}
\pgfsetstrokecolor{dialinecolor}
\draw (1.800000\du,-2.100000\du)--(5.200000\du,-2.100000\du);
\pgfsetlinewidth{0.050000\du}
\pgfsetdash{}{0pt}
\pgfsetmiterjoin
\pgfsetbuttcap
\definecolor{dialinecolor}{rgb}{0.000000, 0.000000, 0.000000}
\pgfsetfillcolor{dialinecolor}
\pgfpathmoveto{\pgfpoint{1.800000\du}{-2.100000\du}}
\pgfpathcurveto{\pgfpoint{1.800000\du}{-2.150000\du}}{\pgfpoint{1.850000\du}{-2.200000\du}}{\pgfpoint{1.900000\du}{-2.200000\du}}
\pgfpathcurveto{\pgfpoint{1.950000\du}{-2.200000\du}}{\pgfpoint{2.000000\du}{-2.150000\du}}{\pgfpoint{2.000000\du}{-2.100000\du}}
\pgfpathcurveto{\pgfpoint{2.000000\du}{-2.050000\du}}{\pgfpoint{1.950000\du}{-2.000000\du}}{\pgfpoint{1.900000\du}{-2.000000\du}}
\pgfpathcurveto{\pgfpoint{1.850000\du}{-2.000000\du}}{\pgfpoint{1.800000\du}{-2.050000\du}}{\pgfpoint{1.800000\du}{-2.100000\du}}
\pgfusepath{fill}
\definecolor{dialinecolor}{rgb}{0.000000, 0.000000, 0.000000}
\pgfsetstrokecolor{dialinecolor}
\pgfpathmoveto{\pgfpoint{1.800000\du}{-2.100000\du}}
\pgfpathcurveto{\pgfpoint{1.800000\du}{-2.150000\du}}{\pgfpoint{1.850000\du}{-2.200000\du}}{\pgfpoint{1.900000\du}{-2.200000\du}}
\pgfpathcurveto{\pgfpoint{1.950000\du}{-2.200000\du}}{\pgfpoint{2.000000\du}{-2.150000\du}}{\pgfpoint{2.000000\du}{-2.100000\du}}
\pgfpathcurveto{\pgfpoint{2.000000\du}{-2.050000\du}}{\pgfpoint{1.950000\du}{-2.000000\du}}{\pgfpoint{1.900000\du}{-2.000000\du}}
\pgfpathcurveto{\pgfpoint{1.850000\du}{-2.000000\du}}{\pgfpoint{1.800000\du}{-2.050000\du}}{\pgfpoint{1.800000\du}{-2.100000\du}}
\pgfusepath{stroke}
\definecolor{dialinecolor}{rgb}{0.000000, 0.000000, 0.000000}
\pgfsetstrokecolor{dialinecolor}
\node[anchor=west] at (2.315000\du,-5.382500\du){};
\definecolor{dialinecolor}{rgb}{0.000000, 0.000000, 0.000000}
\pgfsetstrokecolor{dialinecolor}
\node[anchor=west] at (1.600000\du,-1.400000\du){};
\definecolor{dialinecolor}{rgb}{0.000000, 0.000000, 0.000000}
\pgfsetstrokecolor{dialinecolor}
\node[anchor=west] at (5.400000\du,-6.200000\du){};
\definecolor{dialinecolor}{rgb}{0.000000, 0.000000, 0.000000}
\pgfsetstrokecolor{dialinecolor}
\node[anchor=west] at (6.200000\du,-6.600000\du){};
\definecolor{dialinecolor}{rgb}{0.000000, 0.000000, 0.000000}
\pgfsetstrokecolor{dialinecolor}
\node[anchor=west] at (3.000000\du,-5.000000\du){};
\definecolor{dialinecolor}{rgb}{0.000000, 0.000000, 0.000000}
\pgfsetstrokecolor{dialinecolor}
\node[anchor=west] at (6.157550\du,-2.545300\du){};
\definecolor{dialinecolor}{rgb}{0.000000, 0.000000, 0.000000}
\pgfsetstrokecolor{dialinecolor}
\node[anchor=west] at (6.232550\du,-7.920300\du){};
\definecolor{dialinecolor}{rgb}{0.000000, 0.000000, 0.000000}
\pgfsetstrokecolor{dialinecolor}
\node[anchor=west] at (7.032550\du,-7.920300\du){};
\definecolor{dialinecolor}{rgb}{0.000000, 0.000000, 0.000000}
\pgfsetstrokecolor{dialinecolor}
\node[anchor=west] at (14.990700\du,-5.673400\du){};
\definecolor{dialinecolor}{rgb}{0.000000, 0.000000, 0.000000}
\pgfsetstrokecolor{dialinecolor}
\node[anchor=west] at (10.990600\du,-7.798400\du){};
\definecolor{dialinecolor}{rgb}{0.000000, 0.000000, 0.000000}
\pgfsetstrokecolor{dialinecolor}
\node[anchor=west] at (13.740700\du,-0.748400\du){};
\definecolor{dialinecolor}{rgb}{0.000000, 0.000000, 0.000000}
\pgfsetstrokecolor{dialinecolor}
\node[anchor=west] at (12.040600\du,-10.362500\du){};
\definecolor{dialinecolor}{rgb}{0.000000, 0.000000, 0.000000}
\pgfsetstrokecolor{dialinecolor}
\node[anchor=west] at (8.815630\du,0.412500\du){};
\definecolor{dialinecolor}{rgb}{0.000000, 0.000000, 0.000000}
\pgfsetstrokecolor{dialinecolor}
\node[anchor=west] at (4.600000\du,-5.800000\du){$v_1$};
\definecolor{dialinecolor}{rgb}{0.000000, 0.000000, 0.000000}
\pgfsetstrokecolor{dialinecolor}
\node[anchor=west] at (4.900000\du,-1.400000\du){$v_2$};
\definecolor{dialinecolor}{rgb}{0.000000, 0.000000, 0.000000}
\pgfsetstrokecolor{dialinecolor}
\node[anchor=west] at (1.600000\du,-1.400000\du){$v_3$};
\definecolor{dialinecolor}{rgb}{0.000000, 0.000000, 0.000000}
\pgfsetstrokecolor{dialinecolor}
\node[anchor=west] at (1.600000\du,-5.800000\du){$v_4$};
\definecolor{dialinecolor}{rgb}{0.000000, 0.000000, 0.000000}
\pgfsetstrokecolor{dialinecolor}
\node[anchor=west] at (3.000000\du,-4.800000\du){};
\definecolor{dialinecolor}{rgb}{0.000000, 0.000000, 0.000000}
\pgfsetstrokecolor{dialinecolor}
\node[anchor=west] at (2.400000\du,-4.000000\du){};
\pgfsetlinewidth{0.050000\du}
\pgfsetdash{}{0pt}
\pgfsetdash{}{0pt}
\pgfsetbuttcap
{
\definecolor{dialinecolor}{rgb}{0.000000, 0.000000, 0.000000}
\pgfsetfillcolor{dialinecolor}
\definecolor{dialinecolor}{rgb}{0.000000, 0.000000, 0.000000}
\pgfsetstrokecolor{dialinecolor}
\draw (5.200000\du,-2.000000\du)--(5.200000\du,-5.400000\du);
}
\definecolor{dialinecolor}{rgb}{0.000000, 0.000000, 0.000000}
\pgfsetstrokecolor{dialinecolor}
\draw (5.200000\du,-2.000000\du)--(5.200000\du,-5.400000\du);
\pgfsetlinewidth{0.050000\du}
\pgfsetdash{}{0pt}
\pgfsetmiterjoin
\pgfsetbuttcap
\definecolor{dialinecolor}{rgb}{0.000000, 0.000000, 0.000000}
\pgfsetfillcolor{dialinecolor}
\pgfpathmoveto{\pgfpoint{5.200000\du}{-2.000000\du}}
\pgfpathcurveto{\pgfpoint{5.150000\du}{-2.000000\du}}{\pgfpoint{5.100000\du}{-2.050000\du}}{\pgfpoint{5.100000\du}{-2.100000\du}}
\pgfpathcurveto{\pgfpoint{5.100000\du}{-2.150000\du}}{\pgfpoint{5.150000\du}{-2.200000\du}}{\pgfpoint{5.200000\du}{-2.200000\du}}
\pgfpathcurveto{\pgfpoint{5.250000\du}{-2.200000\du}}{\pgfpoint{5.300000\du}{-2.150000\du}}{\pgfpoint{5.300000\du}{-2.100000\du}}
\pgfpathcurveto{\pgfpoint{5.300000\du}{-2.050000\du}}{\pgfpoint{5.250000\du}{-2.000000\du}}{\pgfpoint{5.200000\du}{-2.000000\du}}
\pgfusepath{fill}
\definecolor{dialinecolor}{rgb}{0.000000, 0.000000, 0.000000}
\pgfsetstrokecolor{dialinecolor}
\pgfpathmoveto{\pgfpoint{5.200000\du}{-2.000000\du}}
\pgfpathcurveto{\pgfpoint{5.150000\du}{-2.000000\du}}{\pgfpoint{5.100000\du}{-2.050000\du}}{\pgfpoint{5.100000\du}{-2.100000\du}}
\pgfpathcurveto{\pgfpoint{5.100000\du}{-2.150000\du}}{\pgfpoint{5.150000\du}{-2.200000\du}}{\pgfpoint{5.200000\du}{-2.200000\du}}
\pgfpathcurveto{\pgfpoint{5.250000\du}{-2.200000\du}}{\pgfpoint{5.300000\du}{-2.150000\du}}{\pgfpoint{5.300000\du}{-2.100000\du}}
\pgfpathcurveto{\pgfpoint{5.300000\du}{-2.050000\du}}{\pgfpoint{5.250000\du}{-2.000000\du}}{\pgfpoint{5.200000\du}{-2.000000\du}}
\pgfusepath{stroke}
\definecolor{dialinecolor}{rgb}{0.000000, 0.000000, 0.000000}
\pgfsetstrokecolor{dialinecolor}
\node[anchor=west] at (3.200000\du,-4.200000\du){};
\end{tikzpicture}

%% file: figure/F.tex
\ifx\du\undefined
  \newlength{\du}
\fi
\setlength{\du}{15\unitlength}
\begin{tikzpicture}
\pgftransformxscale{1.000000}
\pgftransformyscale{-1.000000}
\definecolor{dialinecolor}{rgb}{0.000000, 0.000000, 0.000000}
\pgfsetstrokecolor{dialinecolor}
\definecolor{dialinecolor}{rgb}{1.000000, 1.000000, 1.000000}
\pgfsetfillcolor{dialinecolor}
\definecolor{dialinecolor}{rgb}{0.000000, 0.000000, 0.000000}
\pgfsetstrokecolor{dialinecolor}
\node[anchor=west] at (6.000000\du,-2.000000\du){};
\pgfsetlinewidth{0.050000\du}
\pgfsetdash{}{0pt}
\pgfsetdash{}{0pt}
\pgfsetbuttcap
{
\definecolor{dialinecolor}{rgb}{0.000000, 0.000000, 0.000000}
\pgfsetfillcolor{dialinecolor}
}
\definecolor{dialinecolor}{rgb}{0.000000, 0.000000, 0.000000}
\pgfsetstrokecolor{dialinecolor}
\draw (5.287500\du,-5.399213\du)--(1.800000\du,-5.400000\du);
\pgfsetlinewidth{0.050000\du}
\pgfsetdash{}{0pt}
\pgfsetmiterjoin
\pgfsetbuttcap
\definecolor{dialinecolor}{rgb}{0.000000, 0.000000, 0.000000}
\pgfsetfillcolor{dialinecolor}
\pgfpathmoveto{\pgfpoint{5.287500\du}{-5.399213\du}}
\pgfpathcurveto{\pgfpoint{5.287489\du}{-5.349213\du}}{\pgfpoint{5.237477\du}{-5.299224\du}}{\pgfpoint{5.187477\du}{-5.299235\du}}
\pgfpathcurveto{\pgfpoint{5.137477\du}{-5.299246\du}}{\pgfpoint{5.087489\du}{-5.349258\du}}{\pgfpoint{5.087500\du}{-5.399258\du}}
\pgfpathcurveto{\pgfpoint{5.087511\du}{-5.449258\du}}{\pgfpoint{5.137523\du}{-5.499246\du}}{\pgfpoint{5.187523\du}{-5.499235\du}}
\pgfpathcurveto{\pgfpoint{5.237523\du}{-5.499224\du}}{\pgfpoint{5.287511\du}{-5.449212\du}}{\pgfpoint{5.287500\du}{-5.399213\du}}
\pgfusepath{fill}
\definecolor{dialinecolor}{rgb}{0.000000, 0.000000, 0.000000}
\pgfsetstrokecolor{dialinecolor}
\pgfpathmoveto{\pgfpoint{5.287500\du}{-5.399213\du}}
\pgfpathcurveto{\pgfpoint{5.287489\du}{-5.349213\du}}{\pgfpoint{5.237477\du}{-5.299224\du}}{\pgfpoint{5.187477\du}{-5.299235\du}}
\pgfpathcurveto{\pgfpoint{5.137477\du}{-5.299246\du}}{\pgfpoint{5.087489\du}{-5.349258\du}}{\pgfpoint{5.087500\du}{-5.399258\du}}
\pgfpathcurveto{\pgfpoint{5.087511\du}{-5.449258\du}}{\pgfpoint{5.137523\du}{-5.499246\du}}{\pgfpoint{5.187523\du}{-5.499235\du}}
\pgfpathcurveto{\pgfpoint{5.237523\du}{-5.499224\du}}{\pgfpoint{5.287511\du}{-5.449212\du}}{\pgfpoint{5.287500\du}{-5.399213\du}}
\pgfusepath{stroke}
\pgfsetlinewidth{0.050000\du}
\pgfsetdash{}{0pt}
\pgfsetmiterjoin
\pgfsetbuttcap
\definecolor{dialinecolor}{rgb}{0.000000, 0.000000, 0.000000}
\pgfsetfillcolor{dialinecolor}
\pgfpathmoveto{\pgfpoint{1.800000\du}{-5.400000\du}}
\pgfpathcurveto{\pgfpoint{1.800011\du}{-5.450000\du}}{\pgfpoint{1.850023\du}{-5.499989\du}}{\pgfpoint{1.900023\du}{-5.499977\du}}
\pgfpathcurveto{\pgfpoint{1.950023\du}{-5.499966\du}}{\pgfpoint{2.000011\du}{-5.449955\du}}{\pgfpoint{2.000000\du}{-5.399955\du}}
\pgfpathcurveto{\pgfpoint{1.999989\du}{-5.349955\du}}{\pgfpoint{1.949977\du}{-5.299966\du}}{\pgfpoint{1.899977\du}{-5.299977\du}}
\pgfpathcurveto{\pgfpoint{1.849977\du}{-5.299989\du}}{\pgfpoint{1.799989\du}{-5.350000\du}}{\pgfpoint{1.800000\du}{-5.400000\du}}
\pgfusepath{fill}
\definecolor{dialinecolor}{rgb}{0.000000, 0.000000, 0.000000}
\pgfsetstrokecolor{dialinecolor}
\pgfpathmoveto{\pgfpoint{1.800000\du}{-5.400000\du}}
\pgfpathcurveto{\pgfpoint{1.800011\du}{-5.450000\du}}{\pgfpoint{1.850023\du}{-5.499989\du}}{\pgfpoint{1.900023\du}{-5.499977\du}}
\pgfpathcurveto{\pgfpoint{1.950023\du}{-5.499966\du}}{\pgfpoint{2.000011\du}{-5.449955\du}}{\pgfpoint{2.000000\du}{-5.399955\du}}
\pgfpathcurveto{\pgfpoint{1.999989\du}{-5.349955\du}}{\pgfpoint{1.949977\du}{-5.299966\du}}{\pgfpoint{1.899977\du}{-5.299977\du}}
\pgfpathcurveto{\pgfpoint{1.849977\du}{-5.299989\du}}{\pgfpoint{1.799989\du}{-5.350000\du}}{\pgfpoint{1.800000\du}{-5.400000\du}}
\pgfusepath{stroke}
\definecolor{dialinecolor}{rgb}{0.000000, 0.000000, 0.000000}
\pgfsetstrokecolor{dialinecolor}
\node[anchor=west] at (1.800000\du,-5.800000\du){};
\pgfsetlinewidth{0.050000\du}
\pgfsetdash{}{0pt}
\pgfsetdash{}{0pt}
\pgfsetbuttcap
{
\definecolor{dialinecolor}{rgb}{0.000000, 0.000000, 0.000000}
\pgfsetfillcolor{dialinecolor}
}
\definecolor{dialinecolor}{rgb}{0.000000, 0.000000, 0.000000}
\pgfsetstrokecolor{dialinecolor}
\draw (1.800000\du,-2.000000\du)--(5.281250\du,-3.861712\du);
\pgfsetlinewidth{0.050000\du}
\pgfsetdash{}{0pt}
\pgfsetmiterjoin
\pgfsetbuttcap
\definecolor{dialinecolor}{rgb}{0.000000, 0.000000, 0.000000}
\pgfsetfillcolor{dialinecolor}
\pgfpathmoveto{\pgfpoint{1.800000\du}{-2.000000\du}}
\pgfpathcurveto{\pgfpoint{1.776421\du}{-2.044091\du}}{\pgfpoint{1.796933\du}{-2.111761\du}}{\pgfpoint{1.841024\du}{-2.135340\du}}
\pgfpathcurveto{\pgfpoint{1.885115\du}{-2.158920\du}}{\pgfpoint{1.952785\du}{-2.138408\du}}{\pgfpoint{1.976364\du}{-2.094317\du}}
\pgfpathcurveto{\pgfpoint{1.999943\du}{-2.050226\du}}{\pgfpoint{1.979432\du}{-1.982555\du}}{\pgfpoint{1.935340\du}{-1.958976\du}}
\pgfpathcurveto{\pgfpoint{1.891249\du}{-1.935397\du}}{\pgfpoint{1.823579\du}{-1.955909\du}}{\pgfpoint{1.800000\du}{-2.000000\du}}
\pgfusepath{fill}
\definecolor{dialinecolor}{rgb}{0.000000, 0.000000, 0.000000}
\pgfsetstrokecolor{dialinecolor}
\pgfpathmoveto{\pgfpoint{1.800000\du}{-2.000000\du}}
\pgfpathcurveto{\pgfpoint{1.776421\du}{-2.044091\du}}{\pgfpoint{1.796933\du}{-2.111761\du}}{\pgfpoint{1.841024\du}{-2.135340\du}}
\pgfpathcurveto{\pgfpoint{1.885115\du}{-2.158920\du}}{\pgfpoint{1.952785\du}{-2.138408\du}}{\pgfpoint{1.976364\du}{-2.094317\du}}
\pgfpathcurveto{\pgfpoint{1.999943\du}{-2.050226\du}}{\pgfpoint{1.979432\du}{-1.982555\du}}{\pgfpoint{1.935340\du}{-1.958976\du}}
\pgfpathcurveto{\pgfpoint{1.891249\du}{-1.935397\du}}{\pgfpoint{1.823579\du}{-1.955909\du}}{\pgfpoint{1.800000\du}{-2.000000\du}}
\pgfusepath{stroke}
\pgfsetlinewidth{0.050000\du}
\pgfsetdash{}{0pt}
\pgfsetmiterjoin
\pgfsetbuttcap
\definecolor{dialinecolor}{rgb}{0.000000, 0.000000, 0.000000}
\pgfsetfillcolor{dialinecolor}
\pgfpathmoveto{\pgfpoint{5.281250\du}{-3.861712\du}}
\pgfpathcurveto{\pgfpoint{5.304829\du}{-3.817621\du}}{\pgfpoint{5.284317\du}{-3.749951\du}}{\pgfpoint{5.240226\du}{-3.726372\du}}
\pgfpathcurveto{\pgfpoint{5.196135\du}{-3.702793\du}}{\pgfpoint{5.128465\du}{-3.723305\du}}{\pgfpoint{5.104886\du}{-3.767396\du}}
\pgfpathcurveto{\pgfpoint{5.081307\du}{-3.811487\du}}{\pgfpoint{5.101818\du}{-3.879157\du}}{\pgfpoint{5.145910\du}{-3.902736\du}}
\pgfpathcurveto{\pgfpoint{5.190001\du}{-3.926316\du}}{\pgfpoint{5.257671\du}{-3.905804\du}}{\pgfpoint{5.281250\du}{-3.861712\du}}
\pgfusepath{fill}
\definecolor{dialinecolor}{rgb}{0.000000, 0.000000, 0.000000}
\pgfsetstrokecolor{dialinecolor}
\pgfpathmoveto{\pgfpoint{5.281250\du}{-3.861712\du}}
\pgfpathcurveto{\pgfpoint{5.304829\du}{-3.817621\du}}{\pgfpoint{5.284317\du}{-3.749951\du}}{\pgfpoint{5.240226\du}{-3.726372\du}}
\pgfpathcurveto{\pgfpoint{5.196135\du}{-3.702793\du}}{\pgfpoint{5.128465\du}{-3.723305\du}}{\pgfpoint{5.104886\du}{-3.767396\du}}
\pgfpathcurveto{\pgfpoint{5.081307\du}{-3.811487\du}}{\pgfpoint{5.101818\du}{-3.879157\du}}{\pgfpoint{5.145910\du}{-3.902736\du}}
\pgfpathcurveto{\pgfpoint{5.190001\du}{-3.926316\du}}{\pgfpoint{5.257671\du}{-3.905804\du}}{\pgfpoint{5.281250\du}{-3.861712\du}}
\pgfusepath{stroke}
\definecolor{dialinecolor}{rgb}{0.000000, 0.000000, 0.000000}
\pgfsetstrokecolor{dialinecolor}
\node[anchor=west] at (2.315000\du,-5.382500\du){};
\definecolor{dialinecolor}{rgb}{0.000000, 0.000000, 0.000000}
\pgfsetstrokecolor{dialinecolor}
\node[anchor=west] at (1.600000\du,-1.400000\du){};
\definecolor{dialinecolor}{rgb}{0.000000, 0.000000, 0.000000}
\pgfsetstrokecolor{dialinecolor}
\node[anchor=west] at (5.400000\du,-6.200000\du){};
\definecolor{dialinecolor}{rgb}{0.000000, 0.000000, 0.000000}
\pgfsetstrokecolor{dialinecolor}
\node[anchor=west] at (6.200000\du,-6.600000\du){};
\definecolor{dialinecolor}{rgb}{0.000000, 0.000000, 0.000000}
\pgfsetstrokecolor{dialinecolor}
\node[anchor=west] at (3.000000\du,-5.000000\du){};
\definecolor{dialinecolor}{rgb}{0.000000, 0.000000, 0.000000}
\pgfsetstrokecolor{dialinecolor}
\node[anchor=west] at (6.157550\du,-2.545300\du){};
\definecolor{dialinecolor}{rgb}{0.000000, 0.000000, 0.000000}
\pgfsetstrokecolor{dialinecolor}
\node[anchor=west] at (6.232550\du,-7.920300\du){};
\definecolor{dialinecolor}{rgb}{0.000000, 0.000000, 0.000000}
\pgfsetstrokecolor{dialinecolor}
\node[anchor=west] at (7.032550\du,-7.920300\du){};
\definecolor{dialinecolor}{rgb}{0.000000, 0.000000, 0.000000}
\pgfsetstrokecolor{dialinecolor}
\node[anchor=west] at (14.990700\du,-5.673400\du){};
\definecolor{dialinecolor}{rgb}{0.000000, 0.000000, 0.000000}
\pgfsetstrokecolor{dialinecolor}
\node[anchor=west] at (10.990600\du,-7.798400\du){};
\definecolor{dialinecolor}{rgb}{0.000000, 0.000000, 0.000000}
\pgfsetstrokecolor{dialinecolor}
\node[anchor=west] at (13.740700\du,-0.748400\du){};
\definecolor{dialinecolor}{rgb}{0.000000, 0.000000, 0.000000}
\pgfsetstrokecolor{dialinecolor}
\node[anchor=west] at (12.040600\du,-10.362500\du){};
\definecolor{dialinecolor}{rgb}{0.000000, 0.000000, 0.000000}
\pgfsetstrokecolor{dialinecolor}
\node[anchor=west] at (8.815630\du,0.412500\du){};
\definecolor{dialinecolor}{rgb}{0.000000, 0.000000, 0.000000}
\pgfsetstrokecolor{dialinecolor}
\node[anchor=west] at (4.600000\du,-5.800000\du){$v_1$};
\definecolor{dialinecolor}{rgb}{0.000000, 0.000000, 0.000000}
\pgfsetstrokecolor{dialinecolor}
\node[anchor=west] at (4.900000\du,-1.400000\du){$v_2$};
\definecolor{dialinecolor}{rgb}{0.000000, 0.000000, 0.000000}
\pgfsetstrokecolor{dialinecolor}
\node[anchor=west] at (1.600000\du,-1.400000\du){$v_3$};
\definecolor{dialinecolor}{rgb}{0.000000, 0.000000, 0.000000}
\pgfsetstrokecolor{dialinecolor}
\node[anchor=west] at (1.600000\du,-5.800000\du){$v_4$};
\definecolor{dialinecolor}{rgb}{0.000000, 0.000000, 0.000000}
\pgfsetstrokecolor{dialinecolor}
\node[anchor=west] at (3.000000\du,-4.800000\du){};
\definecolor{dialinecolor}{rgb}{0.000000, 0.000000, 0.000000}
\pgfsetstrokecolor{dialinecolor}
\node[anchor=west] at (2.400000\du,-4.000000\du){};
\pgfsetlinewidth{0.050000\du}
\pgfsetdash{}{0pt}
\pgfsetdash{}{0pt}
\pgfsetbuttcap
{
\definecolor{dialinecolor}{rgb}{0.000000, 0.000000, 0.000000}
\pgfsetfillcolor{dialinecolor}
\definecolor{dialinecolor}{rgb}{0.000000, 0.000000, 0.000000}
\pgfsetstrokecolor{dialinecolor}
\draw (5.200000\du,-2.000000\du)--(5.200000\du,-5.400000\du);
}
\definecolor{dialinecolor}{rgb}{0.000000, 0.000000, 0.000000}
\pgfsetstrokecolor{dialinecolor}
\draw (5.200000\du,-2.000000\du)--(5.200000\du,-5.400000\du);
\pgfsetlinewidth{0.050000\du}
\pgfsetdash{}{0pt}
\pgfsetmiterjoin
\pgfsetbuttcap
\definecolor{dialinecolor}{rgb}{0.000000, 0.000000, 0.000000}
\pgfsetfillcolor{dialinecolor}
\pgfpathmoveto{\pgfpoint{5.200000\du}{-2.000000\du}}
\pgfpathcurveto{\pgfpoint{5.150000\du}{-2.000000\du}}{\pgfpoint{5.100000\du}{-2.050000\du}}{\pgfpoint{5.100000\du}{-2.100000\du}}
\pgfpathcurveto{\pgfpoint{5.100000\du}{-2.150000\du}}{\pgfpoint{5.150000\du}{-2.200000\du}}{\pgfpoint{5.200000\du}{-2.200000\du}}
\pgfpathcurveto{\pgfpoint{5.250000\du}{-2.200000\du}}{\pgfpoint{5.300000\du}{-2.150000\du}}{\pgfpoint{5.300000\du}{-2.100000\du}}
\pgfpathcurveto{\pgfpoint{5.300000\du}{-2.050000\du}}{\pgfpoint{5.250000\du}{-2.000000\du}}{\pgfpoint{5.200000\du}{-2.000000\du}}
\pgfusepath{fill}
\definecolor{dialinecolor}{rgb}{0.000000, 0.000000, 0.000000}
\pgfsetstrokecolor{dialinecolor}
\pgfpathmoveto{\pgfpoint{5.200000\du}{-2.000000\du}}
\pgfpathcurveto{\pgfpoint{5.150000\du}{-2.000000\du}}{\pgfpoint{5.100000\du}{-2.050000\du}}{\pgfpoint{5.100000\du}{-2.100000\du}}
\pgfpathcurveto{\pgfpoint{5.100000\du}{-2.150000\du}}{\pgfpoint{5.150000\du}{-2.200000\du}}{\pgfpoint{5.200000\du}{-2.200000\du}}
\pgfpathcurveto{\pgfpoint{5.250000\du}{-2.200000\du}}{\pgfpoint{5.300000\du}{-2.150000\du}}{\pgfpoint{5.300000\du}{-2.100000\du}}
\pgfpathcurveto{\pgfpoint{5.300000\du}{-2.050000\du}}{\pgfpoint{5.250000\du}{-2.000000\du}}{\pgfpoint{5.200000\du}{-2.000000\du}}
\pgfusepath{stroke}
\definecolor{dialinecolor}{rgb}{0.000000, 0.000000, 0.000000}
\pgfsetstrokecolor{dialinecolor}
\node[anchor=west] at (3.200000\du,-4.200000\du){};
\definecolor{dialinecolor}{rgb}{0.000000, 0.000000, 0.000000}
\pgfsetstrokecolor{dialinecolor}
\node[anchor=west] at (5.500000\du,-3.700000\du){$r$};
\definecolor{dialinecolor}{rgb}{0.000000, 0.000000, 0.000000}
\pgfsetstrokecolor{dialinecolor}
\node[anchor=west] at (5.800000\du,-3.600000\du){};
\end{tikzpicture}

%% file: figure/H.tex
\ifx\du\undefined
  \newlength{\du}
\fi
\setlength{\du}{15\unitlength}
\begin{tikzpicture}
\pgftransformxscale{1.000000}
\pgftransformyscale{-1.000000}
\definecolor{dialinecolor}{rgb}{0.000000, 0.000000, 0.000000}
\pgfsetstrokecolor{dialinecolor}
\definecolor{dialinecolor}{rgb}{1.000000, 1.000000, 1.000000}
\pgfsetfillcolor{dialinecolor}
\definecolor{dialinecolor}{rgb}{0.000000, 0.000000, 0.000000}
\pgfsetstrokecolor{dialinecolor}
\node[anchor=west] at (6.000000\du,-2.000000\du){};
\pgfsetlinewidth{0.050000\du}
\pgfsetdash{}{0pt}
\pgfsetdash{}{0pt}
\pgfsetbuttcap
{
\definecolor{dialinecolor}{rgb}{0.000000, 0.000000, 0.000000}
\pgfsetfillcolor{dialinecolor}
}
\definecolor{dialinecolor}{rgb}{0.000000, 0.000000, 0.000000}
\pgfsetstrokecolor{dialinecolor}
\draw (3.100000\du,-3.600000\du)--(1.800000\du,-5.400000\du);
\pgfsetlinewidth{0.050000\du}
\pgfsetdash{}{0pt}
\pgfsetmiterjoin
\pgfsetbuttcap
\definecolor{dialinecolor}{rgb}{0.000000, 0.000000, 0.000000}
\pgfsetfillcolor{dialinecolor}
\pgfpathmoveto{\pgfpoint{3.100000\du}{-3.600000\du}}
\pgfpathcurveto{\pgfpoint{3.059466\du}{-3.570725\du}}{\pgfpoint{2.989658\du}{-3.581985\du}}{\pgfpoint{2.960383\du}{-3.622519\du}}
\pgfpathcurveto{\pgfpoint{2.931108\du}{-3.663053\du}}{\pgfpoint{2.942368\du}{-3.732861\du}}{\pgfpoint{2.982902\du}{-3.762136\du}}
\pgfpathcurveto{\pgfpoint{3.023436\du}{-3.791410\du}}{\pgfpoint{3.093244\du}{-3.780151\du}}{\pgfpoint{3.122519\du}{-3.739617\du}}
\pgfpathcurveto{\pgfpoint{3.151793\du}{-3.699083\du}}{\pgfpoint{3.140534\du}{-3.629275\du}}{\pgfpoint{3.100000\du}{-3.600000\du}}
\pgfusepath{fill}
\definecolor{dialinecolor}{rgb}{0.000000, 0.000000, 0.000000}
\pgfsetstrokecolor{dialinecolor}
\pgfpathmoveto{\pgfpoint{3.100000\du}{-3.600000\du}}
\pgfpathcurveto{\pgfpoint{3.059466\du}{-3.570725\du}}{\pgfpoint{2.989658\du}{-3.581985\du}}{\pgfpoint{2.960383\du}{-3.622519\du}}
\pgfpathcurveto{\pgfpoint{2.931108\du}{-3.663053\du}}{\pgfpoint{2.942368\du}{-3.732861\du}}{\pgfpoint{2.982902\du}{-3.762136\du}}
\pgfpathcurveto{\pgfpoint{3.023436\du}{-3.791410\du}}{\pgfpoint{3.093244\du}{-3.780151\du}}{\pgfpoint{3.122519\du}{-3.739617\du}}
\pgfpathcurveto{\pgfpoint{3.151793\du}{-3.699083\du}}{\pgfpoint{3.140534\du}{-3.629275\du}}{\pgfpoint{3.100000\du}{-3.600000\du}}
\pgfusepath{stroke}
\pgfsetlinewidth{0.050000\du}
\pgfsetdash{}{0pt}
\pgfsetmiterjoin
\pgfsetbuttcap
\definecolor{dialinecolor}{rgb}{0.000000, 0.000000, 0.000000}
\pgfsetfillcolor{dialinecolor}
\pgfpathmoveto{\pgfpoint{1.800000\du}{-5.400000\du}}
\pgfpathcurveto{\pgfpoint{1.840534\du}{-5.429275\du}}{\pgfpoint{1.910342\du}{-5.418015\du}}{\pgfpoint{1.939617\du}{-5.377481\du}}
\pgfpathcurveto{\pgfpoint{1.968892\du}{-5.336947\du}}{\pgfpoint{1.957632\du}{-5.267139\du}}{\pgfpoint{1.917098\du}{-5.237864\du}}
\pgfpathcurveto{\pgfpoint{1.876564\du}{-5.208590\du}}{\pgfpoint{1.806756\du}{-5.219849\du}}{\pgfpoint{1.777481\du}{-5.260383\du}}
\pgfpathcurveto{\pgfpoint{1.748207\du}{-5.300917\du}}{\pgfpoint{1.759466\du}{-5.370725\du}}{\pgfpoint{1.800000\du}{-5.400000\du}}
\pgfusepath{fill}
\definecolor{dialinecolor}{rgb}{0.000000, 0.000000, 0.000000}
\pgfsetstrokecolor{dialinecolor}
\pgfpathmoveto{\pgfpoint{1.800000\du}{-5.400000\du}}
\pgfpathcurveto{\pgfpoint{1.840534\du}{-5.429275\du}}{\pgfpoint{1.910342\du}{-5.418015\du}}{\pgfpoint{1.939617\du}{-5.377481\du}}
\pgfpathcurveto{\pgfpoint{1.968892\du}{-5.336947\du}}{\pgfpoint{1.957632\du}{-5.267139\du}}{\pgfpoint{1.917098\du}{-5.237864\du}}
\pgfpathcurveto{\pgfpoint{1.876564\du}{-5.208590\du}}{\pgfpoint{1.806756\du}{-5.219849\du}}{\pgfpoint{1.777481\du}{-5.260383\du}}
\pgfpathcurveto{\pgfpoint{1.748207\du}{-5.300917\du}}{\pgfpoint{1.759466\du}{-5.370725\du}}{\pgfpoint{1.800000\du}{-5.400000\du}}
\pgfusepath{stroke}
\definecolor{dialinecolor}{rgb}{0.000000, 0.000000, 0.000000}
\pgfsetstrokecolor{dialinecolor}
\node[anchor=west] at (1.800000\du,-5.800000\du){};
\pgfsetlinewidth{0.050000\du}
\pgfsetdash{}{0pt}
\pgfsetdash{}{0pt}
\pgfsetbuttcap
{
\definecolor{dialinecolor}{rgb}{0.000000, 0.000000, 0.000000}
\pgfsetfillcolor{dialinecolor}
\definecolor{dialinecolor}{rgb}{0.000000, 0.000000, 0.000000}
\pgfsetstrokecolor{dialinecolor}
\draw (1.800000\du,-2.000000\du)--(3.100000\du,-3.700000\du);
}
\definecolor{dialinecolor}{rgb}{0.000000, 0.000000, 0.000000}
\pgfsetstrokecolor{dialinecolor}
\draw (1.800000\du,-2.000000\du)--(3.100000\du,-3.700000\du);
\pgfsetlinewidth{0.050000\du}
\pgfsetdash{}{0pt}
\pgfsetmiterjoin
\pgfsetbuttcap
\definecolor{dialinecolor}{rgb}{0.000000, 0.000000, 0.000000}
\pgfsetfillcolor{dialinecolor}
\pgfpathmoveto{\pgfpoint{1.800000\du}{-2.000000\du}}
\pgfpathcurveto{\pgfpoint{1.760282\du}{-2.030373\du}}{\pgfpoint{1.750937\du}{-2.100463\du}}{\pgfpoint{1.781309\du}{-2.140181\du}}
\pgfpathcurveto{\pgfpoint{1.811682\du}{-2.179899\du}}{\pgfpoint{1.881772\du}{-2.189244\du}}{\pgfpoint{1.921490\du}{-2.158872\du}}
\pgfpathcurveto{\pgfpoint{1.961208\du}{-2.128499\du}}{\pgfpoint{1.970553\du}{-2.058409\du}}{\pgfpoint{1.940181\du}{-2.018691\du}}
\pgfpathcurveto{\pgfpoint{1.909808\du}{-1.978973\du}}{\pgfpoint{1.839718\du}{-1.969627\du}}{\pgfpoint{1.800000\du}{-2.000000\du}}
\pgfusepath{fill}
\definecolor{dialinecolor}{rgb}{0.000000, 0.000000, 0.000000}
\pgfsetstrokecolor{dialinecolor}
\pgfpathmoveto{\pgfpoint{1.800000\du}{-2.000000\du}}
\pgfpathcurveto{\pgfpoint{1.760282\du}{-2.030373\du}}{\pgfpoint{1.750937\du}{-2.100463\du}}{\pgfpoint{1.781309\du}{-2.140181\du}}
\pgfpathcurveto{\pgfpoint{1.811682\du}{-2.179899\du}}{\pgfpoint{1.881772\du}{-2.189244\du}}{\pgfpoint{1.921490\du}{-2.158872\du}}
\pgfpathcurveto{\pgfpoint{1.961208\du}{-2.128499\du}}{\pgfpoint{1.970553\du}{-2.058409\du}}{\pgfpoint{1.940181\du}{-2.018691\du}}
\pgfpathcurveto{\pgfpoint{1.909808\du}{-1.978973\du}}{\pgfpoint{1.839718\du}{-1.969627\du}}{\pgfpoint{1.800000\du}{-2.000000\du}}
\pgfusepath{stroke}
\definecolor{dialinecolor}{rgb}{0.000000, 0.000000, 0.000000}
\pgfsetstrokecolor{dialinecolor}
\node[anchor=west] at (2.315000\du,-5.382500\du){};
\definecolor{dialinecolor}{rgb}{0.000000, 0.000000, 0.000000}
\pgfsetstrokecolor{dialinecolor}
\node[anchor=west] at (1.600000\du,-1.400000\du){};
\definecolor{dialinecolor}{rgb}{0.000000, 0.000000, 0.000000}
\pgfsetstrokecolor{dialinecolor}
\node[anchor=west] at (5.400000\du,-6.200000\du){};
\definecolor{dialinecolor}{rgb}{0.000000, 0.000000, 0.000000}
\pgfsetstrokecolor{dialinecolor}
\node[anchor=west] at (6.200000\du,-6.600000\du){};
\definecolor{dialinecolor}{rgb}{0.000000, 0.000000, 0.000000}
\pgfsetstrokecolor{dialinecolor}
\node[anchor=west] at (2.800000\du,-4.100000\du){$r_1$};
\definecolor{dialinecolor}{rgb}{0.000000, 0.000000, 0.000000}
\pgfsetstrokecolor{dialinecolor}
\node[anchor=west] at (3.000000\du,-5.000000\du){};
\definecolor{dialinecolor}{rgb}{0.000000, 0.000000, 0.000000}
\pgfsetstrokecolor{dialinecolor}
\node[anchor=west] at (6.157550\du,-2.545300\du){};
\definecolor{dialinecolor}{rgb}{0.000000, 0.000000, 0.000000}
\pgfsetstrokecolor{dialinecolor}
\node[anchor=west] at (6.232550\du,-7.920300\du){};
\definecolor{dialinecolor}{rgb}{0.000000, 0.000000, 0.000000}
\pgfsetstrokecolor{dialinecolor}
\node[anchor=west] at (7.032550\du,-7.920300\du){};
\definecolor{dialinecolor}{rgb}{0.000000, 0.000000, 0.000000}
\pgfsetstrokecolor{dialinecolor}
\node[anchor=west] at (14.990700\du,-5.673400\du){};
\definecolor{dialinecolor}{rgb}{0.000000, 0.000000, 0.000000}
\pgfsetstrokecolor{dialinecolor}
\node[anchor=west] at (10.990600\du,-7.798400\du){};
\definecolor{dialinecolor}{rgb}{0.000000, 0.000000, 0.000000}
\pgfsetstrokecolor{dialinecolor}
\node[anchor=west] at (13.740700\du,-0.748400\du){};
\definecolor{dialinecolor}{rgb}{0.000000, 0.000000, 0.000000}
\pgfsetstrokecolor{dialinecolor}
\node[anchor=west] at (12.040600\du,-10.362500\du){};
\definecolor{dialinecolor}{rgb}{0.000000, 0.000000, 0.000000}
\pgfsetstrokecolor{dialinecolor}
\node[anchor=west] at (8.815630\du,0.412500\du){};
\definecolor{dialinecolor}{rgb}{0.000000, 0.000000, 0.000000}
\pgfsetstrokecolor{dialinecolor}
\node[anchor=west] at (4.600000\du,-5.800000\du){$v_1$};
\definecolor{dialinecolor}{rgb}{0.000000, 0.000000, 0.000000}
\pgfsetstrokecolor{dialinecolor}
\node[anchor=west] at (4.900000\du,-1.400000\du){$v_2$};
\definecolor{dialinecolor}{rgb}{0.000000, 0.000000, 0.000000}
\pgfsetstrokecolor{dialinecolor}
\node[anchor=west] at (1.600000\du,-1.400000\du){$v_3$};
\definecolor{dialinecolor}{rgb}{0.000000, 0.000000, 0.000000}
\pgfsetstrokecolor{dialinecolor}
\node[anchor=west] at (1.600000\du,-5.800000\du){$v_4$};
\definecolor{dialinecolor}{rgb}{0.000000, 0.000000, 0.000000}
\pgfsetstrokecolor{dialinecolor}
\node[anchor=west] at (3.000000\du,-4.800000\du){};
\definecolor{dialinecolor}{rgb}{0.000000, 0.000000, 0.000000}
\pgfsetstrokecolor{dialinecolor}
\node[anchor=west] at (2.400000\du,-4.000000\du){};
\pgfsetlinewidth{0.050000\du}
\pgfsetdash{}{0pt}
\pgfsetdash{}{0pt}
\pgfsetbuttcap
{
\definecolor{dialinecolor}{rgb}{0.000000, 0.000000, 0.000000}
\pgfsetfillcolor{dialinecolor}
\definecolor{dialinecolor}{rgb}{0.000000, 0.000000, 0.000000}
\pgfsetstrokecolor{dialinecolor}
\draw (5.200000\du,-2.000000\du)--(3.900000\du,-3.700000\du);
}
\definecolor{dialinecolor}{rgb}{0.000000, 0.000000, 0.000000}
\pgfsetstrokecolor{dialinecolor}
\draw (5.200000\du,-2.000000\du)--(3.900000\du,-3.700000\du);
\pgfsetlinewidth{0.050000\du}
\pgfsetdash{}{0pt}
\pgfsetmiterjoin
\pgfsetbuttcap
\definecolor{dialinecolor}{rgb}{0.000000, 0.000000, 0.000000}
\pgfsetfillcolor{dialinecolor}
\pgfpathmoveto{\pgfpoint{5.200000\du}{-2.000000\du}}
\pgfpathcurveto{\pgfpoint{5.160282\du}{-1.969627\du}}{\pgfpoint{5.090192\du}{-1.978973\du}}{\pgfpoint{5.059819\du}{-2.018691\du}}
\pgfpathcurveto{\pgfpoint{5.029447\du}{-2.058409\du}}{\pgfpoint{5.038792\du}{-2.128499\du}}{\pgfpoint{5.078510\du}{-2.158872\du}}
\pgfpathcurveto{\pgfpoint{5.118228\du}{-2.189244\du}}{\pgfpoint{5.188318\du}{-2.179899\du}}{\pgfpoint{5.218691\du}{-2.140181\du}}
\pgfpathcurveto{\pgfpoint{5.249063\du}{-2.100463\du}}{\pgfpoint{5.239718\du}{-2.030373\du}}{\pgfpoint{5.200000\du}{-2.000000\du}}
\pgfusepath{fill}
\definecolor{dialinecolor}{rgb}{0.000000, 0.000000, 0.000000}
\pgfsetstrokecolor{dialinecolor}
\pgfpathmoveto{\pgfpoint{5.200000\du}{-2.000000\du}}
\pgfpathcurveto{\pgfpoint{5.160282\du}{-1.969627\du}}{\pgfpoint{5.090192\du}{-1.978973\du}}{\pgfpoint{5.059819\du}{-2.018691\du}}
\pgfpathcurveto{\pgfpoint{5.029447\du}{-2.058409\du}}{\pgfpoint{5.038792\du}{-2.128499\du}}{\pgfpoint{5.078510\du}{-2.158872\du}}
\pgfpathcurveto{\pgfpoint{5.118228\du}{-2.189244\du}}{\pgfpoint{5.188318\du}{-2.179899\du}}{\pgfpoint{5.218691\du}{-2.140181\du}}
\pgfpathcurveto{\pgfpoint{5.249063\du}{-2.100463\du}}{\pgfpoint{5.239718\du}{-2.030373\du}}{\pgfpoint{5.200000\du}{-2.000000\du}}
\pgfusepath{stroke}
\pgfsetlinewidth{0.050000\du}
\pgfsetdash{}{0pt}
\pgfsetdash{}{0pt}
\pgfsetbuttcap
{
\definecolor{dialinecolor}{rgb}{0.000000, 0.000000, 0.000000}
\pgfsetfillcolor{dialinecolor}
}
\definecolor{dialinecolor}{rgb}{0.000000, 0.000000, 0.000000}
\pgfsetstrokecolor{dialinecolor}
\draw (5.200000\du,-5.400000\du)--(3.900000\du,-3.600000\du);
\pgfsetlinewidth{0.050000\du}
\pgfsetdash{}{0pt}
\pgfsetmiterjoin
\pgfsetbuttcap
\definecolor{dialinecolor}{rgb}{0.000000, 0.000000, 0.000000}
\pgfsetfillcolor{dialinecolor}
\pgfpathmoveto{\pgfpoint{5.200000\du}{-5.400000\du}}
\pgfpathcurveto{\pgfpoint{5.240534\du}{-5.370725\du}}{\pgfpoint{5.251793\du}{-5.300917\du}}{\pgfpoint{5.222519\du}{-5.260383\du}}
\pgfpathcurveto{\pgfpoint{5.193244\du}{-5.219849\du}}{\pgfpoint{5.123436\du}{-5.208590\du}}{\pgfpoint{5.082902\du}{-5.237864\du}}
\pgfpathcurveto{\pgfpoint{5.042368\du}{-5.267139\du}}{\pgfpoint{5.031108\du}{-5.336947\du}}{\pgfpoint{5.060383\du}{-5.377481\du}}
\pgfpathcurveto{\pgfpoint{5.089658\du}{-5.418015\du}}{\pgfpoint{5.159466\du}{-5.429275\du}}{\pgfpoint{5.200000\du}{-5.400000\du}}
\pgfusepath{fill}
\definecolor{dialinecolor}{rgb}{0.000000, 0.000000, 0.000000}
\pgfsetstrokecolor{dialinecolor}
\pgfpathmoveto{\pgfpoint{5.200000\du}{-5.400000\du}}
\pgfpathcurveto{\pgfpoint{5.240534\du}{-5.370725\du}}{\pgfpoint{5.251793\du}{-5.300917\du}}{\pgfpoint{5.222519\du}{-5.260383\du}}
\pgfpathcurveto{\pgfpoint{5.193244\du}{-5.219849\du}}{\pgfpoint{5.123436\du}{-5.208590\du}}{\pgfpoint{5.082902\du}{-5.237864\du}}
\pgfpathcurveto{\pgfpoint{5.042368\du}{-5.267139\du}}{\pgfpoint{5.031108\du}{-5.336947\du}}{\pgfpoint{5.060383\du}{-5.377481\du}}
\pgfpathcurveto{\pgfpoint{5.089658\du}{-5.418015\du}}{\pgfpoint{5.159466\du}{-5.429275\du}}{\pgfpoint{5.200000\du}{-5.400000\du}}
\pgfusepath{stroke}
\pgfsetlinewidth{0.050000\du}
\pgfsetdash{}{0pt}
\pgfsetmiterjoin
\pgfsetbuttcap
\definecolor{dialinecolor}{rgb}{0.000000, 0.000000, 0.000000}
\pgfsetfillcolor{dialinecolor}
\pgfpathmoveto{\pgfpoint{3.900000\du}{-3.600000\du}}
\pgfpathcurveto{\pgfpoint{3.859466\du}{-3.629275\du}}{\pgfpoint{3.848207\du}{-3.699083\du}}{\pgfpoint{3.877481\du}{-3.739617\du}}
\pgfpathcurveto{\pgfpoint{3.906756\du}{-3.780151\du}}{\pgfpoint{3.976564\du}{-3.791410\du}}{\pgfpoint{4.017098\du}{-3.762136\du}}
\pgfpathcurveto{\pgfpoint{4.057632\du}{-3.732861\du}}{\pgfpoint{4.068892\du}{-3.663053\du}}{\pgfpoint{4.039617\du}{-3.622519\du}}
\pgfpathcurveto{\pgfpoint{4.010342\du}{-3.581985\du}}{\pgfpoint{3.940534\du}{-3.570725\du}}{\pgfpoint{3.900000\du}{-3.600000\du}}
\pgfusepath{fill}
\definecolor{dialinecolor}{rgb}{0.000000, 0.000000, 0.000000}
\pgfsetstrokecolor{dialinecolor}
\pgfpathmoveto{\pgfpoint{3.900000\du}{-3.600000\du}}
\pgfpathcurveto{\pgfpoint{3.859466\du}{-3.629275\du}}{\pgfpoint{3.848207\du}{-3.699083\du}}{\pgfpoint{3.877481\du}{-3.739617\du}}
\pgfpathcurveto{\pgfpoint{3.906756\du}{-3.780151\du}}{\pgfpoint{3.976564\du}{-3.791410\du}}{\pgfpoint{4.017098\du}{-3.762136\du}}
\pgfpathcurveto{\pgfpoint{4.057632\du}{-3.732861\du}}{\pgfpoint{4.068892\du}{-3.663053\du}}{\pgfpoint{4.039617\du}{-3.622519\du}}
\pgfpathcurveto{\pgfpoint{4.010342\du}{-3.581985\du}}{\pgfpoint{3.940534\du}{-3.570725\du}}{\pgfpoint{3.900000\du}{-3.600000\du}}
\pgfusepath{stroke}
\pgfsetlinewidth{0.050000\du}
\pgfsetdash{}{0pt}
\pgfsetdash{}{0pt}
\pgfsetbuttcap
{
\definecolor{dialinecolor}{rgb}{0.000000, 0.000000, 0.000000}
\pgfsetfillcolor{dialinecolor}
\definecolor{dialinecolor}{rgb}{0.000000, 0.000000, 0.000000}
\pgfsetstrokecolor{dialinecolor}
\draw (3.100000\du,-3.700000\du)--(4.000000\du,-3.700000\du);
}
\definecolor{dialinecolor}{rgb}{0.000000, 0.000000, 0.000000}
\pgfsetstrokecolor{dialinecolor}
\node[anchor=west] at (3.700000\du,-4.100000\du){$r_2$};
\definecolor{dialinecolor}{rgb}{0.000000, 0.000000, 0.000000}
\pgfsetstrokecolor{dialinecolor}
\node[anchor=west] at (3.200000\du,-4.200000\du){};
\end{tikzpicture}

%% file: cube.tex
In this section, the main goal in our context is the case of a cube $Q$,
which consists of six squares and the co-curvature of each vertex is $270^\circ$.
We however show a stronger result for general orthogonal boxes.

\begin{theorem}
\label{th:cube}
Let $a,b,c$ be any positive natural numbers.
Then the folding problem of a box $Q$ of size $a\times b\times c$ from a given simple polygon $P$ 
can be solved in $O(D^2 n^3)$ time, 
where $n$ is the number of vertices in $P$, and $D$ is the diameter of $P$.
In our context, the running time can be represented as $O(L (L+n) n^2)$ time, where $L$ is the perimeter of $P$.
\end{theorem}

In \cite{MHU2020}, the authors investigated the same folding problem for a box $Q$, 
and they gave a pseudo-polynomial time algorithm that runs in $O(D^{11}n^2(D^5+\log n))$ time.
It checks all combinations of $a,b$, and $c$ in $O(D^4)$ time.
Therefore, the algorithm in \cite{MHU2020} gives us $O(D^{7}n^2(D^5+\log n))$ time algorithm
when $a$, $b$, and $c$ are explicitly given.
Our algorithm drastically improves its running time to $O(D^2 n^3)$.

\begin{cor}
\label{cor:cube}
The folding problem of a cube $Q$ from a given simple polygon $P$ can be solved in $O(D^2 n^3)$ time, 
where $n$ is the number of vertices in $P$, and $D$ is the diameter of $P$.
In our context, the running time can be represented as $O(L(L+n) n^2)$ time, where $L$ is the perimeter of $P$.
\end{cor}

In our algorithm, the fact $a=b=c$ is not useful to improve the running time.
Therefore, hereafter, we assume that $Q$ is a box of size $a\times b\times c$
(the size is explicitly given as a part of the input), 
$P$ is a simple polygon with diameter $D$ and perimeter $L$,
and let $\ell_{\max}$ be the length of the longest edge of $P$.



\subsection{Algorithm}

\begin{algorithm}[th]
\caption{Outline of the algorithm in \cite{MHU2020} (when $a,b,c$ are given).}
\label{alg:1}
 \SetKwInOut{Input}{Input}
 \SetKwInOut{Output}{Output}
 \Input{A polygon $P=(p_0,p_1,\ldots,p_{n-1},p_0)$}
 \Output{A set $S=\{Q_0,Q_1,\ldots,Q_k\}$ of boxes of size $a\times b\times c$ that can be folded from $P$}
   \For{$i\leftarrow 0$ \KwTo $n-1$}{
    \If{curvature at $p_i$ is $270^\circ$}{
     Find a position of $Q$ on $P$ such that $p_i$ is a vertex of $Q$,
      and all vertices of $Q$ are not inside of $P$ by stamping\;
     Check if $P$ can fold $Q$ by gluing, and output if it can\;
     Find the next position by rotating $Q$ on $P$ at $p_i$ and repeat the check\;
    }
   }
\end{algorithm}

The algorithm in \cite{MHU2020} can be outlined as Algorithm \ref{alg:1}.
That is, the algorithm checks all possible points $p_i$ if it makes $270^\circ$.
The key issue is to find the next position of $Q$ on $P$ by rotating at one point $p_i$ of curvature $270^\circ$.
In \cite{MHU2020}, the upper bound of the number of possible rotations is the main factor of the running time.

By Lemma \ref{lem:non180}, if $Q$ can be folded from $P$, 
there are at least \emph{two} vertices on $P$ that fold to the vertices of $Q$. 
Our main idea is to use the second vertex as shown in Algorithm \ref{alg:2}.

\begin{algorithm}[th]
\caption{Outline of our algorithm.}
\label{alg:2}
 \SetKwInOut{Input}{Input}
 \SetKwInOut{Output}{Output}
 \Input{A polygon $P=(p_0,p_1,\ldots,p_{n-1},p_0)$}
 \Output{A set $S=\{Q_0,Q_1,\ldots,Q_k\}$ of boxes of size $a\times b\times c$ that can be folded from $P$}
   \For{$i\leftarrow 0$ \KwTo $n-1$}{
    \For{$j\leftarrow i+1$ \KwTo $n-1$}{
    \If{curvature at $p_i$ and $p_j$ are $270^\circ$}{
      Let $\ell:=(x(p_i)-x(p_j))^2+(y(p_i)-y(p_j))^2$\;
      \For{$X\leftarrow 0$ \KwTo $\floor{\sqrt{\ell}}$}{
        Let $Y:=\sqrt{\ell -X^2}$\;
        \If{$Y$ is an integer}{
          Set $x$-axis and $y$-axis by $X$ and $Y$\;
          Put $Q$ on $P$ such that $p_i$ is a vertex of $Q$ along the axes\;
          Check if all vertices of $Q$ are not inside of $P$ by stamping\;
          Check if $P$ can fold to $Q$ by gluing, and output if it can\;
        }
      }
   }
 }
}
\end{algorithm}

We can use the same arguments in \cite{MHU2020} to show the correctness of our algorithm.
Essentially, we check all the possible cases, which is guaranteed by Lemma \ref{lem:non180}.
That is, if $P$ can fold to $Q$, 
there are two vertices $p_i$ and $p_j$ of curvature $270^\circ$ on the boundary of $P$.
Our algorithm checks all combinations of them in $O(n^2)$ time.
When we fix the right pair of $p_i$ and $p_j$, 
the edge $p_i p_j$ should be the oblique side of some right triangle $\angle p_i q p_j$
such that $x(q)=x(p_j)$ and $y(q)=y(p_i)$.
Since we can assume an integral grid for the box $Q$, we have 
$\msize{p_i q}=X$ and $\msize{p_j q}=Y$ for some positive integers $X$ and $Y$, 
which give the $y$-axis and $x$-axis for box stamping.
Then we have a lemma:
\begin{lemma}
\label{lem:XY}
Assume that we fix the right pair of $p_i$ and $p_j$.
Then the number of the corresponding pairs of $x$-axis and $y$-axis is $O(D)$.
Since $2D\le L$, the number is also $O(L)$.
\end{lemma}
\begin{proof}
As discussed above, $X$ is a positive integer with $0\le X\le \msize{p_i p_j}$ for the given $p_i$ and $p_j$.
Since $\msize{p_i p_j}\le D$, we have the lemma.
\qed\end{proof}

Once we fix the $x$-axis and $y$-axis with letting $p_i=(0,0)$ and $p_j=(X,Y)$,
we now put $Q$ at $p_i$. Then, since the curvature at $p_i$ is $270^\circ$,
at least one of four orthants is included in $P$. More precisely, for sufficiently small $\epsilon>0$, 
all points $(x,y)$ with $x^2+y^2\le \epsilon$ are in $P$ for at least one of the following four conditions:
(1)  $x\ge 0$ and $y\ge 0$, (2)  $x\ge 0$ and $y\le 0$, 
(3)  $x\le 0$ and $y\ge 0$, and (4)  $x\le 0$ and $y\le 0$.
Without loss of generality, we assume that we have case (1).
Then we can put the box $Q$ at $p_i$ on this orthant.
There are six possible cases. 
That is, the first rectangle occupies $(x,y)$ such that 
(1) $0\le x\le a$ and $0\le y\le b$,
(2) $0\le x\le a$ and $0\le y\le c$,
(3) $0\le x\le b$ and $0\le y\le c$,
(4) $0\le x\le b$ and $0\le y\le a$,
(5) $0\le x\le c$ and $0\le y\le a$, or 
(6) $0\le x\le c$ and $0\le y\le b$.
We check these 6 cases one by one.
From each of these six initial positions, we start stamping by rolling the box $Q$ on $P$.
As discussed in Sections \ref{sec:stamping} and \ref{sec:glue-check},
we can check all possible ways of folding of $Q$ from $P$ in this way.
Now we give a crucial lemma for the estimation of time complexity:
\begin{lemma}
\label{lem:rolling}
The number of rollings of $Q$ on $P$ in each loop is $O(Dn)$. The number is also $O(L+n)$.
\end{lemma}


\begin{proof}
The stamping is done in the DFS manner.
Then, by Observation \ref{obs:out}, any vertex of $Q$ cannot be inside of $P$ 
during the rolling (otherwise, $P$ is not a net of $Q$ at that point, and we refuse).
Therefore, two rollings contribute at least length $a$ to consume the length of an edge of $P$
in the proof of Lemma \ref{lem:traverse}. 
That is, each edge of $P$ is consumed by $O(D)$ rollings of $Q$.
Since $P$ has $n$ edges, the total number of rolling of $Q$ to consume all the edges of $P$ is $O(Dn)$.
When all edges of $P$ are consumed, or covered by $Q$, the stamping is finished. Thus we have the lemma.
In the same argument, each edge $e$ requires $O(\msize{e}/a)$ rolling.
Thus the number is also $O(L+n)$ in total.
\qed\end{proof}
We now turn to the proof of the main theorem in this section:

\begin{proof}(of Theorem \ref{th:cube})
We owe the previous work in \cite{MHU2020} the proof of the correctness of our algorithm.
Essentially, both algorithms check all feasible combinations that contain correct answers if they exist.
Therefore, our remaining task is to show the time complexity of our algorithm.

As shown already, the first loop for $p_i$ and $p_j$ takes $O(n^2)$ time.
For each pair of $p_i$ and $p_j$, 
the number of possible combinations of $X$ and $Y$ is $O(D)$ by Lemma \ref{lem:XY}.
The number of ways of putting the box $Q$ is six, and for each of them, 
the number of rolling takes $O(Dn)$ time by Lemma \ref{lem:rolling}.
By Theorem \ref{th:stamping}, each stamping can be done in $O(Dn+n)$ time.
By Theorem \ref{th:glue}, the gluing check takes $O(n')$ time, where $n'$ is the number of the vertices of $P'$.
In the box case, we can see that $n'=O(D+n)$.
Therefore, total running time of our algorithm is 
$O(n^2\times D\times(Dn+n+D+n)\log n))=O(D^2 n^3)$.

When we use $L$ instead of $D$, the number of rolling is $O(L+n)$ by Lemma \ref{lem:rolling}
and hence this step requires $O(L+n)$ time to traverse $\partial P$.
Therefore, the running time becomes $O(L(L+n) n^2)$.
\qed\end{proof}

%% file: delta.tex
The common outline of our algorithms for a regular dodecahedron and a non-concave deltahedron 
is given in Algorithm \ref{alg:common}.
That is, the algorithm checks all possible combinations of pairs $\{p_{i},p_{i'}\}$ and $q_j$.
By Lemma \ref{lem:non180}, if $Q$ can be folded from $P$, there are at least two vertices
$p_{i},p_{i'}$ of $P$ that correspond to $q_{j},q_{j'}$ of $Q$ for some $q_{j'}$, 
respectively, with $i\neq i'$ and $j\neq j'$.
Hereafter, we assume that the vertex $p_0$ of $P$ corresponds to the vertex $q_0$ of $Q$,
and $p_{i}$ of $P$ corresponds to the vertex $q_j$ of $Q$, respectively, without loss of generality.

The key point is how to decide the relative orientation of $Q$ and $P$, 
which has an influence of time complexity of the algorithm.
Intuitively, for this issue, we also try all possible cases.
The time complexity (or the number of trials) is different depending on the shape of $Q$.
For the remainder of Section \ref{sec:delta}, we assume that the orientation of $Q$ relative to $P$ is fixed.

\subsection{Regular Dodecahedron}
\label{sec:dodeca}

In this section, we assume that $Q$ is a regular dodecahedron and the length of each edge is 1.
Since the area of a pentagon of unit edge is $\frac{\sqrt{25+10\sqrt{5}}}{4}$,
we assume the area of $P$ is $12\times \frac{\sqrt{25+10\sqrt{5}}}{4}=3\sqrt{25+10\sqrt{5}}$
 without loss of generality.
Since we know $Q$, the input of this problem is just 
a polygon $P=(p_0,p_1,\ldots,p_{n-1},p_0)$ of area $3\sqrt{25+10\sqrt{5}}$, and 
we will decide if $P$ can be folded to a unit-size regular dodecahedron $Q$.
By Lemma \ref{lem:non180}, 
we also know that two vertices $p_0$ and $p_i$ of $P$ correspond to two distinct vertices, 
say $q_0$ and $q_j$, of $Q$.
Then the main theorem in this section is as below.

\begin{theorem}
\label{th:penta}
Let $P$ be a simple polygon with $n$ vertices. We denote by $L$ the perimeter of $P$.
Then the folding problem of a regular dodecahedron from $P$ can be solved in $O((L+n)^4 n^2)$ time.
\end{theorem}

\begin{figure}[h]
\centering
\includegraphics[width=0.7\linewidth]{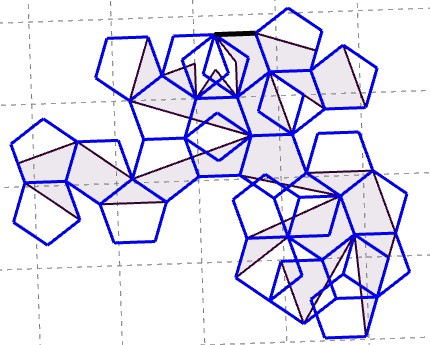}
\caption{An example of overlapping stamping. 
Some pentagons are overlapping by stamping of $Q$ along a feasible net $P$.}
\label{fig:overlap}
\end{figure}

\subsubsection{Stamping}

By assumption, $Q$ can reach from $p_0$ to $p_i$ on $P$ by stamping $Q$ such that
$p_0$ and $p_i$ are corresponding to two different vertices of $Q$.
By rotation of $P$, we have a sequence of regular pentagonal faces $(\hat{f}_0,\hat{f}_1,\ldots,\hat{f}_k)$ such that 
(1) $\hat{f}_0$ contains the edge joining points $p_0=(0,0)$ and $(1,0)$ as its base edge,
(2) $p_i=(x_i,y_i)$ is a vertex of $\hat{f}_k$, and 
(3) two consecutive pentagons $\hat{f}_{j'}$ and $\hat{f}_{j'+1}$ share an edge for each $j'$ with $0\le j'<k$.
Intuitively, the sequence gives us the shortest way of stamping of $Q$ joining $p_0$ and $p_i$ on $P$.
In other words, if we put $Q$ on $P$ with a proper relative angle, 
$Q$ can be unfolded to $P$, and we can reach from $p_0$ to $p_i$ by traversing the edges of these regular pentagons.
We note that two consecutive pentagons do not overlap (without their shared edge), 
but nonconsecutive pentagons can overlap by stamping (see \figurename~\ref{fig:overlap}).

\begin{figure}[h]
\centering
\includegraphics[width=0.3\linewidth]{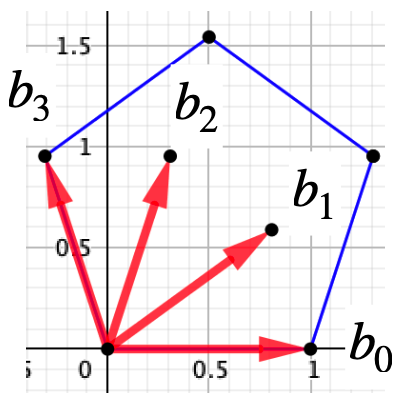}
\caption{Four unit vectors for a unit pentagon.}
\label{fig:vectors}
\end{figure}

\begin{figure}[h]
\centering
\includegraphics[width=0.7\linewidth]{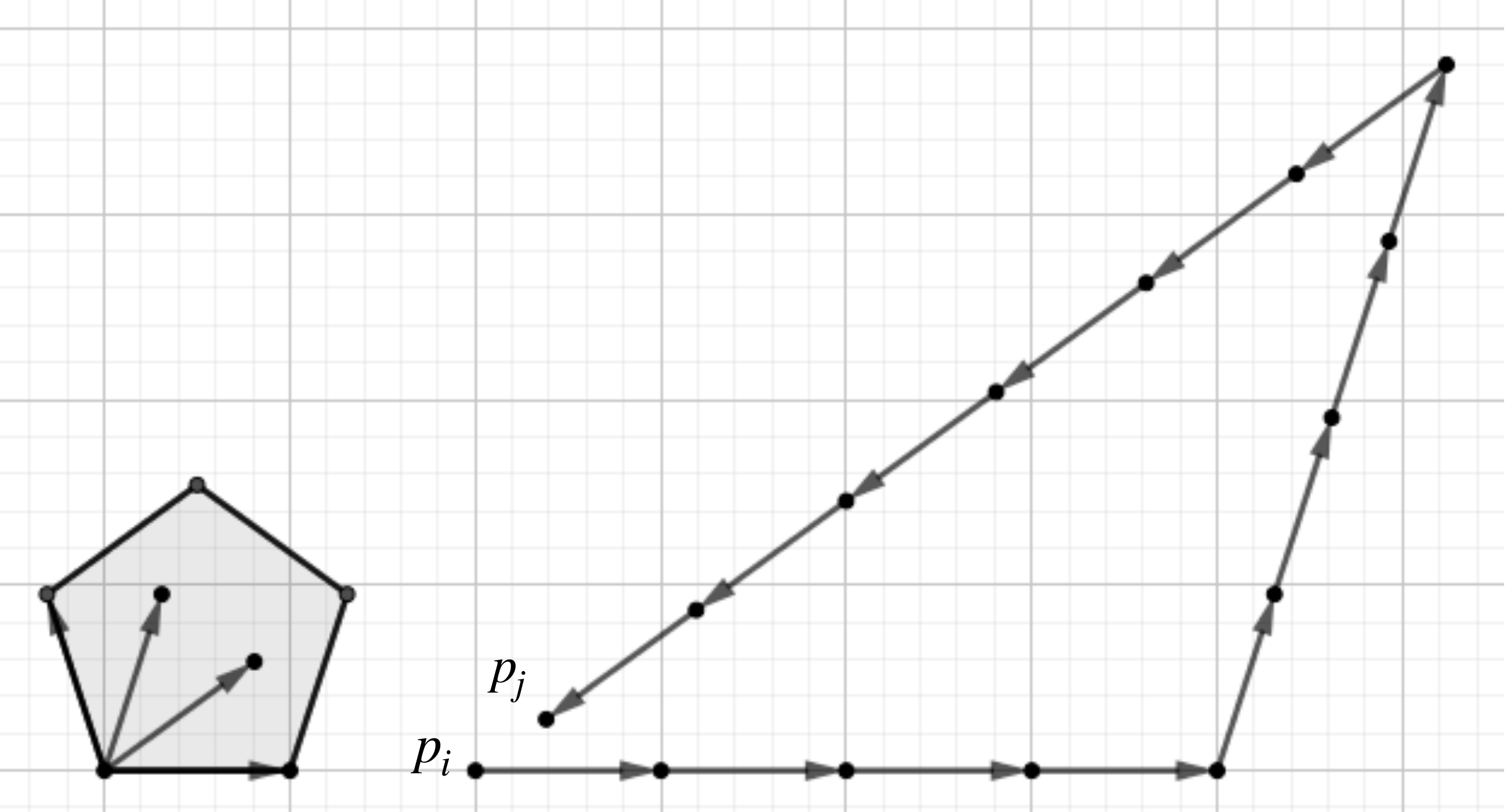}
\caption{Two points $p_i$ and $p_j$ are close, however, they can be spanned by many vectors.}
\label{fig:penta}
\end{figure}

When we consider each edge of the pentagons as a unit vector,
this traverse can be represented by a linear combination of the following four vectors (\figurename~\ref{fig:vectors}):
$\vec{b_0}=(0,1)$, $\vec{b_1}=(\cos\frac{\pi}{5},\sin\frac{\pi}{5})$,
$\vec{b_2}=(\cos\frac{2\pi}{5},\sin\frac{2\pi}{5})$, 
and $\vec{b_3}=(\cos\frac{3\pi}{5},\sin\frac{3\pi}{5})$.
Note that $(\cos\frac{4\pi}{5},\sin\frac{4\pi}{5})=
-\vec{b_0}+\vec{b_1}-\vec{b_2}+\vec{b_3}$.
Thus, $Q$ can be folded from $P$ only if we have four integers $B_0,B_1,B_2,B_3$ such that
\[
\vec{p_i}-\vec{p_0}=B_0\vec{b_0}+B_1\vec{b_1}+B_2\vec{b_2}+B_3\vec{b_3},
\]
and hence
\[
\msize{\vec{p_i}-\vec{p_0}}=\msize{B_0\vec{b_0}+B_1\vec{b_1}+B_2\vec{b_2}+B_3\vec{b_3}}.
\]

We here note that $\msize{B_k}$ does not necessarily small even if
$\msize{\vec{p_i}-\vec{p_0}}$ is small (\figurename~\ref{fig:penta}).
When we consider a grid (as a triangular grid in Lemma \ref{lem:lattice} and
a square grid in Lemma \ref{lem:XY}) which is spanned by two unit vectors,
we can say that $\msize{B_k}\le \msize{\vec{p_i}-\vec{p_0}}$.
In the case of a regular pentagon, we have no such a grid,
and hence we bound the number $\msize{B_k}$ by the number of the stamping.
That is, by Theorem \ref{th:stamping}, we have $\msize{B_k}=O(L+n)$.
Thus, it is enough to check $O((L+n)^4)$ combinations of four integers $B_0,B_1,B_2,B_3$.
For each possible integers $B_0,B_1,B_2,B_3$, we can compute $p_i=(x_i,y_i)$ by rotation of $P$.
After putting $P$ on the proper place so that $p_0=(0,0)$ and $p_i=(x_i,y_i)$, 
we perform the stamping of $Q$ on $P$ and obtain the partition of $P$.
We here note that we use the commutative law of vectors.
Thus the first relative position of $Q$ is one of four positions along 
$\vec{b_0}, \vec{b_1}, \vec{b_2}, \vec{b_3}$.
For each position, we perform the second phase for checking gluing.

\subsubsection{Gluing check}

By stamping of $Q$ on $P$, $P$ is partitioned into regions $\calF=\{F_0,F_1,\ldots,F_{h-1}\}$.
More precisely, $F_0$ is the intersection of $P$ and $Q$ on an initial position such that $p_0=q_0=(0,0)$.
Since it is a valid stamping, there are no vertices of $Q$ inside of $F_0$.
As discussed in Lemma \ref{lem:tree}, the contact graph $T=(P,Q,F_0)$ is a tree.
For notational convenience, we consider $F_0$ is the root of $T$,
and the elements in $\calF$ are numbered from $F_0$ in the way of the BFS manner.

First, we glue $F_0$ on $Q$ so that the corresponding vertex $p_0$ on $P$ (or $F_0$) 
comes to a vertex $q_0$ of $Q$. Then the gluing process is done on $Q$ from $F_0$ in the BFS manner.
As shown in Theorem \ref{th:stamping}, the stamping can be done in $O(\msize{\calF}n)$ time.
We have the following upper bound of $\msize{\calF}$:
\begin{theorem}
\label{th:upper-penta}
$\msize{\calF}=O(L+n)$.
\end{theorem}

\begin{proof}
The number of stampings of $Q$ on $P$ is given by the total number of visits of each region $F_i$.
On the other hand, $\msize{\calF}$ is the number of $F_i$s.
Thus, precisely, $(\msize{\calF}-1)$ is the number of the first visiting each $F_i$ by $Q$ except $F_0$.
The stamping of $Q$ is done along the BFS tree. 
Therefore, since each edge of the BFS tree is traversed twice, 
the number of stampings made by $Q$ is $2(\msize{\calF}-1)$. 
Thus $\msize{\calF}$ is proportional to the number of stampings.

\begin{figure}[h]
\centering
\includegraphics[width=0.6\linewidth]{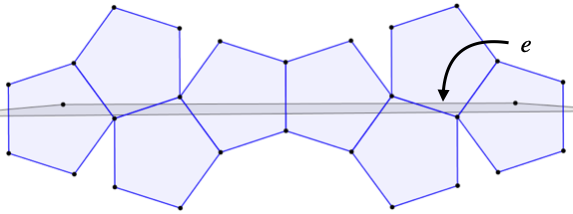}
\caption{An edge $e$ can be covered by $O(\msize{e})$ pentagons
since each angle of a pentagon is $108^\circ$.}
\label{fig:pentagons}
\end{figure}

Let $e$ be an edge of $P$. 
By stamping of $Q$ along the edge $e$, since each pentagonal face of $Q$ has the unit size,
the number of pentagons $\hat{f}_i$ to cover $e$ is $O(\msize{e})$ by Lemma \ref{lem:traverse} with \figurename~\ref{fig:pentagons}.
Thus, the number of pentagonal faces of $Q$ as stamps to cover all the edges $e$ of $P$
is $O(L+n)$ in total. Therefore, we obtain $\msize{\calF}=O(L+n)$.
\qed
\end{proof}

\subsubsection{Time complexity}

Now we consider the time complexity of our algorithm for a regular dodecahedron.
For a given polygon $P=(p_0,p_1,\ldots,p_{n-1},p_0)$,
the algorithm first generates all possible combinations of $(p_i,p_{i'})$, which produce $O(n^2)$ cases. 
We here note that we essentially have one way of choosing $q_0$ by the symmetry of $Q$.
For this $q_0$, we have a constant number (precisely, it is 7) of cases of $q_j$.
Thus we do not need to consider this constant factor for a regular dodecahedron.
For each pair $(p_i,p_{i'})$, we construct 
a vector $\vec{p_{i'}}-\vec{p_i}=B_0\vec{b_0}+B_1\vec{b_1}+B_2\vec{b_2}+B_3\vec{b_3}$
by checking all possible values of $B_0,B_1,B_2,B_3$ with $B_k=O(L+n)$ for $k=0,1,2,3$.
This step generates $O((L+n)^4)$ combinations if we check all combinations in a straightforward way.
However, when $B_0,B_1,B_2$ are fixed, 
since $\msize{\vec{p_i}-\vec{p_0}}=\msize{B_0\vec{b_0}+B_1\vec{b_1}+B_2\vec{b_2}+B_3\vec{b_3}}$ is given,
we have two possible values depending on $B_3\ge 0$ or $B_3<0$,
and they can be computed in a constant time.
Thus it is enough to check $O((L+n)^3)$ combinations by computing two candidates of $B_3$ from $B_0,B_1,B_2$.
For each case, the algorithm performs stamping of $Q$. 
During the stamping, we check if each vertex of a face of $Q$ is inside $P$ or not.
It is done along the traverse of the tree in BFS manner, and hence it can be done in $O(n)$ time in total.
Thus the running time of stamping is $O(\msize{\calF}+n)$, where $(\msize{\calF}-1)$ is the number of stampings.
By Theorem \ref{th:upper-penta}, we have $\msize{\calF}=O(L+n)$.
After the (valid) stamping, we obtain a partition $\calF=\{F_0,F_1,\ldots,F_{\msize{\calF}}\}$ of $P$.
During the stamping, as discussed in Section \ref{sec:glue-check},
we have already constructed a refined polygon
$P'=(p'_0,p'_1,\ldots,p'_{n'-1},p'_{n'}=p'_0)$ with the set $S$ of gluing points in $P'$.
By Theorem \ref{th:glue}, checking the gluing of elements in $\calF$ onto $Q$ takes $O(n')$ time.
Since the tree $T=(P,Q,F_0)$ has $\msize{\calF}$ vertices and $(\msize{\calF}-1)$ edges,
we have $\sum_{i=0}^{\msize{\calF}-1}\msize{F_i}=O(\msize{\calF}+n)$, which is $O(L+n)$ by Theorem \ref{th:upper-penta}.
Therefore, in total, the algorithm runs in $O((L+n)^4 n^2)$ time. It completes the proof of Theorem \ref{th:penta}.

\subsection{Non-concave Deltahedron} 

In this section, we assume that $Q$ is 
a non-concave deltahedron such that it has at least two vertices of curvature not equal to $180^\circ$.
We assume that each face of $Q$ consists of some unit regular triangles;
each unit triangle has three edges of unit length 1 and area $\frac{\sqrt{3}}{4}$.
Let $t$ be the total number of unit triangles on the surface of $Q$.
That is, the surface area of $Q$ is $\frac{\sqrt{3}}{4}t$.
Let $\{q_0,q_1,\ldots,q_{m-1}\}$ be the set of vertices of $Q$.
We assume that 
(1) the set of faces $\{f_0,f_1,\ldots,f_{l-1}\}$ of $Q$ is given, where $l$ is the number of faces of $Q$,
(2) each vertex $q_j$ has its coordinate $(x_j,y_j,z_j)$, and
(3) each face has its vertices in clockwise order.
The basic idea of the algorithm is the same as in Section \ref{sec:dodeca};
we here consider the differences.

\begin{theorem}
\label{th:tri}
Let $P$ be a simple polygon with $n$ vertices of perimeter $L$.
Let $Q$ be a non-concave deltahedron\footnote{For simplicity, we call ``non-concave deltahedron''
a polyhedron that is either a convex deltahedron or a non-strictly-convex deltahedron,
and we assume that it is not a regular tetrahedron.} with $m$ vertices.
Then the folding problem of $Q$ from $P$ can be solved in $O(L(L+n) m n^2)$ time.
\end{theorem}

We still have the property that we can reach from $p_0$ to $p_i$ on $P$ by stamping $Q$ on it.
However, now we have $O(m n^2)$ combinations for pairs of pair $(p_i,p_{i'})$ and $q_j$.
Hereafter, we assume that the vertex $p_{i}$ of $P$ forms a vertex $q_{j}$ on $Q$ and 
the vertex $p_{i'}$ forms some vertex $q_{j'}$ on $Q$.
In the same argument in the tetrahedron case, for two vectors 
$\vec{b_0}=(1,0)$ and 
$\vec{b'_1}=(\cos\frac{\pi}{3},\sin\frac{\pi}{3})=(\frac{1}{2},\frac{\sqrt{3}}{2})$,
$Q$ can be folded from $P$ only if we have two integers $B'_0$ and $B'_1$ such that
\[
\msize{\vec{p_{i'}}-\vec{p_i}}=\msize{B'_0\vec{b_0}+B'_1\vec{b'_1}}.
\]
We have that $\msize{B'_0}$ and $\msize{B'_1}$ are at most $L$ by Lemma \ref{lem:lattice}, 
and hence we have $O(L^2)$ combinations to be checked. 
However, once we fix $B'_0$, then $B'_1$ has two possible values. 
Thus this step requires $O(L)$ combinations.
Each partition of $P$ by stamping of $Q$ takes $O(L+n)$ time by the same argument 
in the case of a dodecahedron.

For gluing, almost all arguments are the same as the pentagonal case 
since they do not use the fact that the shape of a face is a pentagon.
The only difference is that we stamp all (possibly different) faces of $Q$; 
this fact gives us an additional lower bound $l$ of the number of stampings.
Therefore, the time complexity of this algorithm for 
a non-concave deltahedron is $O(L(L+l+n)m n^2)$.
Here, by the Euler characteristic, we have $l=2+e-m$, where $e$ is the number of edges of $Q$.
When we consider $Q$ as a graph, it is a planar graph, which implies that $e=O(m)$, or $l=O(m)$.
By Theorem \ref{th:gauss}, $Q$ has at most four vertices of curvature $180^\circ$.
Thus we have $m=O(n)$. 
Therefore, the time complexity of this algorithm is $O(L(L+n)m n^2)$,
which completes the proof of Theorem \ref{th:tri}.

Since $m$ is a constant for each of a regular octahedron and a regular icosahedron,
we have the time complexities in Table \ref{tab:alg}.